\documentclass[12pt]{article}
\usepackage[margin=1in]{geometry}
\usepackage{float}

\usepackage{cancel}
\usepackage{amsmath,amssymb,bm,url,mathrsfs,mathtools}
\usepackage{algorithmic}
\usepackage[ruled,vlined]{algorithm2e}
\usepackage{setspace}
\usepackage[ruled,vlined]{algorithm2e}
\SetAlFnt{\small\setstretch{1.2}}

\usepackage{amsthm}
\usepackage{hyperref}
\usepackage{makecell}
\usepackage{tablefootnote}
\usepackage{comment}
\usepackage{soul}
\usepackage{graphics,graphicx,subfigure}
\usepackage{subcaption,caption,cleveref}
\hypersetup{
    colorlinks=true,
    linkcolor=blue,
    citecolor=blue,
    filecolor=magenta,      
    urlcolor=cyan,
}
\usepackage{booktabs}
\usepackage{multirow}
\usepackage{multicol}
\usepackage{natbib}
\usepackage{xcolor}
\usepackage{dsfont}
\usepackage{bbm}
\usepackage{enumitem}
\usepackage{blkarray}
\usepackage{amsfonts}
\usepackage{caption}

\usepackage{xr}

\usepackage{tikz}
\usetikzlibrary{positioning}
\usetikzlibrary{arrows.meta}

\makeatletter
\newcommand*{\addFileDependency}[1]{
  \typeout{(#1)}
  \@addtofilelist{#1}
  \IfFileExists{#1}{}{\typeout{No file #1.}}
}
\makeatother

\usepackage{footmisc}
\allowdisplaybreaks[4]
\def\bar{\overline}

\def\calA{{\mathcal A}}

\def\calC{{\mathcal C}}
\def\calD{{\mathcal D}}
\def\calE{{\mathcal E}}
\def\calF{{\mathcal F}}
\def\calG{{\mathcal G}}
\def\calH{{\mathcal H}}
\def\calI{{\mathcal I}}
\def\calJ{{\mathcal J}}
\def\calK{{\mathcal K}}

\def\calM{{\mathcal M}}
\def\calN{{\mathcal N}}

\def\calP{{\mathcal P}}

\def\calT{{\mathcal T}}
\def\calU{{\mathcal U}}
\def\calV{{\mathcal V}}

\def\calX{{\mathcal X}}

\def\calZ{{\mathcal Z}}

\def\MN{{\mathbb N}}

\def\PP{{\mathbb P}}
\def\QQ{{\mathbb Q}}
\def\RR{{\mathbb R}}

\def\ZZ{{\mathbb Z}}

\def\MA{{\mathbf{A}}}
\def\MB{{\mathbf{B}}}

\def\MG{{\mathbf{G}}}

\def\MI{{\mathbf{I}}}

\def\ML{{\mathbf{L}}}

\def\MN{{\mathbf{N}}}

\def\MP{{\mathbf{P}}}
\def\MQ{{\mathbf{Q}}}
\def\MR{{\mathbf{R}}}

\def\MU{{\mathbf{U}}}
\def\MV{{\mathbf{V}}}

\def\MX{{\mathbf{X}}}
\def\MY{{\mathbf{Y}}}
\def\MZ{{\mathbf{Z}}}
\def\MTheta{{\boldsymbol{\Theta}}}

\def\MXi{{\boldsymbol{\Xi}}}
\def\Mtheta{{\boldsymbol{\theta}}}
\def\Mmu{{\boldsymbol{\mu}}}

\def\MPa{{\mathbf{Pa}}}

\def\i{{\boldsymbol i}}
\def\j{{\boldsymbol j}}

\def\p{{\boldsymbol p}}

\def\r{{\boldsymbol r}}

\def\x{{\boldsymbol x}}

\def\z{{\boldsymbol z}}

\def\bgamma{{\boldsymbol \gamma}}

\def\bbeta{{\boldsymbol \beta}}
\def\boeta{{\boldsymbol \eta}}

\def\calA{{\cal  A}} 
 
\def\calC{{\cal  C}} 
\def\calD{{\cal  D}} 
\def\calE{{\cal  E}} 
\def\calF{{\cal  F}} 
\def\calG{{\cal  G}} 
\def\calH{{\cal  H}} 
\def\calI{{\cal  I}} 
\def\calJ{{\cal  J}} 
\def\calK{{\cal  K}} 
 
\def\calM{{\cal  M}} 
\def\calN{{\cal  N}} 
 
\def\calP{{\cal  P}}

\def\calT{{\cal  T}} 
\def\calU{{\cal  U}} 
\def\calV{{\cal  V}} 
 
\def\calX{{\cal  X}} 
 
\def\calZ{{\cal  Z}}  

\def\bone{\bfm 1}

\newcommand{\bfm}[1]{\ensuremath{\mathbf{#1}}}

\def\be{\bfm e}

\def\bh{\bfm h}

     \def\PP{\mathbb{P}}
     \def\QQ{\mathbb{Q}}
     \def\RR{\mathbb{R}}

\def\bu{\bfm u}     
\def\bv{\bfm v}     
     
\def\bx{\bfm x}     
     
\def\bz{\bfm z}     \def\ZZ{\mathbb{Z}}

\def\hat{\widehat}

\newtheoremstyle{customstyle}
  {\topsep}   
  {\topsep}   
  {\itshape} 
  {}          
  {\bfseries} 
  {.}         
  { }        
  {}  
\theoremstyle{customstyle}
\newtheorem{customassumption}{Assumption}

\newtheorem{theorem}{Theorem}
\newtheorem{definition}{Definition}
\newtheorem{proposition}{Proposition}
\newtheorem{remark}{Remark}

\newtheorem{assumption}{Assumption}
\newtheorem{claim}{Claim}
\newtheorem{lemma}{Lemma}
\newtheorem{corollary}{Corollary}

\theoremstyle{remark}

\usepackage{tikz}
\usetikzlibrary{shapes, calc, shapes, arrows, positioning}
\tikzstyle{qedge}=[->,thick,black]

\definecolor{myblue}{rgb}{0.0265,    0.6137,    0.8135}
\definecolor{myyellow}{rgb}{0.9290,    0.6940,    0.1250}

\tikzstyle{neuron}=[draw, circle, minimum size=25pt,inner sep=0pt, fill=black!10]

\tikzstyle{hidden}=[draw, circle, minimum size=25pt,inner sep=0pt, fill=white]

\usetikzlibrary{arrows.meta}
\tikzset{>={Latex[width=3mm,length=2mm]}}

\tikzstyle{arr}=[->, thick, black]

\usetikzlibrary{decorations.text}

\usetikzlibrary{shadows,shadings,shapes.symbols}
\tikzset{
    double color fill/.code 2 args={
        \pgfdeclareverticalshading[%
            tikz@axis@top,tikz@axis@middle,tikz@axis@bottom%
        ]{diagonalfill}{100bp}{%
            color(0bp)=(tikz@axis@bottom);
            color(50bp)=(tikz@axis@bottom);
            color(50bp)=(tikz@axis@middle);
            color(50bp)=(tikz@axis@top);
            color(100bp)=(tikz@axis@top)
        }
        \tikzset{shade, left color=#1, right color=#2, shading=diagonalfill}
    }
}

\newcommand{\E}{\mathbb{E}}

\newcommand{\yg}[1]{\textcolor{red!70!black}{[YG: #1]}}

\usepackage{titlesec}
\titlespacing*{\section}{0pt}{1pt}{1pt}
\titlespacing*{\subsection}{0pt}{1pt}{1pt}
\titlespacing*{\paragraph}{0pt}{*1}{*0}

\newcommand\blfootnote[1]{%
  \begingroup
  \renewcommand\thefootnote{}\footnote{#1}%
  \addtocounter{footnote}{-1}%
  \endgroup
}

\def\spacingset#1{\renewcommand{\baselinestretch}%
{#1}\small\normalsize}
\spacingset{1}
\begin{document}
\title{Discrete Causal Representation Learning}
\author{Wenjin Zhang$^*$\and Yixin Wang$^\dagger$ \and Yuqi Gu$^*$}
\date{$^*$Department of Statistics, Columbia University\\
$^\dagger$Department of Statistics,  University of Michigan}

\maketitle
\begin{abstract}
Causal representation learning seeks to uncover causal relationships among high-level latent variables from low-level, entangled, and noisy observations. Existing approaches often either rely on deep neural networks, which lack interpretability and formal guarantees, or impose restrictive assumptions like linearity, continuous-only observations, and strong structural priors. These limitations particularly challenge applications with a large number of discrete latent variables and mixed-type observations. To address these challenges, we propose \emph{discrete causal representation learning (DCRL)}, a generative framework that models a directed acyclic graph among discrete latent variables, along with a sparse bipartite graph linking latent and observed layers. This design accommodates continuous, count, and binary responses through flexible measurement models while maintaining interpretability. Under mild conditions, we prove that both the bipartite measurement graph and the latent causal graph are identifiable from the observed data distribution alone. We further propose a three-stage \emph{estimate-resample-discovery} pipeline: penalized estimation of the generative model parameters, resampling of latent configurations from the fitted model, and score-based causal discovery on the resampled latents. We establish the consistency of this procedure, ensuring reliable recovery of the latent causal structure. Empirical studies on educational assessment and synthetic image data demonstrate that DCRL recovers sparse and interpretable latent causal structures. 
\end{abstract}

\blfootnote{Correspondence to: Yuqi Gu. Email: \texttt{yuqi.gu@columbia.edu}. Address: 928 SSW Building, 1255 Amsterdam Avenue, New York, NY 10025.}

\spacingset{1.7}
\section{Introduction}\label{sec:intro}

Causal representation learning (CRL) seeks to recover high-level latent variables and their causal structure from low-level, entangled observations such as images, text, or time series \citep{scholkopf2021toward,moran2026towards}. While deep generative modeling approaches to CRL have shown strong empirical performance on complex data \citep{yang2021causal,khemakhem2020variational,javaloy2023causal,fan2023causal}, their  neural architectures remain black boxes with limited interpretability, impeding validation and understanding of what latent variables represent \citep{moran2026towards}.

Interpretability and identifiability of models are therefore central to uncovering latent causal mechanisms in complex datasets in a trustworthy manner. Informally, a causal representation is \emph{identifiable} if the observed distribution uniquely determines the parameters of the latent variable model and the causal relations among these latents, up to a specified equivalence relation capturing the unavoidable indeterminacies. Without identifiability, representation learning is prone to practical failures such as underspecification and posterior collapse \citep{d2022underspecification,wang2021posterior}.
In this work, we study causal structure learning in latent variable models with \emph{discrete} latent variables and general-response observed variables. The model has two structural components:  a directed acyclic graph (DAG) among the latent variables and a sparse bipartite measurement graph linking the latents to the observables. Our goal is to determine when these structures are identifiable from the observational distribution alone, and to develop a consistent procedure for recovering them.

A growing body of CRL work has established that, with \emph{continuous} latent variables, one can typically recover latents only up to permutations and per-coordinate reparameterizations \citep{vonkuegelgen2023nonparametric,jin2024learning}, rather than to a unique canonical form. Moreover, such equivalence classes are essentially tight: without additional structure or side information, stronger identifiability is unattainable \citep{varici2025score}. Even in a fully linear model with perfect single-node interventions, identifiability is limited to scaling and permutation~\citep{squires2022causal,buchholz2023learning}. This inherent indeterminacy prevents specific numerical values of latent variables from carrying stable semantic meaning, even in identifiable continuous-latent CRL models.

By contrast, discrete latent models with even highly nonlinear measurements can achieve identifiability up to latent-coordinate permutation alone \citep{lee2024new,lee2025deep}. Consequently, while continuous variables appear more expressive, only a limited portion of the information they encode is invariant under allowable reparameterizations and therefore robustly interpretable. Discrete latent variable models thus offer a more stable form of interpretability: the equivalence classes are smaller and easier to characterize, and latent coordinates can be more directly aligned with the ground truth causal factors. 

From a practical perspective, this discreteness is also often the right abstraction: in many settings, the goal is to infer an unobserved state that drives observations and supports downstream decisions, rather than to estimate a calibrated real-valued quantity. For instance, in medicine, probabilistic models often represent diseases as discrete latent variables that generate observed symptoms or test results, so different latent-state configurations correspond to different clinical regimes \citep{shwe1991probabilistic}. In educational measurement, cognitive diagnosis models are popular tools that employ discrete latent variables to model a student’s mastery/deficiency of multiple latent skills\citep{rupp2008unique,von2019handbook}. In such domains, learning a continuous latent coordinate first and then imposing cutoffs yields non-canonical thresholds whose meanings can shift under scaling without additional anchoring. Discrete latent variables instead represent the abstractions directly.

Motivated by these considerations, we propose a \emph{discrete causal representation learning} (DCRL) framework in which (i) discrete latent variables follow a latent DAG, and (ii) observations are generated through a sparse measurement graph linking the latent and observed layers with flexible mixed-type likelihoods. This framework accommodates continuous, count, and binary responses while allowing highly nonlinear latent-observation relationships.
Within this new framework, our contributions are threefold.

\emph{First}, despite the expressiveness of the proposed framework, we establish formal identifiability guarantees from a single observational distribution, without requiring interventions, multiple environments, or observed auxiliary variables.
Our main contribution is generic identifiability: under mild conditions, outside a measure-zero set of parameter values, the latent distribution, measurement layer, and latent DAG are identifiable from the observed data distribution, uniquely up to latent label permutations, so the unavoidable equivalence class consists solely of relabelings of latent variables.
Generic identifiability is directly analogous to how faithfulness excludes measure-zero violations of conditional independences in causal graphical models \citep{spirtes2000causation,ghassami2020characterizing}.
Under additional design conditions, we also obtain a stronger strict identifiability statement.

\emph{Second}, we propose and analyze a modular three-stage estimation pipeline. Stage~I fits the discrete generative process by penalized maximum likelihood via a stochastic approximation expectation-maximization (SAEM) algorithm with spectral initialization, yielding estimates of the latent distribution and measurement graph while remaining computationally efficient. Stage~II resamples latent configurations from the fitted latent law to construct a synthetic dataset in the latent space. Stage~III applies Greedy Equivalence Search (GES) \citep{chickering2002optimal} to this resampled latent dataset to recover the latent DAG. The validity of this algorithm relies on a key theoretical question: whether GES remains valid when it operates on samples from an \emph{estimated} latent law rather than from the true one. We answer this by extending classical notions of consistency and local consistency for scoring criteria to a rate-robust setting that permits the scoring distribution to converge to the truth at a controlled rate. Our analysis provides an explicit coupling between the convergence rate in Stage~I and the resampling size in Stage~II that guarantees GES applied to the resampled latents still recovers the Markov equivalence class of the true latent DAG.
 
\emph{Third}, we show that DCRL can reveal meaningful latent causal structure from data and yield interpretable discrete causal factors in practice. Through two empirical studies in educational assessment data and high-dimensional image data, we find that the learned latents align closely with domain-specific concepts and that the recovered latent DAG captures the underlying causal dependencies.

\medskip
\noindent\textbf{Organization.}~Section~\ref{sec: model} introduces the DCRL framework. In Section~\ref{sec: identifiablity}, we establish identifiability results for DCRL. Section~\ref{sec:est-theta} describes the proposed three-stage estimation pipeline and establishes its theoretical consistency guarantees. Section~\ref{sec:sim} and Section~\ref{sec:real} present simulation studies and real data applications, respectively, to demonstrate the effectiveness of our approach. Section~\ref{sec-discussion} concludes the paper with a discussion of potential extensions. All proofs are deferred to the Supplementary Material.
 
\medskip
\noindent\textbf{Notation.}
We write $a_N=\omega(b_N)$ (resp. $a_N=o(b_N)$) if $\tfrac{a_N}{b_N}\to\infty$
(resp. $\tfrac{a_N}{b_N}\to0$) as $N\to\infty$. For a positive integer $n$, we write $[n]=\{1,\ldots,n\}$. For any vector \(\bu\in \RR^K\), define \(\mathrm{supp}(\bu):=\{k\in[K]:u_k\ne 0\}\).
For a matrix $\MA$, we use $\MA_{i,:}$ (resp. $\MA_{:,j}$) for its $i$-th row (resp. $j$-th column).
For vectors $x,y\in\mathbb{R}^d$, write $x\succeq y$ (resp. $x\preceq y$) if $x_k\ge y_k$ (resp. $x_k\le y_k$) for all $k=1,\dots,d$.
We write $G \rtimes H$ for the semidirect product of groups $G$ and $H$.
For a set $A$ in a topological space, we write $A^\circ$ for the interior of $A$.

\section{Discrete Causal Representation Learning Framework}\label{sec: model}

\paragraph{Causal Graphical Models.}~
We begin by reviewing essential definitions to fix notation. Let $\MR = (R_1,\dots,R_d)$ be random variables with joint distribution $p^\star(\MR)$. We consider graphs $G = (V,E)$, where $V = {1,\dots,d}$ corresponds to the variables $R_i$ and $E \subseteq V \times V$ is the set of edges. A directed acyclic graph (DAG) is a directed graph with no cycles.

We now relate these graphs to conditional independencies. For a distribution $p$ on $\MR=(R_1,\dots,R_d)$, let
\(\mathcal{I}(p)
  :=
  \{ (A \perp\!\!\!\perp B \mid C)_p : R_A \perp\!\!\!\perp_p R_B \mid R_C \}
\)
denote the set of all conditional independence statements that hold under $p$, where $A,B,C \subseteq V$ are pairwise disjoint and $R_A = \{R_i : i \in A\}$. Let
\(\mathcal{I}(G)
  :=\{ (A \perp\!\!\!\perp B \mid C)_G : R_A \perp\!\!\!\perp_G R_B \mid R_C \}
\)
be the collection of conditional independences encoded by a DAG $G$ via d-separation. We say $p$ is Markov with respect to $G$ if $  \mathcal{I}(G) \subseteq \mathcal{I}(p)$. Write $\mathcal M(G):=\{p: \mathcal{I}(G)\subseteq \mathcal{I}(p)\}.$ For DAGs, $p\in\mathcal M(G)$ is equivalent to the factorization of the joint density according to $G$:
$p(\MR)=\prod_{i=1}^d p(R_i \mid R_{\MPa_i^G})$,
where $\MPa_i^G$ is the parent set of node $i$ in $G$ \citep{lauritzen1996graphical}. We call $p$ faithful to $G$ if $\mathcal{I}(p) \subseteq \mathcal{I}(G)$ \citep{koller2009probabilistic}. Combining the two inclusions, a distribution is DAG-perfect \citep{chickering2002optimal} to $G$ if $\mathcal{I}(G) = \mathcal{I}(p)$, so that $G$ encodes exactly all conditional independences of $p$. In this case, $G$ is a perfect map of $p$.

We say $G_1$ and $G_2$ are Markov equivalent if $\mathcal{I}(G_1) = \mathcal{I}(G_2)$, and write \(G_1\equiv G_2\).
Markov-equivalent DAGs form a Markov equivalence class. Two DAGs are Markov equivalent if and only if they have the same skeleton and v-structures \citep{verma1990equivalence}. Each Markov equivalence class can be uniquely represented by a completed partially directed acyclic graph (CPDAG), in which an edge $i \to k$ is directed if and only if it has the same orientation in every DAG in the class, and an edge $i - k$ is undirected if both orientations $i \to k$ and $i \leftarrow k$ occur among the DAGs in the class \citep{verma1990equivalence,chickering2002optimal}. 

A \emph{causal graphical model} consists of a DAG $G$ and a distribution $p$ that is Markov to $G$, with directed edges interpreted causally. In this paper we are primarily interested in \emph{discrete causal graphical models}, where each variable $R_i$ takes values in a finite state space and $(G,p)$ is a causal graphical model.

\paragraph{Discrete Causal Representation Learning.}~
We take the causal structure as primitive and place probabilistic distributions on top of it. We work with a collection of variables consisting of both observable and latent components. Let \(\MX=(X_1,\ldots,X_J)\in\times_{j=1}^J \calX_j\) denote the observed variables, where \(\mathcal{X}_j\subseteq \mathbb{R}\) is allowed to be general. In particular, our formulation accommodates a wide range of data types, including continuous measurements, count-valued observations, and binary or categorical responses. 
We consider binary latent variables and denote them by \(\MZ=(Z_1,\ldots,Z_K)\in\{0,1\}^K\), where $K\ge2$. The causal structure is specified by (i) a directed acyclic graph (DAG) \(\calG\) on the latent variables \(Z_1,\ldots,Z_K\), and (ii) a directed bipartite structure \(\MQ=(q_{j,k})\in\{0,1\}^{J\times K}\) from the latent variables to the observed variables, where \(q_{j,k}=1\) if and only if \(Z_k\) is a direct cause of \(X_j\) for \(j\in[J]\) and \(k\in[K]\). The bipartite graph describes the measurement modeling structure. Together, \((\calG,\MQ)\) determines a full acyclic causal graph on \((\MZ,\MX)\). An illustration is provided in Figure~\ref{fig:CDMM}.

\begin{figure}[h!]\centering
\resizebox{0.9\textwidth}{!}{
    \begin{tikzpicture}[scale=1.8]
       
    \node (v2)[hidden] at (1.8,0) {$Z_1$};
    \node (v3)[hidden] at (0.3,-0.8) {$Z_2$};
    \node (v4)[hidden] at (0.8,-1.7) {$\cdots$};
    \node (v5)[hidden] at (2.8,-1.7) {$\cdots$};
    \node (v6)[hidden] at (3.3,-0.8) {$Z_{K}$};
    
    \node (vv0)[neuron] at (-2.4,-3.4) {$X_1$};
    \node (vv1)[neuron] at (-1.6,-3.4) {$X_2$};
    \node (vv2)[neuron] at (-0.8,-3.4) {$\cdots$};
    \node (vv3)[neuron] at (0,-3.4) {$\cdots$};
    \node (vv4)[neuron] at (0.8,-3.4) {$\cdots$};
    \node (vv5)[neuron] at (1.6,-3.4) {$\cdots$};
    \node (vv6)[neuron] at (2.4,-3.4) {$\cdots$};
    \node (vv7)[neuron] at (3.2,-3.4) {$\cdots$};
    \node (vv8)[neuron] at (4.0,-3.4) {$\cdots$};
    \node (vv9)[neuron] at (4.8,-3.4) {$\cdots$};
     \node (vv10)[neuron] at (5.6,-3.4) {$X_J$};
    
    \draw[arr] (v2) -- (vv4);
    \draw[arr] (v2) -- (vv5);
    \draw[arr] (v2) -- (vv1);
    \draw[arr] (v3) -- (vv0);
    \draw[arr] (v3) -- (vv1);
    \draw[arr] (v3) -- (vv8);
    \draw[arr] (v4) -- (vv2);
    \draw[arr] (v4) -- (vv6);
    \draw[arr] (v4) -- (vv8);
    \draw[arr] (v5) -- (vv3);
    \draw[arr] (v5) -- (vv7);
    \draw[arr] (v5) -- (vv8);
    \draw[arr] (v6) -- (vv9);
    \draw[arr] (v6) -- (vv10);
    
    \node at (0,-2.2) {$\cdots$};
    \node at (0.7,-2.5) {$\cdots$};
    \node at (2.3,-2.8) {$\cdots$};
    \node at (3.8,-2.4) {$\cdots$};
    
    \draw[arr,red] (v3) -- (v2);
    \draw[arr,red] (v6) -- (v2);
    \draw[arr,red] (v2) -- (v5);
    \draw[arr,red] (v6) -- (v4);
    \path (v5) edge[dotted,thick,red] node [right] {} (v4);
    
    \node[anchor=west] (g0) at (6, -2.6) {$\MQ$: causal relations between $\MX$ and $\MZ$};
    \node[anchor=west] (g1) at (4.2, -0.8) {$\calG$: causal relations among $\MZ$};
\end{tikzpicture}
}
\caption{Latent DAG and measurement graph in discrete causal representation learning}
\label{fig:CDMM}
\end{figure}
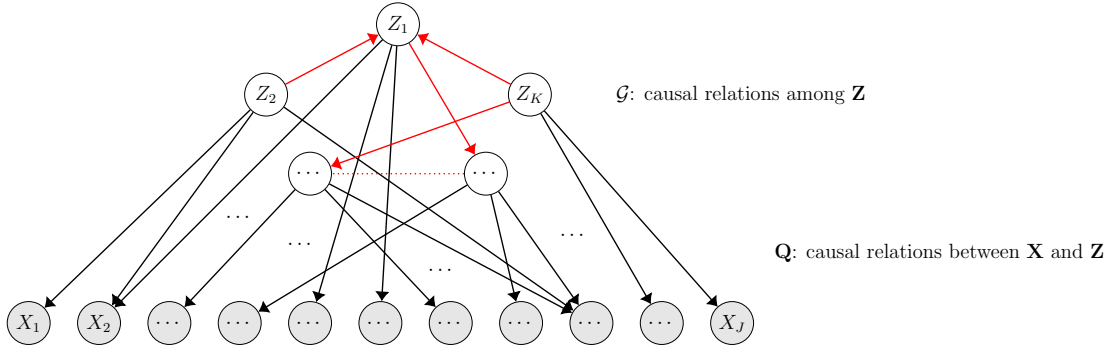

To obtain a data-generating process from $(\calG,\MQ)$, we view the joint distribution of the latent variables as the primitive object and assume that it satisfies the Markov factorization associated with $\calG$. Let $\p=(\PP(\MZ=\bz))_{\bz\in\{0,1\}^K}$ denote the $2^K$-dimensional probability vector of $\MZ$ with $\p\in\calM(\calG)$. Equivalently, $\p$ factorizes according to $\calG$ as
  $\p(\MZ)
  =
  \prod_{k=1}^K \p(Z_k | Z_{\MPa^{\calG}_k})$,
which encodes the directed dependencies prescribed by the latent DAG $\calG$. 

Next, we specify how the latent causes act on the observed variables according to $\MQ$. For each item $j\in[J]$, let $K_j=\{k\in[K]: q_{j,k}=1\}$ be the index set of its latent parents. We write the linear predictor \(\eta_j(\bz)\) as a multilinear polynomial in the binary latent vector \(\bz=(z_1,\ldots,z_K)\in\{0,1\}^K\).
For any subset \(S\subseteq [K]\), define the monomial feature \(\phi_S(\bz):=\prod_{k\in S} z_k\), with the convention \(\phi_\varnothing(\bz)\equiv 1\).
Fix once and for all an ordering \((S_1,\ldots,S_{2^K})\) of all subsets of \([K]\). Let \(\phi(\bz)\in\{0,1\}^{2^K}\) be the corresponding feature vector with entries \(\phi_m(\bz)=\phi_{S_m}(\bz)\).
For item \(j\), collect coefficients into \(\bbeta_j\in\RR^{2^K}\) by \((\bbeta_j)_m:=\beta_{j,S_m}\), and impose \((\bbeta_j)_m=0\) whenever \(S_m\nsubseteq K_j\), so that only main effects and interactions among coordinates in \(K_j\) are allowed.
Equivalently, \(\eta_j(\bz)=\bbeta_j^\top \phi(\bz)=\sum_{S\subseteq K_j}\beta_{j,S}\prod_{k\in S} z_k\). Stacking the rows yields the matrix
$\MB=[\bbeta_1\
\dots\ \bbeta_J]^\top
\in\mathbb{R}^{J\times 2^K}.$

In practice, it is often sufficient to restrict attention to the main effects of the latent causes, which corresponds to setting $\beta_{j,S}=0$ whenever \(|S|>1\). In this case, the linear predictor reduces to $\eta_j(\bz)=\beta_{j,\varnothing}+ \sum_{k\in K_j}\beta_{j,{k}}z_k,$ which is analogous to a generalized linear specification with an intercept and main effects in the binary latent vector \(\bz\) as covariates. 
Compared with the all-effect specification, this restriction reduces the number of parameters from \(2^{K}\) to \(|K_j|+1\) per item, yielding a more parsimonious and interpretable specification. In our subsequent estimation procedure, we will primarily focus on this main-effect specification.

Conditionally on $\MZ$, we model each $X_j$ by an item-specific parametric family consistent with $\MQ$ and assume conditional independence across $j$ given $\MZ$:
\begin{equation}\label{eq:observed exp fam}
X_j \mid \MZ \ \sim\ \mathrm{ParFam}_j\Big(g_j\big(\eta_j(\MZ),\ \gamma_j\big)\Big),
\qquad j\in[J].
\end{equation}
Here $\mathrm{ParFam}_j=\{\PP_{j,\Mtheta}:\Mtheta\in H_j\}$ is a known family with parameter space $H_j\subseteq\RR^{h_j}$, and $g_j:\RR\times[0,\infty)\to H_j$ is a known link mapping the linear predictor $\eta_j(\cdot)$ and, when applicable, a dispersion parameter $\gamma_j> 0$ to the parameter of $\mathrm{ParFam}_j$. Write $\bgamma=(\gamma_1,\ldots,\gamma_J)$.

Integrating out the latent $\MZ$ yields the marginal law of the observables,
\begin{equation}\label{eq:marginal probability}
\PP(\MX)\ =\ \sum_{\bz\in\{0,1\}^K} \PP\big(\MX \mid \MZ=\bz\big)\, \PP(\MZ=\bz),
\end{equation}
which is determined by the triple $(\p,\MB,\bgamma)$ together with the causal structure $(\calG,\MQ)$. 
\begin{definition}
We consider the following discrete causal representation learning \emph{data-generating process} parameterized by $(\p,\calG,\MB,\MQ,\bgamma)$, where the marginal law of the observed data is given by \eqref{eq:marginal probability}.
When $\calG$ and $\MQ$ are known and fixed, this induces a family of probability distributions on $\MX$, indexed by $(\MTheta,\calG,\MQ)$, which we denote by $\PP_{\MTheta,\calG,\MQ}$, where $\MTheta:=(\p,\MB,\bgamma)$.
\end{definition}

\begin{remark}\label{rem:poly_extension_preview}
For clarity, Section~\ref{sec: model} presents the framework under binary latent attributes $Z_k\in\{0,1\}$.
All components extend to ordered polytomous attributes $Z_k\in[M_k]=\{0,1,\ldots,M_k-1\}$.
In the extension, $\eta_j$ is still a linear combination of coefficients $\{\beta_{j,\bu}\}$ indexed by latent ``states'' $\bu$, and a term $\beta_{j,\bu}$ contributes to $\eta_j(\bz)$ only when $\bu\preceq\bz$ coordinatewise.
Moreover, $\beta_{j,\bu}$ is nonzero only if $\mathrm{supp}(\bu)\subseteq K_j$.
The extension and its generic identifiability result are stated in Section~\ref{sec:poly_extension}, with full definitions in Appendix~\ref{app:poly_extension}.
\end{remark}

Taken together, the DCRL framework yields a highly flexible and potentially strongly nonlinear measurement layer from the latent configuration to the observed responses. On the one hand, the all-effect specification for $\eta_j(\bz)$ allows arbitrary interaction patterns among the latent variables, including highest–order interactions over all attributes in $K_j$. On the other hand, by allowing a general parametric family $\mathrm{ParFam}_j$ and link mapping $g_j$, we impose no single fixed response family, so that a wide range of nonlinear conditional distributions $X_j \mid \MZ$ can be accommodated. 

Under DCRL, the complete causal structure comprises two components: the DAG $\calG$ among latent variables $Z_1,\ldots,Z_K$, and the directed bipartite structure $\MQ$ from latent variables to observed variables. Our goal is to jointly recover $\calG$ and $\MQ$ from $\MX$.

\paragraph{Connections with Existing Studies.}~
Most identifiable models for causal discovery with partially unobserved variables rely on linearity \citep{anandkumar2013learning,squires2022causal,huang2022latent,dong2026score}, an assumption that often fails in practice. Beyond linearity, identifiability has also been established for certain nonlinear latent hierarchical models \citep{prashant2024differentiable}. However, those results impose much stronger restrictions on the latent causal architecture than ours. See Supplement~\ref{sec: Connections with Existing Studies} for more details.

\section{Identifiability}\label{sec: identifiablity}

Before introducing our identifiability notion, we first relate our framework to the statistical-identifiability formulation used in recent CRL work \citep{xi2023indeterminacy,moran2026towards}. Consider the model $X=f(Z)+\epsilon$, where \(Z\in\calZ\) is the latent variable, \(f:\calZ\to\calX\) is the representation map, and \(\epsilon\) is noise. Let \(\calF\) be a class of admissible maps \(f\), and let \(\calP\) be a class of admissible latent distributions \(p\) on \(\calZ\). For any bijection \(\xi:\calZ\to\calZ\), one may rewrite the model as $X=(f\circ\xi^{-1})(\xi(Z))+\epsilon$. Therefore, without restrictions on \(\calF\) and \(\calP\), the model is trivially nonidentifiable. To formalize this, let \(\xi_{\#}p\) denote the push-forward of \(p\) by \(\xi\). One calls \(\xi\) an indeterminacy transformation if $f\circ\xi^{-1}\in\calF$ and $\xi_{\#}p\in\calP$. The collection of all such transformations is the indeterminacy set $A(\calF,\calP)$. Equivalently, $A(\calF,\calP)$ indexes the ``transformation-based'' comparison class $\{(f\circ\xi^{-1},\xi_{\#}p):\xi\in A(\calF,\calP)\}$ associated with a fixed pair \((f,p)\). The key question is to determine which restrictions on \(\calF\) and \(\calP\) make \(A(\calF,\calP)\) small while remaining flexible.

Our DCRL framework fits naturally into this framework. Here \(\calZ=\{0,1\}^K\). For Gaussian responses and identity link, we may write $X_j=\eta_j(Z)+\varepsilon_j$ with $\varepsilon_j\sim N(0,\gamma_j^2)$, so that \(f(Z)=(\eta_1(Z),\ldots,\eta_J(Z))\in\RR^J\). For a fixed latent DAG \(\calG\) and measurement graph \(\MQ\), the corresponding classes can be viewed as $\calP_{\calG}
:=
\{
p \text{ on }\{0,1\}^K:\ \p\text{ is Markov to }\calG
\}$, and $\calF_{\MQ}:=\{
f=(\eta_1,\ldots,\eta_J):\{0,1\}^K\to\RR^J:\ 
\eta_j(z)=\eta_j(z')\ \text{whenever } z_{K_j}=z'_{K_j},\ \forall j
\}$.
Additional assumptions used later in the paper can be understood precisely as further restrictions on \(\calP_{\calG}\) and \(\calF_{\MQ}\), imposed to shrink the admissible indeterminacy set.

This viewpoint is conceptually useful, but it is important to distinguish settings where the full statistical equivalence class coincides with the transformation-based class indexed by \(A(\calF,\calP)\) from those where it does not. In identifiable CRL, it is widely assumed that the generative map \(f\) is injective \citep{ahuja2023interventional,hartford2023beyond} and that the observation model is well-posed in the sense that the distribution of \(f(Z)\) is determined by that of \(X=g(f(Z),\epsilon)\). A simple sufficient condition is an additive and independent observation-noise model \(X=f(Z)+\epsilon\) with Gaussian noise distribution. Under these commonly used injective-generator and well-posed observation assumptions, \citet[Lemma~2.1]{xi2023indeterminacy} shows that if two parameterizations \((f_a,p_a)\) and \((f_b,p_b)\) induce the same observational distribution, then they must be related by a latent-space automorphism \(\xi\in\mathrm{Aut}(\calZ)\) in the sense that \(f_b=f_a\circ\xi^{-1}\) and \(p_b=\xi_{\#}p_a\) up to null sets. Hence, for such models, restricting attention to the transformation-based comparison class (equivalently, studying the indeterminacy set \(A(\calF,\calP)\)) entails no loss of generality. In the discrete case \(\calZ=\{0,1\}^K\), such a \(\xi\) is simply a permutation of the finite latent state space, so \(\xi_{\#}\p\) is merely a relabeling of \(\p\) and the remaining transformation-based comparison is finite.

In contrast, our DCRL framework does not directly assume a well-posed observation model in the above sense. Without this end-to-end well-posedness assumption, which would bypass much of the substantive difficulty, the reduction from general parameter-level equivalence to the transformation-based comparison class indexed by \(A(\calF,\calP)\) is no longer automatic. Accordingly, we begin by comparing arbitrary admissible parameter triples \((\Theta,\calG,\MQ)\) and \((\widetilde{\Theta},\widetilde{\calG},\widetilde{\MQ})\) that induce the same observational distribution, and then enforce this reduction by imposing explicit, verifiable structural restrictions encoded by \(\MQ\), together with more delicate techniques tailored to our model class. This additional collapse step is precisely what makes our analysis more involved than approaches that assume well-posedness directly.


Before proceeding, we first specify the parameter spaces. Define
\begin{align*}
&\Omega_{K}(\MTheta;\calG,\MQ) :=
\Bigl\{
\MTheta:\ \calG \ \text{is a perfect map of }\p,\
\beta_{j,S}=0\ \text{if }S\nsubseteq K_j,\ 
\beta_{j,\{k\}}\neq 0\ \text{iff }k\in K_j
\Bigr\}.
\\&\Omega_{K}(\MTheta,\calG,\MQ) :=
\Bigl\{
(\MTheta,\calG,\MQ):\MTheta\in\Omega_K(\MTheta;\calG,\MQ)
\Bigr\}.
\end{align*}
Now we introduce the equivalence relation that specifies the unavoidable ambiguity, and then define generic identifiability relative to this equivalence.
\begin{definition}\label{def:identifiability upto permutation saturated}
For the discrete CRL framework, define an equivalence relationship ``$\sim_{K}$'' by setting $(\MTheta, \calG,\MQ) \sim_{\calK} (\tilde{\MTheta},\tilde{\calG}, \tilde{\MQ})$ iff $\bgamma = \tilde{\bgamma}$ and there exists a permutation $\sigma \in S_{[K]}$ such that the following hold. First, $\p_{(z_{\sigma(1)}, \ldots, z_{\sigma(K)})} = \tilde{\p}_{\bz}$ for all $\bz \in \{0, 1\}^{K}$ and $\calG \equiv \sigma(\tilde{\calG})$, where $\sigma(\tilde{\calG})$ denotes the DAG obtained from $\tilde{\calG}$ by relabeling each node $k$ as $\sigma(k)$. Second, $q_{j,k} = \tilde{q}_{j,\sigma(k)}$ for all $j \in [J],\, k \in [K]$, and for all $j$ and $S \subseteq [K]$, $\beta_{j,S} = \tilde{\beta}_{j,\sigma(S)}$, where $\sigma(S) := \{\sigma(k) : k \in S\}$.
\end{definition}
\begin{definition}\label{def:generic identifiability}
Let $(\MTheta^\star,\calG^\star,\MQ^\star)\in\Omega_{K}(\MTheta, \mathcal{G}, \MQ)$ be the true parameter triple of the discrete causal representation learning framework.
The framework is generically identifiable up to $\sim_\calK$ if $\{\MTheta \in \Omega_K(\MTheta;\calG^\star,\MQ^\star)\!:\, \exists\, (\widetilde{\MTheta},\widetilde{\MQ},\widetilde{\calG}) \not\sim_\calK (\MTheta,\MQ^\star,\calG^\star) \ \text{such that } 
\PP_{\widetilde{\MTheta},\widetilde{\MQ},\widetilde{\calG}} = \PP_{\MTheta,\MQ^\star,\calG^\star}\}$ is a measure-zero set with respect to $\Omega_K(\MTheta;\calG^\star,\MQ^\star)$.
\end{definition}

This equivalence relation \(\sim_{\calK}\) is the discrete analogue of the transformation-based indeterminacies encoded by \(A(\calF,\calP)\), specialized to latent-coordinate relabelings. 

The measure-zero qualifier in Definition~\ref{def:generic identifiability} parallels the faithfulness convention in causal discovery: for a fixed DAG, distributions that are Markov but unfaithful form a Lebesgue-null subset of the parameter space, so faithfulness excludes only a negligible set of degenerate configurations. Our generic-identifiability definition plays the same role here. After imposing the latent Markov-plus-faithfulness condition on \(\p\), non-identifiability can still occur for exceptional values of the continuous measurement parameters \((\MB,\bgamma)\), but these exceptional configurations form a Lebesgue-null set. Hence, restricting attention to generic identifiability excludes only a measure-zero subset of ``bad'' $(\MB,\bgamma)$ configurations. In this sense, the loss incurred by discarding measure-zero subsets of parameters is as harmless as the loss incurred when imposing faithfulness in the first place.

One may also consider the stronger notion of strict identifiability, under which equality of observational distributions implies equivalence up to $\sim_{\calK}$ for every admissible parameter triple, rather than for all but a measure-zero subset. We do not emphasize this stronger notion in the main text, because the corresponding strict-identifiability statements together with several related extensions can be obtained by adapting existing arguments from \citet{liu2025exploratory,lee2025deep}. We therefore record these formal results in Section~\ref{app:strict_id}.

We now introduce the assumptions needed for our generic identifiability result.
\begin{assumption}\label{assm1}
\begin{enumerate}
    \item[(a)] $\calG$ is a perfect map of $\p$ and $\p_{\bz}\in(0,1)$ for all $\bz\in\{0,1\}^{K}$.
    \item[(b)] For each item $j$,
    \(\eta_j(\bz) >\eta_j(\bz')\) whenever \(
    \bz\succeq \MQ_{j,:}\) and \( \bz'\not\succeq \MQ_{j,:}.\)
\end{enumerate}
\end{assumption}
Assumption~\ref{assm1}(a) is a restriction on the latent distribution class \(\calP_\calG\). Assumption~\ref{assm1}(b) is a restriction on \(\calF_\MQ\): it imposes a monotonicity condition on each item response function \(\eta_j\). This type of condition is also used in \citet{liu2025exploratory,lee2025deep} to avoid the sign-flipping for each latent variable.

To ensure generic identifiability, we introduce an additional analytic assumption, which holds for regular minimal exponential families on the interior of natural parameter spaces.

\begin{assumption}\label{assm2}
For each $j\in[J]$, define the canonical countable separating class
$\calC_j^{\mathrm{can}}:=
\{\calX_j\cap(a,b] : a,b\in\QQ,\ a<b\}\cup\{\calX_j\}$.
Assume that $H_j^\circ\neq\emptyset$, and that the following hold.
\begin{enumerate}
\item[(i)] For every $S\in\calC_j^{\mathrm{can}}$, the map $\Mtheta\mapsto\PP_{j,\Mtheta}(S)$ is real-analytic on $H_j^\circ$.
\item[(ii)]
$\PP_{j,\Mtheta}=\PP_{j,\Mtheta'}$ implies $\Mtheta=\Mtheta'$ for all $\Mtheta,\Mtheta'\in H_j^\circ$.
\item[(iii)] The link $g_j$ maps $\RR\times(0,\infty)$ into $H_j^\circ$ and is real-analytic.
Moreover, exactly one of the following holds.
(a) (No dispersion) The link is independent of $\gamma$ and the slice map $\eta\mapsto g_j(\eta,\gamma_0)$ is injective for some (equivalently, any) fixed $\gamma_0\in[0,\infty)$.
(b) (With dispersion) The full map $(\eta,\gamma)\mapsto g_j(\eta,\gamma)$ is injective on $\RR\times(0,\infty)$.
\end{enumerate}
\end{assumption}

Our main identifiability result is stated in the following theorem.

\begin{theorem}\label{thm2}
Under Assumptions~\ref{assm1} and~\ref{assm2}, DCRL is generically identifiable if the following hold.
\begin{enumerate}
\item[(i)] After a row permutation, we can write $\MQ^\star=[\MQ_1^\top, \MQ_2^\top,\MQ_3^\top]^\top$, where $\MQ_1,\MQ_2\in\{0,1\}^{K\times K}$ have unit diagonals (off–diagonals arbitrary), and $\MQ_3$ has no all-zero column.
\item[(ii)] No column of $\MQ^\star$ contains another: for any $p\neq q$,
neither $\MQ^\star_{:,p}\succeq \MQ^\star_{:,q}$ nor $\MQ^\star_{:,q}\succeq \MQ^\star_{:,p}$.
\end{enumerate}
\end{theorem}
Condition (i) is best viewed as a weak coverage requirement on the measurement design: it guarantees that every latent coordinate affects at least one observed variable, and that there exist some anchor-like items in which each latent is forced to appear.
However, it is weak in the sense that such anchor-like items may still depend on many other latents.
Our condition (ii) coincides with Condition 3.1 (the subset condition) in \cite{kivva2021learning}.
As they observed, violating this subset condition can lead to non-identifiability. Similarly, we emphasize that condition (i) alone is not sufficient for generic identifiability.
In particular, Appendix~\ref{counterexample-2} constructs a model satisfying condition (i) in which $\MQ$ has distinct columns and each column contains at least one zero, yet the framework fails to be generically identifiable.
\textbf{}

Our proof is a three-step reduction of the comparison class. First, a Kruskal-type tensor argument collapses the original continuous parameter comparison to a finite, transformation-generated comparison class: any remaining competitor with the same observational law must be of the form \((\xi_{\#}\p,\boeta\circ \xi^{-1})\), where \(\xi\in S_{2^K}\) is a permutation of the \(2^K\) latent states of \(\MZ\). Thus, after the tensor step, we are essentially in the setting of \citet{moran2026towards}, and the problem becomes how to shrink the indeterminacy set \(\calA(\calF,\calP)=S_{2^K}\).
Before further reducing \(\calA(\calF,\calP)\), we first identify \(\MQ\): using Assumption~\ref{assm1}(b) and an inclusion--exclusion argument on the transformed \(\boeta\)-array, we show that all these \(2^K\) admissible competitors must have measurement matrices that agree with the true \(\MQ\) up to a coordinate permutation. This observation makes the structural constraint on \(\calF_\MQ\) much clearer. We then return to reducing \(\calA(\calF,\calP)\): the subset condition, together with the structural constraint on \(\calF_\MQ\), reduces the indeterminacy set from \(S_{2^K}\) to \((\ZZ_2)^K\rtimes S_K\), corresponding to coordinate permutations combined with coordinatewise bit-flips, and Assumption~\ref{assm1}(b) further rules out bit-flips, leaving only coordinate relabelings in \(S_K\). Finally, \(\bbeta\) and \(\calG\) are recovered from the invertible linear map \(\boeta\leftrightarrow\bbeta\) and the perfect-map condition, yielding identifiability.

\paragraph{Extension to Polytomous Attributes.}\label{sec:poly_extension}~The binary-latent framework of Section~\ref{sec: model} extends naturally to the case where each latent attribute \(Z_k\) takes values in \([M_k]=\{0,1,\ldots,M_k-1\}\) with \(M_k\ge2\). The linear predictor \(\eta_j(\bz)\) generalizes to a sum over coefficients \(\{\beta_{j,\bu}\}\) indexed by \(\bu\in\prod_{k=1}^K [M_k]\), where \(\beta_{j,\bu}\) contributes to \(\eta_j(\bz)\) only if \(\operatorname{supp}(\bu)\subseteq K_j\) and \(\bu\preceq\bz\) coordinatewise, thereby preserving both the sparsity structure encoded by \(\MQ\) and the causal interpretation of each entry \(q_{j,k}\). Under Assumption~\ref{assm1}(a), Assumption~\ref{assm2}, and an ordered-level analogue of the monotonicity condition in Assumption~\ref{assm1}(b) (Assumption~\ref{assm:poly_monotone_binary_style}), we establish generic identifiability (Theorem~\ref{thm:poly_generic_ident_delta} in Section~\ref{app:poly_extension}) with two key differences from the binary case. First, the measurement design requires at least \(2\sum_{k=1}^K \lceil \log_2 M_k\rceil+1\) observed variables, which grows only logarithmically in the numbers of categories and is order-sharp for the general-response setting. Second, because the inclusion--exclusion reconstruction of \(\MQ\) does not directly extend to the polytomous setting, we instead use the subset condition to identify \(\MQ\) and shrink the indeterminacy set;
this corresponds to within-coordinate level permutations composed with across-coordinate relabelings.
We then use the monotonicity condition to eliminate the within-coordinate permutations, leaving only coordinate relabelings.
Unlike in the binary case, this group-reduction step excludes a Lebesgue-null exceptional set, reflecting the genuine additional difficulty of the polytomous setting; see Supplement~\ref{app:poly_extension} for details.

\paragraph{Nonparametric Disentanglement. }~As a further consequence, we obtain a nonparametric disentanglement result in the following corollary.
\begin{corollary}\label{cor:poly_moran_indeterminacy}
Let $\calZ=\prod_{k=1}^K [M_k]$, where $M_k\ge 2$, and consider the representation model $\MX=f(\MZ)+ \epsilon$, where $f=(f_1,\ldots,f_J):\calZ\to\RR^J$. For a binary matrix \(\MQ=(q_{j,k})\in\{0,1\}^{J\times K}\), write $K_j:=\{k\in[K]:q_{j,k}=1\}$.
Define
$\calP:=\{\p\text{ on }\calZ:\ \sum_{\bz\in\calZ} p_\bz=1\}$ and 
\[
\calF
:=
\Biggl\{
f:\calZ\to\RR^J\ \Bigg|\ 
\exists\,\MQ\in\{0,1\}^{J\times K}\ \text{s.t., }\begin{aligned}
&\text{neither }\MQ_{:,a}\succeq \MQ_{:,b}\ \text{nor }\MQ_{:,b}\succeq \MQ_{:,a},\quad (\forall\ a\neq b),\\
&f_j(\bz)=f_j(\bz')\ \text{ whenever }\ \bz_{K_j}=\bz'_{K_j}\quad (j\in[J]),\\
&f_j(\bz)\neq f_j(\bz')\ \text{ whenever }\ \bz_{K_j}\neq \bz'_{K_j}\quad (j\in[J])
\end{aligned}\Biggr\}
\]
Then every admissible indeterminacy transformation \(\xi\in A(\calF,\calP)\) can only permute coordinates (with the same number of categories) and relabel the levels within each coordinate.
\end{corollary}

Corollary~\ref{cor:poly_moran_indeterminacy} establishes disentanglement in the sense of \citet{moran2026towards}: among all admissible latent bijections \(\xi:\calZ\to\calZ\) that preserve membership in \((\calF,\calP)\), the only ones that remain are element-wise transformations and permutations. Notably, this conclusion is obtained with \(f\) treated nonparametrically, subject only to the restrictions in \(\calF\), rather than under a specific parameterization like all-effect forms. Moreover, under injective-generator and well-posed observation assumptions that are common in identifiable CRL \citep{ahuja2023interventional,khemakhem2020variational}, the full statistical equivalence class coincides with the transformation-based indeterminacy class \citep{xi2023indeterminacy}, so the corollary yields a direct disentanglement conclusion from observational data alone. This contrasts with much of the existing literature, where even after adopting the same baseline, additional sources of variation---such as auxiliary variables, multiple environments, or interventions---are typically invoked to further shrink the indeterminacy set down to permutations and component-wise transformations \citep{khemakhem2020variational,ahuja2023interventional}. In our framework, the structural constraints on $f$ encoded by $\MQ$ together with the subset condition are already sufficient to enforce this shrinkage at the level of $A(\calF,\calP)$ from observational data alone. This can be viewed as a natural restriction when $\MQ$ is used to describe which latents affect which measurements. Compared with other observational-only identifiability results, our assumptions are often milder: prior work commonly relies on anchor features \citep{moran2022identifiable,prashant2024differentiable}, Gaussian-mixture latent structure \citep{kivva2022identifiability}, or access to a mixture oracle \citep{kivva2021learning}, whose existence can fail for common discrete-response observation models (see Supplement~\ref{sec: Connections with Existing Studies}). If one further imposes a monotonicity condition, the within-coordinate relabelings here can also be removed.

\section{Estimation Procedure and Theoretical Guarantees}\label{sec:est-theta}
With a slight abuse of notation, we let $\MX$ denote the $N\times J$ data matrix with rows $\MX_1,\ldots,\MX_N$. Similarly, let $\MZ$ denote the corresponding $N\times K$ latent variable matrix with rows $\MZ_1,\ldots,\MZ_N$. Given the bottom-layer data $\MX\in\RR^{N\times J}$  with unknown latent causal structures and parameters $(\p, \calG,\MB, \MQ,\bgamma)$, our objective is to recover $\MQ$ and $\calG$. In this section, we provide the complete pipeline for recovering the measurement graph and latent DAG (Algorithm~\ref{al:pipeline}).

\begin{algorithm}[h!]
\SetAlgoLined
\caption{{\small Discrete Causal Representation Learning: Estimate-Resample-Discovery}}
\label{al:pipeline}
\KwData{$\MX$, $K$, strictly increasing sampling rule $f: \ZZ_+\to\ZZ_+$.}

\BlankLine
\textbf{Stage I: Parameter Estimation}\;
  Obtain the estimated $(\hat{\p}, \hat{\MB}, \hat{\bgamma})$\;
  Set $\hat{\MQ}$ by the support of $\hat{\MB}$: 
  $\hat{q}_{j,k} \gets \mathbf{1}(\hat{\beta}_{j,k} \neq 0)$\;

\BlankLine
\textbf{Stage II: Latent Resampling from $\hat{\p}$}\;
  Draw $f(N)$ i.i.d.\ samples $\MZ^{(1)}, \dots, \MZ^{(f(N))} \stackrel{\text{i.i.d.}}{\sim} \hat{\p}$ and stack them into $\hat{\MZ} \in \{0,1\}^{f(N)\times K}$\;

\BlankLine
\textbf{Stage III: Causal Discovery on Latent Space}\;
  Perform a causal discovery method on $\hat{\MZ}$ to obtain the estimated $\hat{\calG}$\;

\BlankLine
\textbf{Output:} $(\hat{\p}, \hat{\MB}, \hat{\bgamma}, \hat{\MQ}, \hat{\calG})$.
\end{algorithm}

Next, we describe the pipeline in detail. In Stage~I, we apply a penalized maximum likelihood estimator to obtain $(\hat{\p}, \hat{\MB}, \hat{\bgamma})$ and hence $\hat{\MQ}$, as detailed in Section~\ref{subsec: EM}. In Stage~II, we fix a strictly increasing sampling rule $f:\ZZ_+\to\ZZ_+$ and draw $f(N)$ i.i.d.\ latent samples from $\hat{\p}=(\hat p_\bz)_{\bz\in\{0,1\}^K}$ to form the resampled matrix $\hat{\MZ} \in \{0,1\}^{f(N)\times K}$. The requirement that $f$ be strictly increasing is imposed purely for notational convenience, and our analysis depends only on the growth rate of $f(N)$ relative to $N$. In Stage~III, we run Greedy Equivalence Search (GES; \citealp{chickering2002optimal}) on $\hat{\MZ}$ to obtain an estimate $\hat{\calG}$; details on GES are given in Section~\ref{subsec: GES}. In Section~\ref{subsec: Local Consistency}, we specify suitable ranges for the resampling size $f(N)$ and conditions to ensure that this three-stage pipeline enjoys rigorous consistency guarantees. 

While our identifiability results hold for the all-effect specification, in practice we focus on the main-effect specification for parsimony, interpretability, and computational tractability.

\subsection{Penalized likelihood estimation via Gibbs–SAEM}\label{subsec: EM}
Let $\ell(\MTheta\mid\MX)=\sum_{i=1}^N \log \PP(\MX_i\mid\MTheta)$ be the marginal log-likelihood from \eqref{eq:marginal probability}, and define

\vspace{-5mm}
\begin{equation}\label{eq:penalized optimization}
\hat{\MTheta}\in\arg\max_{\MTheta}\Big\{\ell(\MTheta\mid\MX)-p_{\lambda_N,\tau_N}(\MB)\Big\},
\end{equation}
where $p_{\lambda_N,\tau_N}$ is an entrywise penalty extended additively over $\MB=(\beta_{j,k})$: $p_{\lambda_N,\tau_N}(\MB)=\sum_{k=1}^K\sum_{j=1}^J p_{\lambda_N,\tau_N}(\beta_{j,k}).$
The bipartite matrix is then estimated by thresholding:

\vspace{-5mm}
\begin{equation}\label{eq:graphical matrix estimator}
    \hat{q}_{j,k}=\mathbf{1}(\hat{\beta}_{j,k}\neq 0),\qquad j\in[J],\ k\in[K].
\end{equation}
Throughout, we assume that the penalty $p_{\lambda_N,\tau_N}$ satisfies the regularity conditions in \citet[Supplement~S.1.4]{lee2025deep}, including standard truncated sparsity–inducing penalties such as the truncated Lasso penalty (TLP) \citep{shen2012likelihood} and the smoothly clipped absolute deviation (SCAD) penalty \citep{fan2001variable}, and we do not reproduce them here 

The following result establishes consistency of the parameters and the bipartite graph.
\begin{theorem}[Theorem 3 in \cite{lee2025deep}]\label{thm:estimation consistency}
Let $(\MTheta^\star,\MQ^\star,\calG^\star)$ denote the true parameters in the discrete CRL framework. Suppose the parameter space is compact and all entries of $\MB$ are bounded, the data-generating process at $\MTheta^\star$ is identifiable and has nonsingular Fisher information, and the tuning parameters satisfy $1/\sqrt{N}\ll\tau_N\ll \lambda_N/\sqrt{N}\ll 1$. If $\hat{\MTheta}$ solves \eqref{eq:penalized optimization}, then there exists a relabeling $\tilde{\MTheta}\sim_\calK\hat{\MTheta}$ such that $\|\tilde{\MTheta}-\MTheta^\star\|=O_p(1/\sqrt{N})$. Moreover, with $\tilde{\MQ}$ computed from $\tilde{\MTheta}$ via \eqref{eq:graphical matrix estimator}, one has $\PP(\tilde{\MQ}\neq \MQ^\star)\to 0$.
\end{theorem}

We compute \eqref{eq:penalized optimization} via a penalized Gibbs--SAEM algorithm \citep{delyon1999convergence,kuhn2004coupling,lee2025deep}. A direct EM implementation would require evaluating the conditional expectation of the complete-data log-likelihood by summing over all \(2^K\) latent configurations. This leads to an \(O(NJ\,2^K)\) cost per outer iteration, which becomes prohibitive once \(K\) is moderate. Gibbs–SAEM sidesteps this by replacing the exact E–step with low-cost Gibbs updates. The full pseudocode is given in Algorithm~\ref{al: PSAEM} in Section~\ref{sec: Detailed SAEM}.

In the E--step we run an alternating coordinate Gibbs sampler~\citep{gu2023a} targeting the posterior \(\PP(\MZ\mid \MX;\MTheta^{[t]})\) under the parameter iterate \(\MTheta^{[t]}=(\bbeta^{[t]},\bgamma^{[t]},\mathbf p^{[t]})\). In our implementation we take a single Gibbs sweep per outer iteration (\(C=1\)) following \citet{delyon1999convergence}. 
Within this sweep we visit each coordinate \(Z_{i,k}\) in turn. Conditional on the current value of all other coordinates \(\MZ_{i,-k}\), the conditional distribution of \(Z_{i,k}\) is Bernoulli with success probability
$\PP\big(Z_{i,k}=1 \mid \mathbf Z_{i,-k},\MY;\MTheta^{[t]}\big)
\;\propto\;
\PP\big(Z_{i,k}=1,\mathbf Z_{i,-k};\mathbf p^{[t]}\big)
\prod_{j=1}^J
\PP\big(X_{ij}\mid \mathbf Z_i;\bbeta^{[t]}_j,\gamma^{[t]}_j\big)$,
and similarly for \(Z_{i,k}=0\). Equivalently, the log-odds of \(Z_{i,k}=1\) versus \(0\) is given by the difference in log joint densities when flipping that bit. This conditional has a closed form and is straightforward to sample from. To keep each flip inexpensive, we maintain for every sample–item pair the linear score $\psi_{ij}= \beta_{j,0}^{[t]} + \sum_{k=1}^K \beta_{j,k}^{[t]} Z_{ik},$ so that the likelihood contribution of \(\mathbf Z_i\) to \(X_{ij}\) can be updated by rank-one changes in \(\psi_{ij}\). Flipping a single bit \(Z_{i,k}\) modifies all \(\psi_{ij}\) by adding or subtracting \(\beta_{j,k}^{[t]}\), which costs \(O(J)\) operations. A full Gibbs sweep visits all \(K\) coordinates for each of the \(N\) individuals, so the E--step costs \(O(NJK)\) per outer iteration, in contrast to the \(O(NJ\,2^K)\) cost of exact EM when \(K\) grows.

The Gibbs sweep produces an updated \(\mathbf Z^{[t+1]}\), inducing an empirical distribution on \(\{0,1\}^K\) that we use to update \(\mathbf p\) via Robbins--Monro stochastic approximation. Writing \(\widehat{\mathbf p}^{[t+1]}\) for the empirical proportions of the sampled configurations at iteration \(t+1\), we set
\(\mathbf p^{[t+1]}=
(1-\theta_{t+1})\,\mathbf p^{[t]} + \theta_{t+1}\,\widehat{\mathbf p}^{[t+1]},
\)
with stepsizes \(\{\theta_t\}_{t\ge1}\) satisfying \(\sum_t \theta_t=\infty\) and \(\sum_t \theta_t^2<\infty\). 

In the M--step, for each row \(j\) of \(\MB\) we set \[ F_j^{[t+1]}(\bbeta_j,\gamma_j) =\;(1-\theta_{t+1})\,F_j^{[t]}(\bbeta_j,\gamma_j) +\frac{\theta_{t+1}}{C}\sum_{r=1}^{C}\sum_{i=1}^{N}\log \PP\!\big(X_{ij}\mid \mathbf Z_i=\mathbf Z^{[t+1],r}_i;\bbeta_j,\gamma_j\big), \]where $F_j^{[0]}\equiv0$ for $j\in[J]$, and then solve the penalized maximization $ (\bbeta_j^{[t+1]},\gamma_j^{[t+1]})\;=\;\arg\max_{\bbeta_j,\gamma_j}\,\Big\{F_j^{[t+1]}(\bbeta_j,\gamma_j)\;-\;p_{\lambda_N,\tau_N}(\bbeta_j)\Big\}.$
Together, the Gibbs E--step and penalized SAEM M--step provide an efficient stochastic optimization scheme for the penalized objective \eqref{eq:penalized optimization}.

In complex latent variable models, the penalized SAEM algorithm may converge to local maxima if poorly initialized. To mitigate this, we adopt a spectral initialization that first applies the universal singular value thresholding \citep{chatterjee2015matrix}  procedure for low-rank generalized linear factor models \citep{zhang2020note}, followed by an SVD-based Varimax rotation to obtain sparse factor loadings \citep{rohe2023vintage}. See Section~\ref{sec: Detailed Ini} for details.

Algorithm~\ref{al: PSAEM} in the Supplement yields estimates of $\p$ and $\MQ$, which suffices to infer the latent causal structure. It remains to sample latent variables from $\hat\p$ and apply a causal discovery method. In this work, we employ GES and therefore first review its relevant details.

\subsection{Greedy Equivalence Search}\label{subsec: GES}
Greedy Equivalence Search (GES) is a score-based causal discovery method that searches over Markov equivalence classes (MECs) and is guaranteed to identify the MEC of the true DAG under suitable conditions on the scoring criterion. This subsection briefly reviews these conditions, as they will guide our choice of score in the presence of estimation error from $\hat{\p}$. 

Let $\calD$ denote $N$ i.i.d.\ samples from a distribution $p^\star$ that is faithful to DAG $G^\star$. In a score-based framework, each DAG $G$ is assigned a score $S(G;\calD)$, and the estimation problem is formulated as $ G^\star \in \arg\max_{G} S(G;\calD).$ Based on a scoring criterion, GES performs a two-phase greedy search over equivalence classes represented by CPDAGs: a forward phase that repeatedly applies a valid insertion whenever it increases the score, followed by a backward phase that greedily applies valid deletions until no further score increase is possible.

We recall several key properties of scoring criteria. The theoretical guarantees of GES rely on the following standard notions \citep{chickering2002optimal}: score equivalence and local consistency.

\emph{First}, a score is score equivalent if it gives identical scores to all DAGs in the same MEC.



\emph{Second}, the score is locally consistent if, when $G'$ is obtained from $G$ by adding an edge $u\to v$, we have in the limit as $N\to\infty$ that $S(G',\mathcal D)>S(G,\mathcal D)$ if $R_v \not\!\perp\!\!\!\perp_{p^\star} R_u \mid \mathbf{Pa}_v^{G}$, and $S(G',\mathcal D)<S(G,\mathcal D)$ if $R_v \perp\!\!\!\perp_{p^\star} R_u \mid \mathbf{Pa}_v^{G}$.

A score commonly used in score-based methods is BIC, which is defined for graphical models as $S(G; \mathcal{D}) = \log p_{\hat\Mtheta}(\mathcal{D}) - \tfrac{1}{2} d_G \log N$
\citep{schwarz1978estimating,koller2009probabilistic},
where $\hat\Mtheta$ is the maximum likelihood estimate over $\calM(G)$ and $d_G$ denotes the number of free parameters in the model associated with $G$. Since we focus on discrete Bayesian networks, we adopt the score-equivalent BDeu score (Bayesian Dirichlet equivalent uniform), defined as the log marginal likelihood under a uniform Dirichlet prior on each conditional probability table \citep{heckerman1995learning}, to recover $\calG$ in the following estimation pipeline. For a fixed equivalent sample size, the BDeu score and BIC differ only by an $O(1)$ term \citep{koller2009probabilistic}, and are therefore asymptotically equivalent.

If a score-equivalent criterion is locally consistent, GES returns the MEC of $G^\star$ as $N\to \infty$ \citep[Theorem~4]{chickering2002optimal}. Moreover, BDeu is locally consistent for discrete Bayesian networks \citep{chickering2002optimal}, guaranteeing that GES recovers the MEC of $G^\star$ as $N\to\infty$.

\subsection{Theoretical Guarantees}\label{subsec: Local Consistency} 
Our main result is stated as follows.

\vspace{-5mm}
\begin{theorem}
\label{thm:consistency_pipeline}
Consider the discrete causal representation learning framework with true parameters $(\p^\star, \calG^\star, \MB^\star, \MQ^\star, \bgamma^\star)$. 
Suppose Assumption~\ref{assm1} holds, the framework is identifiable up to $\sim_\calK$, the Fisher information at $\Theta^\star$ is nonsingular, and the entries of $\MB$ are uniformly bounded. 
Assume that Stage~I of Algorithm~\ref{al:pipeline} employs the penalized optimization problem~\eqref{eq:penalized optimization} with tuning parameters $(\tau_N,\lambda_N)$ satisfying $\frac{1}{\sqrt{N}} \ll \tau_N \ll \frac{\lambda_N}{\sqrt{N}} \ll 1$, and that Stage~III applies Greedy Equivalence Search with the BDeu score to the resampled latent data $\hat{\MZ}$. 
Further assume that the sampling rule $f:\ZZ_+\to\ZZ_+$ is strictly increasing and satisfies $f(N) = o(N \log N)$. 
Then, as $N\to \infty$:
\begin{enumerate}
    \item[\textnormal{(i)}]  
    There exists $\tilde{\Theta} \sim_\calK (\hat{\p}, \hat{\MB}, \hat{\bgamma})$ such that 
    $\|\tilde{\Theta} - \Theta^\star\| = O_p(N^{-1/2})$ and 
    $\mathbb{P}(\hat{\MQ} = \MQ^\star) \to 1$.
    \item[\textnormal{(ii)}] 
    $\hat{\calG}$ recovers the MEC of $\calG^\star$ with probability tending to $1$.
\end{enumerate}
\end{theorem}
In practice, we specifically employ Penalized Gibbs--SAEM (Algorithm~\ref{al: PSAEM}) to solve \eqref{eq:penalized optimization} and thereby complete Stage I of Algorithm~\ref{al:pipeline}. Below, we provide a detailed account of how the theorem is established and highlight the main technical challenges.

Because Stage~III in Algorithm~\ref{al:pipeline} applies GES to resampled latents drawn from an estimated law rather than the unobserved $p^\star$, our setting is more complex than \cite{chickering2002optimal}, where the data are drawn directly from a fixed $p^\star$. Classical local consistency must therefore be strengthened to tolerate sampling from a sequence $\{p_N\}$ approaching $p^\star$ at a specified rate. This refinement is crucial for applying GES to $\hat{\MZ}$ rather than the unobserved $\MZ$. We now introduce the rate-robust notions needed for this purpose.

\begin{definition}\label{def:c_N local consistency}
Let $p^\star=P_{\Mtheta^\star}$ and $\calD_N=\{\MR_{N,1},\dots,\MR_{N,N}\}$ be i.i.d.\ from some $p_N=P_{\Mtheta_N}$.  Let $G'$ be obtained from $G$ by adding the edge $i\to j$. We say a score $S$ is $\{c_N\}$-locally consistent if $\|\Mtheta_N-\Mtheta^\star\|=O_p(1/c_N)$ with $c_N\to\infty$, and as $N\to\infty$ the following hold:
(i) if $R_j\not\!\perp\!\!\!\perp_{p^\star}R_i\mid\mathbf{Pa}_j^{G}$, then $S(G',\calD_N)>S(G,\calD_N)$ with probability $\to1$;
(ii) if $R_j\!\perp\!\!\!\perp_{p^\star}R_i\mid\mathbf{Pa}_j^{G}$, then $S(G',\calD_N)<S(G,\calD_N)$ with probability $\to1$.
\end{definition}
If a score-equivalent score criterion $S$ is $\{c_N\}$-locally consistent and the sampling law $\p_N$ used for resampling remains within the prescribed $\{c_N\}$-local tolerance of $\p^\star$, then it can be shown that applying GES with $S$ to samples $\calD_N$ drawn from $\p_N$ still returns the MEC of $G^\star$ as $N$ increases. Our first task is therefore to identify a $\{c_N\}$-locally consistent criterion; a standard choice already does so. The next theorem places BDeu in our $\{c_N\}$ framework. 
\begin{theorem}\label{cor: local consistency}
The BDeu score is $\{c_N\}$-locally consistent for discrete causal graphical models, where $c_N=\omega(\sqrt{\frac{N}{\log N}})$.
\end{theorem}

Leveraging this key property, we apply GES with the BDeu score to the resampled latent data. Concretely, in our pipeline we first use the \(N\) observed samples to construct an estimator $\widetilde\Mtheta_N$ of the true \(\Mtheta^\star\), then draw \(f(N)\) samples from $P_{\widetilde\Mtheta_N}$ and run GES scored by BDeu on these resamples. For ease of further discussion, assume the estimator obeys the rate
$\|\widetilde\Mtheta_N-\Mtheta^\star\| = O_p(1/g(N))$,
where \(g:\ZZ_{+}\to\RR\) is strictly increasing with \(g(N)\to\infty\).

Under the indexing convention of Definitions~\ref{def:c_N local consistency}, the distribution generating \(\calD_N\) should be labeled by the sample size that produced it; in particular, it is an estimate from $f^{-1}(N)$ samples rather than from $N$. Equivalently, $\widetilde\Mtheta_{f^{-1}(N)}=\Mtheta_N$, so $\|\Mtheta_N-\Mtheta^\star\|=O_p(1/g(f^{-1}(N))).$
Invoking Theorem~\ref{cor: local consistency}, it thus suffices that
$g\big(f^{-1}(N)\big)=\omega(\sqrt{\tfrac{N}{\log N}})$
to guarantee that GES on the new data recovers the correct MEC.
Recall that Theorem~\ref{thm:estimation consistency} gives $\|\hat\p-\p^\star\|=O_p(N^{-1/2})$. Combining this \(N^{-1/2}\) rate with $g\big(f^{-1}(N)\big)=\omega(\sqrt{\tfrac{N}{\log N}})$
yields \(f^{-1}(N)=\omega(N/\log N)\), which holds whenever $f(N)=o( N \log N)$. This determines an appropriate sampling rule $f(N)$ guaranteeing consistency for the causal graph.


\section{Simulation Studies}\label{sec:sim}
We conduct simulation studies across diverse settings to evaluate the proposed pipeline (Algorithm~\ref{al:pipeline}), using Algorithm~\ref{al: PSAEM} in Stage I. 
We begin with the comparison experiments and the generative setup shared across all simulations. These settings are chosen to satisfy the strict identifiability conditions in Theorem~\ref{thm1} and create scenarios of varying difficulty for structure recovery. In all simulations, we use three measurement families—Gaussian, Poisson, and Bernoulli—to represent continuous, count, and binary observations, respectively. 
The exact constructions of two possible measurement graphs $\MQ_1$, $\MQ_2$, and parameters $\p$, $(\MB,\bgamma)$ are given in Supplement~\ref{sec:sim-details}. For the resampling step, we fix $f(N)=N$ in all main-text experiments, while Supplement~\ref{sec: choice of fn} reports a sensitivity analysis over $f(N)\in\{2N,3N\}$.

\paragraph{Simulation Study I: Comparison experiments.}~
We benchmark our method against the influential work of \cite{kivva2021learning}, which proposed a theoretically well-founded framework for learning DAGs with discrete latent variables. Its publicly available implementation currently supports latent dimension only up to \(K \le 5\), so in this comparison we focus on Gaussian models with \(K=5\). 
For these comparison runs we set \(\MQ=\MQ_1\). We adopt three standard benchmark latent DAGs with \(K=5\) shown in Figure~\ref{Different DAGs for Comparison Experiments}. 

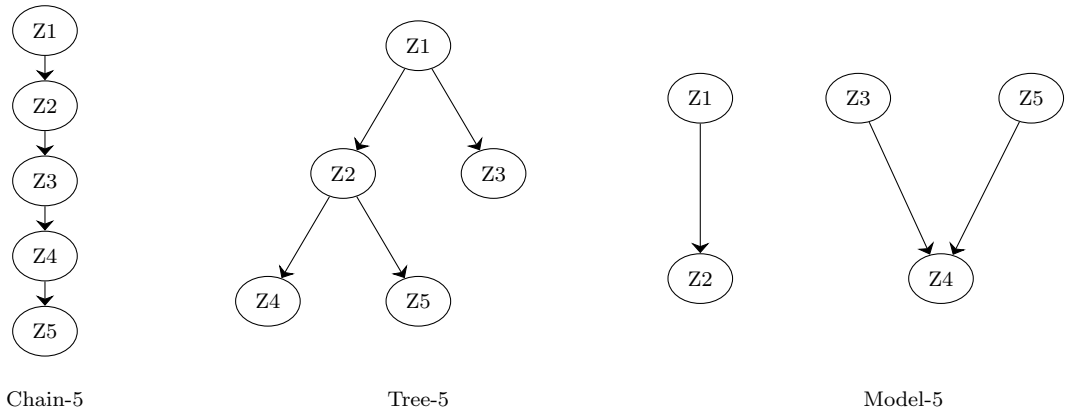
\begin{figure}[h!]
\centering

\begin{tikzpicture}[
  var/.style={draw, ellipse, minimum width=0.6cm, minimum height=0.6cm},font=\scriptsize,>={Stealth[length=4pt,width=7pt]}
]
  \node[var] (Z1) at (0,4) {Z1};
  \node[var] (Z2) at (0,3) {Z2};
  \node[var] (Z3) at (0,2) {Z3};
  \node[var] (Z4) at (0,1) {Z4};
  \node[var] (Z5) at (0,0) {Z5};
  \draw[->] (Z1) -- (Z2);
  \draw[->] (Z2) -- (Z3);
  \draw[->] (Z3) -- (Z4);
  \draw[->] (Z4) -- (Z5);
  \node[draw=none] at (0,-0.9) {Chain-5};
\end{tikzpicture}
\hspace{1.6cm}
\begin{tikzpicture}[
  var/.style={draw, ellipse, minimum width=0.6cm, minimum height=0.6cm},font=\scriptsize,>={Stealth[length=4pt,width=7pt]}
]
  \node[var] (Z1) at (0,5.5) {Z1};
  \node[var] (Z2) at (-1,3.8) {Z2};
  \node[var] (Z3) at (1,3.8) {Z3};
  \node[var] (Z4) at (-2,2.1) {Z4};
  \node[var] (Z5) at (0,2.1) {Z5};
  \draw[->] (Z1) -- (Z2);
  \draw[->] (Z1) -- (Z3);
  \draw[->] (Z2) -- (Z4);
  \draw[->] (Z2) -- (Z5);
  \node[draw=none] at (0,0.8) {Tree-5};
\end{tikzpicture}
\hspace{1.6cm}
\begin{tikzpicture}[
  var/.style={draw, ellipse, minimum width=0.6cm, minimum height=0.6cm},font=\scriptsize,>={Stealth[length=4pt,width=7pt]}
]
  \node[var] (Z1) at (0,4.8) {Z1};
  \node[var] (Z2) at (0,2.4) {Z2};
  \draw[->] (Z1) -- (Z2);

  \node[var] (Z3) at (2.1,4.8) {Z3};
  \node[var] (Z5) at (4.4,4.8) {Z5};
  \node[var] (Z4) at (3.2,2.4) {Z4};
  \draw[->] (Z3) -- (Z4);
  \draw[->] (Z5) -- (Z4);

  \node[draw=none] at (2.7,0.8) {Model-5};
\end{tikzpicture}
\caption{Simulation benchmarks for the comparison experiments}
\label{Different DAGs for Comparison Experiments}
\end{figure}

We fix \(f(N)=N\) with \(N\in\{1000,5000,10000\}\). Structural Hamming Distance (SHD) is computed on the full composite graph obtained by combining the latent DAG \(\calG\) with the bipartite measurement graph \(\MQ\). Since the composite graph contains both edges between the latent variables and edges in the bipartite graph, the absolute SHD values can be numerically large, yet the comparison remains fair because both methods are evaluated on the same composite graph. Table~\ref{tab:SHD-MT8} summarizes results from 100 independent replicates and shows that our method substantially outperforms \cite{kivva2021learning}. 
\begin{table}[ht]
\centering
\begin{tabular}{cccccccc}
\toprule
 & \multicolumn{3}{c}{Proposed DCRL} & \multicolumn{3}{c}{Mixture-Oracle}  \\
\cmidrule(lr){2-4}\cmidrule(lr){5-7}
 & 1000 & 5000 & 10000 & 1000 & 5000 & 10000 \\
\midrule
Chain-5  & 0.38 & 0.02 & 0 & 23.66 & 22.68 & 21.7 \\
Tree-5   & 0.47 & 0.07 & 0 & 23.16 & 22.3  & 21.94 \\
Model-5  & 1.42 & 0.06 & 0 & 22.6  & 22.38 & 22.37 \\
\bottomrule
\end{tabular}
\caption{SHD on the composite graph \(\calG\cup\MQ\) under two methods with \(f(N)=N\);\\ penalized Gibbs--SAEM attains far smaller SHD across all settings and improves to near-perfect recovery as \(N\) grows.}
\label{tab:SHD-MT8}
\end{table}

To interpret these results, recall that \citet{kivva2021learning} formulate their identifiability theory at a high level of generality: they do not assume a specific likelihood, nor fix $K$ or the number of categories in advance. Instead, such information is, in principle, recoverable from the geometry of the mixture distribution over \(\MX\). In practice, this is implemented via a mixture oracle that is approximated by clustering on the observed space, effectively treating each latent configuration as a separate mixture component and thus requiring the recovery of up to \(2^K\) clusters on the full \(\MX\). For moderate or large \(K\), especially when the mixture components are only weakly separated, this reliance on clustering becomes statistically fragile. In addition to these statistical challenges, Theorem~4.8 of \citet{kivva2021learning} shows that even the subroutine for recovering the bipartite graph \(\Gamma\) has complexity \(O(N^4)\), further limiting the applicability of their implementation to relatively small sample sizes and latent dimensions, consistent with the constraint \(K \le 5\) in the released code.

By contrast, we work within a more structured but still flexible framework: we posit a sparse measurement graph and an explicit item-wise likelihood for \(\MX\mid \MZ\). This additional structure allows us to exploit the likelihood via a penalized SAEM algorithm, which leads to substantially more accurate estimation in the weakly separated regimes considered here.

\paragraph{Simulation Study II: Larger $K$ and more challenging settings.}~
We next consider larger latent dimensions and a broader collection of latent DAGs, illustrated in Figure~\ref{Different DAGs}. The five benchmark latent DAGs in Figure~\ref{Different DAGs} (Chain-10, Tree-10, Model-7, Model-8, and Model-13) have latent dimensions \(K = 10, 10, 7, 8,\) and \(13\), respectively. The latent dimension \(K\) here is substantially larger than in \cite{kivva2021learning, huang2022latent, prashant2024differentiable}, adding to the difficulty of accurate recovery. All three distributional types (Gaussian, Poisson, Bernoulli) and both $\MQ_1$ and $\MQ_2$ are considered.

\begin{figure}[h!]
\centering
\resizebox{\textwidth}{!}{
\begin{tikzpicture}[
  var/.style={draw, ellipse, minimum width=0.6cm, minimum height=0.5cm},font=\scriptsize,
  line width=0.5pt
]

\begin{scope}[shift={(-7,0)},>={Stealth[length=3pt,width=7pt]}]
  \node[var] (C1)  at (0,4.5) {Z1};
  \node[var] (C2)  at (0,3.7) {Z2};
  \node[var] (C3)  at (0,2.9) {Z3};
  \node[var] (C4)  at (0,2.1) {Z4};
  \node[var] (C5)  at (0,1.3) {Z5};
  \node[var] (C6)  at (0,0.5) {Z6};
  \node[var] (C7)  at (0,-0.3) {Z7};
  \node[var] (C8)  at (0,-1.1) {Z8};
  \node[var] (C9)  at (0,-1.9) {Z9};
  \node[var] (C10) at (0,-2.7){Z10};

  \foreach \i/\j in {1/2,2/3,3/4,4/5,5/6,6/7,7/8,8/9,9/10}
    {\draw[->] (C\i) -- (C\j);}

  \node[draw=none, anchor=north] at (0,-3.5) {\small Chain-10};
\end{scope}

\begin{scope}[shift={(-2,0)},>={Stealth[length=4pt,width=7pt]}]
  \node[var] (T1)  at (0,4.5)  {Z1};
  \node[var] (T2)  at (-1.2,3.4){Z2};
  \node[var] (T3)  at ( 1.2,3.4){Z3};
  \node[var] (T4)  at (-1.8,2.3){Z4};
  \node[var] (T5)  at (-0.6,2.3){Z5};
  \node[var] (T6)  at ( 0.6,2.3){Z6};
  \node[var] (T7)  at ( 1.8,2.3){Z7};
  \node[var] (T8)  at (-2.4,1.2){Z8};
  \node[var] (T9)  at (-1.2,1.2){Z9};
  \node[var] (T10) at ( 0.0,1.2){Z10};

  \draw[->] (T1) -- (T2);
  \draw[->] (T1) -- (T3);
  \draw[->] (T2) -- (T4);
  \draw[->] (T2) -- (T5);
  \draw[->] (T3) -- (T6);
  \draw[->] (T3) -- (T7);
  \draw[->] (T4) -- (T8);
  \draw[->] (T4) -- (T9);
  \draw[->] (T5) -- (T10);

  \node[draw=none, anchor=north] at (0,0.8) {\small Tree-10};
\end{scope}

\begin{scope}[shift={(3,0)},>={Stealth[length=4pt,width=7pt]}]
  \node[var] (M81) at (0,4.5)   {Z1};
  \node[var] (M82) at (-1.2,3.4){Z2};
  \node[var] (M83) at ( 1.2,3.4){Z3};
  \node[var] (M84) at (-1.8,2.3){Z4};
  \node[var] (M85) at (-0.6,2.3){Z5};
  \node[var] (M86) at ( 0.6,2.3){Z6};
  \node[var] (M87) at ( 1.8,2.3){Z7};
  \node[var] (M88) at ( 0,1.2){Z8};

  \draw[->] (M81) -- (M82);
  \draw[->] (M81) -- (M83);
  \draw[->] (M82) -- (M84);
  \draw[->] (M82) -- (M85);
  \draw[->] (M83) -- (M86);
  \draw[->] (M83) -- (M87);
  \draw[->] (M85) -- (M88);
  \draw[->] (M86) -- (M88);

  \node[draw=none, anchor=north] at (0,0.8) {\small Model-8};
\end{scope}

\begin{scope}[shift={(0,-4)},>={Stealth[length=4pt,width=7pt]}]
  \node[var] (M131) at (0,4.5)   {Z1};
  \node[var] (M132) at (-1.8,3.4){Z2};
  \node[var] (M133) at ( 1.8,3.4){Z3};
  \node[var] (M134) at (-2.5,2.3){Z4};
  \node[var] (M135) at (-1.1,2.3){Z5};
  \node[var] (M136) at ( 1.1,2.3){Z6};
  \node[var] (M137) at ( 2.5,2.3){Z7};
  \node[var] (M138) at (-3.2,1.2){Z8};
  \node[var] (M139) at (-1.8,1.2){Z9};
  \node[var] (M1310)at (-0.6,1.2){Z10};
  \node[var] (M1311)at ( 0.6,1.2){Z11};
  \node[var] (M1312)at ( 1.8,1.2){Z12};
  \node[var] (M1313)at ( 3.2,1.2){Z13};

  \draw[->] (M131) -- (M132);
  \draw[->] (M131) -- (M133);
  \draw[->] (M132) -- (M134);
  \draw[->] (M132) -- (M135);
  \draw[->] (M133) -- (M136);
  \draw[->] (M133) -- (M137);

  \draw[->] (M134)  -- (M138);
  \draw[->] (M134)  -- (M139);
  \draw[->] (M134)  -- (M1310);
  \draw[->] (M135)  -- (M138);
  \draw[->] (M135)  -- (M139);
  \draw[->] (M135)  -- (M1310);
  \draw[->] (M136)  -- (M1311);
  \draw[->] (M136)  -- (M1312);
  \draw[->] (M136)  -- (M1313);
  \draw[->] (M137)  -- (M1311);
  \draw[->] (M137)  -- (M1312);
  \draw[->] (M137)  -- (M1313);

  \node[draw=none, anchor=north] at (0,0.4) {\small Model-13};
\end{scope}

\begin{scope}[shift={(8,0)},>={Stealth[length=4pt,width=7pt]}]
  \node[var] (M71) at (0,4.5)   {Z1};
  \node[var] (M72) at (-1.2,2.4){Z2};
  \node[var] (M73) at ( 1.2,2.4){Z3};
  \node[var] (M74) at (-1.2,0.3){Z4};
  \node[var] (M75) at ( 1.2,0.3){Z5};
  \node[var] (M76) at (-1.2,-1.8){Z6};
  \node[var] (M77) at ( 1.2,-1.8){Z7};

  \draw[->] (M71) -- (M72);
  \draw[->] (M71) -- (M73);
  \draw[->] (M72) -- (M74);
  \draw[->] (M72) -- (M75);
  \draw[->] (M73) -- (M75);
  \draw[->] (M73) -- (M74);
  \draw[->] (M74) -- (M76);
  \draw[->] (M75) -- (M77);

  \node[draw=none, anchor=north] at (0,-2.5) {\small Model-7};
\end{scope}

\end{tikzpicture}}
\caption{True latent DAGs for Simulation Study II.}
\label{Different DAGs}
\end{figure}
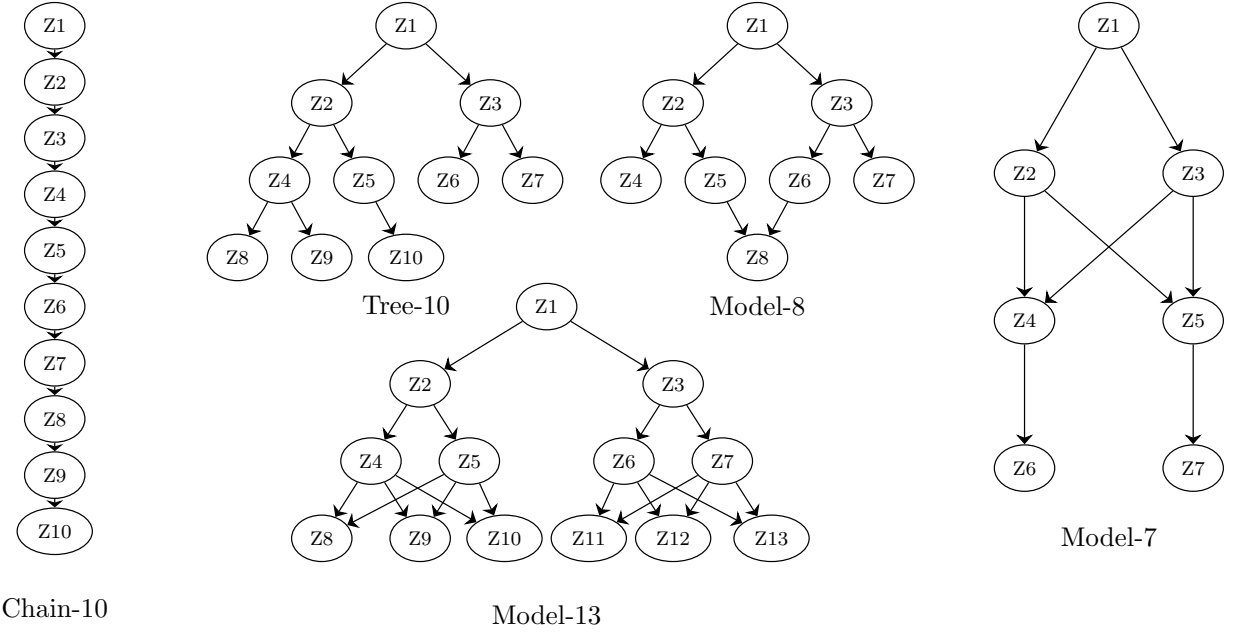

\begin{table}[h!]
\centering  \footnotesize
\setlength{\tabcolsep}{4pt} 
\begin{tabular}{llcccccccccc}
\toprule
& & & \multicolumn{3}{c}{Bernoulli} & \multicolumn{3}{c}{Poisson} & \multicolumn{3}{c}{Gaussian} \\
\cmidrule(lr){4-6}\cmidrule(lr){7-9}\cmidrule(lr){10-12}
& & & \multicolumn{3}{c}{$N$} & \multicolumn{3}{c}{$N$} & \multicolumn{3}{c}{$N$} \\
Model & $\MQ$ & {\scriptsize \#edges} & 3000 & 5000 & 7000 & 3000 & 5000 & 7000 & 3000 & 5000 & 7000 \\
\midrule
\multirow{2}{*}{Chain-10} 
  & $\MQ_1$ & {\scriptsize 57} & 5.55  & 3.294 & 2.938 & 2.248 & 1.362 & 0.638 & 0.412 & 0.254 & 0.122 \\
  & $\MQ_2$ & {\scriptsize 73} & 5.664 & 3.258 & 3.042 & 2.174 & 0.704 & 0.488 & 0.406 & 0.208 & 0.146 \\
\addlinespace
\multirow{2}{*}{Tree-10} 
  & $\MQ_1$ & {\scriptsize 57} & 4.372 & 2.308 & 1.308 & 1.85  & 1.692 & 1.36  & 0.94  & 0.51  & 0.308 \\
  & $\MQ_2$ & {\scriptsize 73} & 4.3   & 1.936 & 1.26  & 2.676 & 1.556 & 1.146 & 1.14  & 0.618 & 0.346 \\
\addlinespace
\multirow{2}{*}{Model-7} 
  & $\MQ_1$ & {\scriptsize 41} & 7.692 & 6.334 & 5.798 & 5.838 & 5.422 & 4.892 & 0.196 & 0.042 & 0     \\
  & $\MQ_2$ & {\scriptsize 51} & 7.554 & 6.304 & 5.68  & 5.848 & 5.526 & 5.082 & 0.262 & 0.054 & 0.004 \\
\addlinespace
\multirow{2}{*}{Model-8} 
  & $\MQ_1$ & {\scriptsize 46} & 4.336 & 2.682 & 2.19  & 2.106 & 1.916 & 1.878 & 0.132 & 0.048 & 0     \\
  & $\MQ_2$ & {\scriptsize 58} & 4.342 & 2.72  & 2.374 & 2.264 & 1.9   & 1.782 & 0.202 & 0.052 & 0.002 \\
\addlinespace
\multirow{2}{*}{Model-13} 
  & $\MQ_1$ & {\scriptsize 81}  & 22.37  & 16.454 & 14.062 & 24.65  & 16.472 & 14.162 & 3.206 & 1.646 & 1.008 \\
  & $\MQ_2$ & {\scriptsize 103} & 22.252 & 16.29  & 14.606 & 25.134 & 15.626 & 12.934 & 3.032 & 1.872 & 0.994 \\
\bottomrule
\end{tabular}
\caption{Average SHD on the composite graph $\calG\cup\MQ$ with $f(N)=N$; SHD decreases with $N$ across all designs and is smallest for Gaussian, intermediate for Poisson, and largest for Bernoulli.
The column $\#\text{edges}$ reports the number of edges in the composite graph.}
\label{tab:SHD-MT7}
\end{table}


We conduct 500 independent replicates in each simulation setting, and report in Table~\ref{tab:SHD-MT7} the average SHD computed on the full composite graph \(\calG \cup \MQ\). Table~\ref{tab:SHD-MT7} reveals a clear pattern: Bernoulli data are the most challenging, followed by Poisson, and Gaussian is the easiest. This is consistent with the intuition that discrete observations carry less information, which is why most existing simulation studies focus on the continuous Gaussian case. Although some of the reported SHD values may look large in absolute terms, they should be interpreted relative to the size of the underlying graph: the composite object contains not only the latent DAG \(\calG\) but also the bipartite layer induced by \(\MQ\), which contributes a substantial number of edges. For instance, in the Model-13 design, the target structure consists of a DAG on 13 latent variables together with a bipartite graph between 13 latent and 39 observed variables; when \(\MQ=\MQ_1\) the resulting composite graph already contains 81 edges, and when \(\MQ=\MQ_2\) it contains 103 edges. Consequently, even a small relative error can translate into a seemingly large SHD. Viewed in this light, the results in Table~\ref{tab:SHD-MT7} already indicate accurate recovery across all three measurement families. Moreover, within each simulation configuration, the average SHD decreases systematically as the sample size \(N\) increases, which provides empirical support for our identifiability theory.

\vspace{-5mm}
\section{Applications to Educational Data and Image Data}\label{sec:real}
We evaluate our method on an educational assessment dataset and a synthetic ball-image dataset. In the educational dataset, we examine whether the recovered causal structure among latent cognitive skills is interpretable. The image dataset is a high-dimensional benchmark with a known generative process and ground-truth latent DAG, inspired by ``balls'' image setups in causal representation learning \citep{ahuja2023interventional,ahuja2024multi}.
The image data allows us to assess whether our method simultaneously recovers the true causal relationships and learns interpretable latent representations from high-dimensional observations.


\subsection{TIMSS 2019 Response Time Data}\label{sec: TIMSS}
We apply DCRL to response time data from the TIMSS 2019 eighth-grade mathematics assessment for U.S. students \citep{fishbein2021TIMSS}, recording each student's time (in seconds) on each item screen. The assessment evaluates seven skills: four content skills (``Number'', ``Algebra'', ``Geometry'', ``Data and Probability'') and three cognitive skills (``Knowing'', ``Applying'', ``Reasoning''). 
We follow the preprocessing steps of \cite{lee2024new}, but pursue the more demanding goal of recovering causal relationships among the latent skills.

We consider students who received booklet 14. Following \cite{lee2024new}, we log-transform and truncate response times to mitigate outlier influence, yielding a dataset of $N = 620$ students and $J = 29$ items. We fit this dataset under a lognormal observation specification within the DCRL framework.

Since the TIMSS 2019 database already specifies which skills are assessed by each item, the skill-item association matrix $\MQ$ is known and can be directly constructed as a $29 \times 7$ binary matrix (see Table~\ref{tab:TIMSS}). Notably, this matrix satisfies the identifiability conditions in Corollary~\ref{cor2}, ensuring that the latent causal structure can be uniquely recovered up to label permutations and Markov equivalence. Therefore, 
we accordingly make minor modifications to our previous algorithm: we eliminate the sparsity-inducing penalty term and restrict each latent variable to be connected only by those items that are designed to measure it, as determined by the known matrix $\MQ$. 
Moreover, the screen-level response time matrix contains some missing entries. 
Similar to \cite{lee2025deep}, we treat all such entries as missing at random.
Algorithmically, objective functions are computed by summing only over person–item cells with observed response times. 

With these modifications, the algorithm outputs the causal structure as a CPDAG (Figure~\ref{TIMSS}), since the DAG is identifiable only up to its Markov equivalence class.

\begin{figure}[ht]
\centering
\resizebox{0.85\textwidth}{!}{
\begin{tikzpicture}[
  node distance = 9mm and 18mm,
  >={Stealth[length=2.2mm,width=2.2mm]},
  every node/.style={font=\small},
  var/.style={ellipse,draw,thick,minimum height=12mm,minimum width=4mm,align=center},
  content/.style={var,fill=blue!6},
  cognitive/.style={var,fill=orange!9},
  dir/.style={-Latex,very thick},
  undir/.style={-,line width=1.2pt} 
]
\node[content] (X1) at (3,4.2){Number};
\node[content] (X2) at (-3,2.2) {Algebra};
\node[content] (X3) at (3,0.2) {Geometry};
\node[content] (X4) at (1,2.2) {Data \& Probability};

\node[cognitive] (X5) at (5,2.2) {Knowing};
\node[cognitive] (X6) at (8,2.2) {Applying};
\node[cognitive] (X7) at (11,2.2) {Reasoning};

\draw[undir] (X2) -- (X3);
\draw[undir] (X3) -- (X4);
\draw[dir] (X4) -- (X1);
\draw[dir] (X2) -- (X1);
\draw[undir] (X5) -- (X6);
\draw[undir] (X6) -- (X7);
\draw[undir] (X4) -- (X5);

\node[draw,rounded corners,inner sep=1.5mm,anchor=north west,font=\scriptsize] at ($(X6)+(1,-1.2)$)  {%
  \begin{tabular}{@{}l@{}}
  \tikz{\node[content,minimum width=6mm,minimum height=3mm] {};} content skills \\
  \tikz{\node[cognitive,minimum width=6mm,minimum height=3mm] {};} cognitive skills 
  \end{tabular}
};

\end{tikzpicture}}
\caption{Learned Causal Relationships Among Seven Latent Skills}\label{TIMSS}
\end{figure}
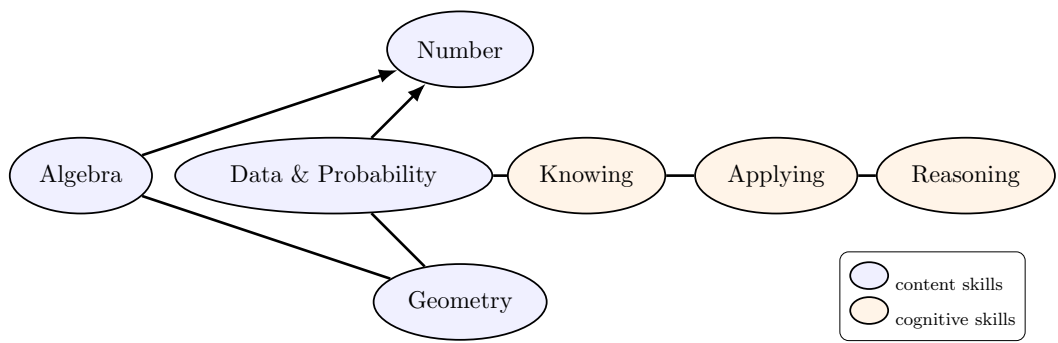
The recovered causal graph aligns with cognitive expectations and curriculum structure. The three cognitive skills, ``Knowing", ``Applying", and ``Reasoning", form a directed path, consistent with the progressive nature of cognitive processing. Among the content skills, ``Number" is foundational and shows strong links to both ``Algebra" and ``Data and Probability", reflecting their shared reliance on numerical reasoning. Although ``Geometry" typically requires less numerical computation, it remains connected to both ``Algebra" and ``Data and Probability", suggesting overlapping problem-solving strategies. Among these four content skills, ``Data and Probability" is likely the most comprehensive, as it requires a broad range of skills, making it directly linked to the three cognitive skills. This supports the interpretation that tasks involving data interpretation demand a combination of factual knowledge, application, and reasoning, positioning it as an integrative skill in mathematical cognition.

\subsection{Ball Image Data}
We build a  ``seesaw + occlusion'' experiment in which each latent variable 
represents a visibility, presence, or configuration event in the observed image.
Two binary variables indicate the presence of balls on two well-separated slots of a forked tray, which acts as a load on the right side of a seesaw.
Let $Z_1,Z_2\in\{0,1\}$ denote the presence indicators of these two tray balls.
The seesaw's mechanical response determines whether the left-side ball rises to an ``up'' configuration, denoted $Z_3\in\{0,1\}$.
To reflect natural heterogeneity in physical conditions (e.g., slight variations in ball weights or instrument failures), we model this mechanism stochastically rather than deterministically by setting $\mathbb P(Z_3=1\mid Z_1=1,Z_2=1)=0.8$ and $\mathbb P(Z_3=1\mid (Z_1,Z_2)\neq(1,1))=0.2$,
so that the tray balls increase the probability of the ``rise'' event without forcing it.

Finally, we introduce a fourth ball that is physically present but typically occluded by the left ball when the seesaw is in the down configuration.
When the seesaw rises ($Z_3=1$), the fourth ball may become visible; we define $Z_4\in\{0,1\}$ as its visibility indicator, with $\mathbb P(Z_4=1\mid Z_3=1)=0.99$ and $\mathbb P(Z_4=1\mid Z_3=0)=0$.
Overall, this construction yields a physically motivated latent causal structure $Z_1,Z_2 \to Z_3 \to Z_4$ while keeping the latent variables binary for a principled reason:
each $Z_k$ corresponds to a discrete, image-level event (presence, configuration, or visibility).

Figure~\ref{fig:seesaw_examples} shows representative images; full rendering and preprocessing details are in Supplement~\ref{app:seesaw_generation}. Each rendered grayscale image is converted to a balls-only binary mask, resized to $96\times 96$, and then pooled to a $16\times 16$ binary image. We fit the model using this pooled representation, so each sample is a $256$-dimensional binary vector, where $1$ denotes a bright pixel and $0$ denotes a dark pixel. We generate $10000$ images.

\begin{figure}[h!] 
\begin{minipage}{0.99\linewidth} \vspace{3pt} \centerline{\includegraphics[width=\textwidth]{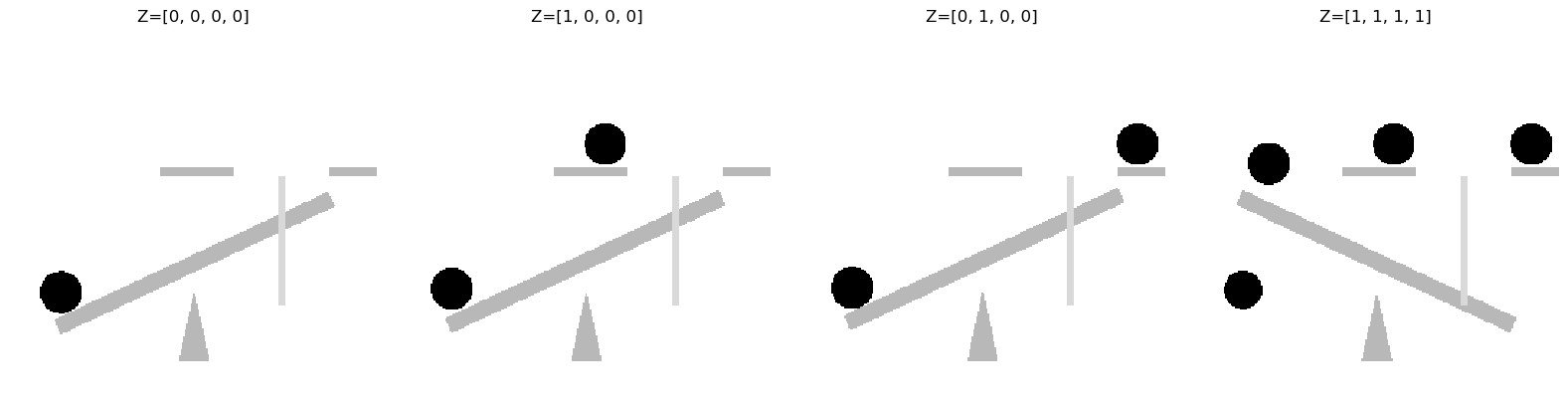}} \end{minipage}
\caption{Representative samples from the ``seesaw + occlusion'' image generator. The tray balls correspond to $(Z_1,Z_2)$, the up/down state of the left seesaw-side ball corresponds to $Z_3$, and the occluded ball corresponds to $Z_4$.}
\label{fig:seesaw_examples}
\end{figure}

We fit DCRL with $K=4$ binary latent variables and Bernoulli responses, a high-dimensional setting with $J=256$ observed dimensions substantially larger than in our simulation studies.
The observed variables are modeled as: $X_j\mid\MZ\sim\text{Ber}(g_{\text{logistic}}(\beta_{j,0}+\sum_{k=1}^K\beta_{j,k}Z_k)),$ where $g_{\text{logistic}}$ is the sigmoid function. Our goal is to recover both the latent DAG $G$ and the bipartite structure $\MQ$ from the pooled binary data. Although $\MQ$ has $256$ rows, the estimated $\widehat{\MQ}$ remains highly sparse. Most rows are either zero vectors (blocks where no ball appears) or nearly one-hot vectors (blocks primarily associated with a single latent variable), which matches the generative design in which most spatial cells contain at most one object. The overall sparse support pattern also satisfies the identifiability conditions in Corollary~\ref{cor2}.
The recovered DAG over the latent variables in Figure~\ref{Causal Relationships Among Four Latent Variables} matches the data-generating mechanism.

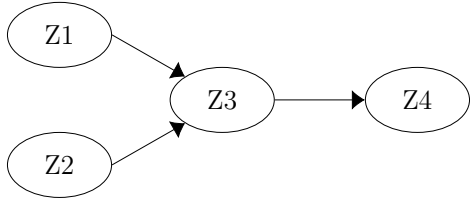
\begin{figure}[h!]
\centering
\resizebox{0.38\textwidth}{!}{
\begin{tikzpicture}[
  var/.style={draw, ellipse, minimum width=1.6cm, minimum height=1cm}
]

  \node[var] (Z1) at (-0.5,2)  {Z1};
  \node[var] (Z2) at (-0.5,0)  {Z2};
  \node[var] (Z3) at (2,1)  {Z3};
  \node[var] (Z4) at (5,1)  {Z4};

  \draw[->] (Z1.east) -- (Z3.north west);
  \draw[->] (Z2.east) -- (Z3.south west);
  \draw[->] (Z3.east) -- (Z4.west);
\end{tikzpicture}}
\caption{Estimated causal relationships among four latent variables by DCRL. This latent DAG matches the ground-truth causal relations exactly.}
\label{Causal Relationships Among Four Latent Variables}
\end{figure}

\begin{figure}[h!]\centering
\begin{minipage}{0.23\linewidth}
\vspace{3pt}
\centerline{\includegraphics[height=4cm]{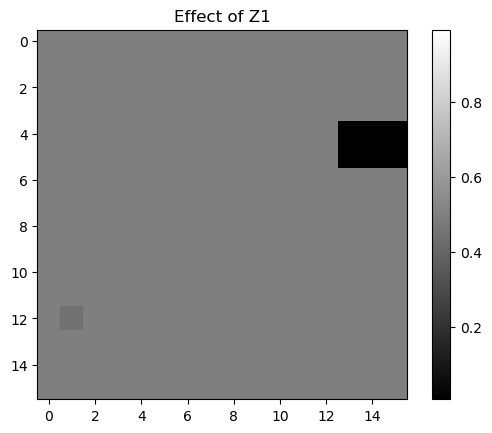}}
\centerline{$g_{\text{logistic}}(\hat\MB\be_{1})$}
\end{minipage}
\begin{minipage}{0.23\linewidth}
\vspace{3pt}
\centerline{\includegraphics[height=4cm]{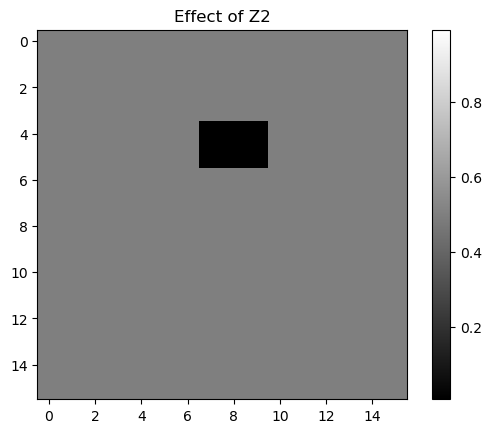}}
\centerline{$g_{\text{logistic}}(\hat\MB\be_{2})$}
\end{minipage}
\begin{minipage}{0.23\linewidth}
\vspace{3pt}
\centerline{\includegraphics[height=4cm]{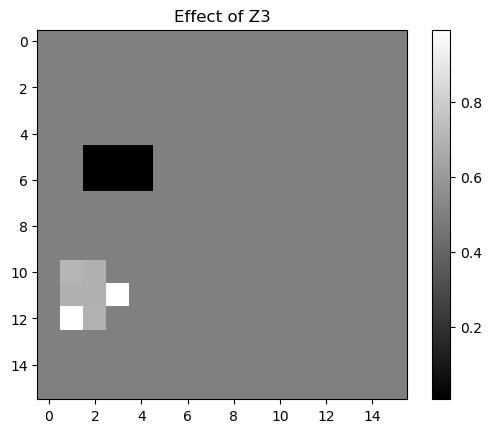}}
\centerline{$g_{\text{logistic}}(\hat\MB\be_{3})$}
\end{minipage}
\begin{minipage}{0.23\linewidth}
\vspace{3pt}
\centerline{\includegraphics[height=4cm]{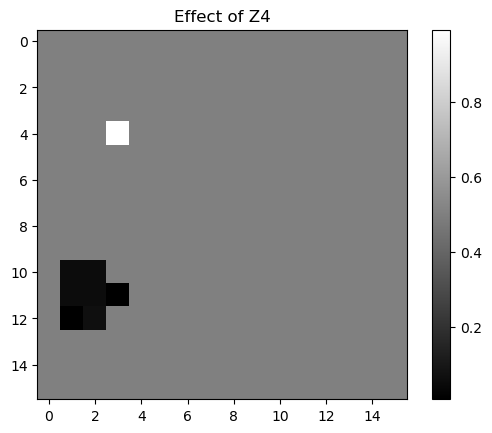}}
\centerline{$g_{\text{logistic}}(\hat\MB\be_{4})$}
\end{minipage}
\caption{Effect maps obtained by activating one latent coordinate at a time. Mid-gray corresponds to probability $0.5$ (no effect). Since the pooled representation is coded as background $=1$ and ball $=0$, white indicates an increased probability of background (ball less likely), while black indicates a decreased probability of background (ball more likely). The dominant localized regions align with the tray balls, the configuration-dependent movement of the seesaw-side ball, and the occluded ball’s visibility.}
\label{latent factors}
\end{figure}



Since the learned causal structure matches the data-generating mechanism, we expect each recovered latent coordinate to correspond to a spatially localized event (i.e., a tray ball being present, the up/down state of the seesaw-side ball, or the occluded ball becoming visible). To interpret the factors, we visualize the effect of activating one latent coordinate at a time on the pixelwise Bernoulli probabilities in the pooled representation. Recall that $\MB\in\mathbb{R}^{J\times (K+1)}$ stacks the intercept and main-effect parameters across the $J=256$ pixels, so that $\MB\bz$ gives the logits of the Bernoulli success probabilities for any latent feature vector $\bz\in\mathbb{R}^{K+1}$. For $k\in\{1,2,3,4\}$, we leave out intercept and activate only the $k$th coordinate, compute $g_{\mathrm{logistic}}(\hat\MB\be_k)\in(0,1)^{256}$, and reshape it into a $16\times16$ image, where $\be_k$ is the $k$th standard basis vector. Since $g_{\text{logistic}}(0)=0.5$, mid-gray indicates no effect, and white indicates an increased probability of background (equivalently, a decreased probability that a ball occupies that cell), while black indicates a decreased probability of background.
Figure~\ref{latent factors} reports the resulting effect maps. The tray-ball factors appear as localized dark patches at the corresponding tray locations, indicating that activating those coordinates increases the probability of a ball in those cells. The factor for the seesaw-side ball shows a signed bright/dark pattern, reflecting the ball's up/down state: one location becomes more likely to contain a ball while another becomes less likely. The visibility factor is concentrated near the occluded ball region, with a localized effect consistent with the intended semantics. The maps are not perfectly clean, as expected given the mild positional jitter and simple preprocessing pipeline. Despite these nuisances and the coarse pooled representation, the recovered causal graph and latent factors remain interpretable and align with the data-generating mechanism.

\vspace{-5mm}
\section{Discussion}\label{sec-discussion}
This paper develops a computationally efficient and provable estimate-resample-discovery pipeline for causal representation learning with discrete latent variables. Our procedure has a clean structure: (i) we estimate the measurement layer and the joint distribution of latent variables via a penalized Gibbs–SAEM algorithm, (ii) we resample pseudo-latent datasets from the fitted latent distribution, and (iii) we perform score-based causal discovery on the resampled latents using GES. Theoretically, we establish strict and generic identifiability for the proposed discrete causal representation learning framework. We prove $\{c_N\}$-consistency and $\{c_N\}$-local consistency of BDeu scores in discrete Bayesian networks, and show that, under mild conditions, this estimate-resample-discovery pipeline consistently recovers both the measurement structure and the Markov equivalence class of the latent DAG.

Although our exposition focuses on binary latent variables, the Gibbs--SAEM updates can also be modified for polytomous latent variables. For simplicity and clarity, we focus on the binary case in this paper.
Although we state our results for the BDeu score, the same analysis and guarantees apply to BIC, which appears as the leading term in the BDeu expansion in our proofs. Since most score-based methods in practice adopt either BIC or BDeu \citep{kitson2023a}, this BIC-type class already covers the dominant use cases.

Several avenues remain for future work. 
Incorporating procedures for estimating $K$, such as information criteria tailored to the latent layer, would make the pipeline automatic and reduce sensitivity to model size. Additionally, since our approach provides a general framework, it is natural to explore replacing the current estimation or causal discovery components with other alternatives \citep{spirtes2000causation,ramsey2016a}. These methods' empirical performance and theoretical guarantees remain open for future investigation.






\renewcommand{\thesection}{S.\arabic{section}}  
\renewcommand{\thetable}{S.\arabic{table}}  
\renewcommand{\thefigure}{S.\arabic{figure}}
\renewcommand{\theequation}{S.\arabic{equation}}
\renewcommand{\theassumption}{S.\arabic{assumption}}
\renewcommand{\thetheorem}{S.\arabic{theorem}}
\renewcommand{\thelemma}{S.\arabic{lemma}}

\setcounter{section}{0}
\setcounter{algocf}{0}
\setcounter{equation}{0}
\setcounter{figure}{0}
\setcounter{theorem}{0}
\setcounter{lemma}{0}





\spacingset{1}
{\small\bibliographystyle{apalike}
\bibliography{references}}
\newpage 

\begin{center}
    \LARGE Supplementary Material
\end{center}
\vspace{5mm}
\addcontentsline{toc}{section}{Appendix}
\setcounter{subsection}{0}
\renewcommand{\thesubsection}{S.\arabic{subsection}} 
\renewcommand{\thetable}{S.\arabic{table}}
\renewcommand{\thefigure}{S.\arabic{figure}}
\renewcommand{\theequation}{S.\arabic{equation}}
\renewcommand{\theHequation}{S.\arabic{equation}}
\renewcommand{\thealgocf}{S.\arabic{algocf}}
\renewcommand{\theHtheorem}{S.\arabic{theorem}}
\renewcommand{\thetheorem}{S.\arabic{theorem}}
\renewcommand{\theproposition}{S.\arabic{proposition}}
\renewcommand{\thelemma}{S.\arabic{lemma}}
\renewcommand{\thecorollary}{S.\arabic{corollary}}
\spacingset{1.7}
This Supplementary Material collects technical results, implementation details, and additional empirical summaries. Section~\ref{counterexample-2} presents a non-generic identifiability example. Sections~\ref{sec: proof of identifiability}-\ref{sec: proof of thm5} contain the main proofs for our identifiability and consistency results. Section~\ref{sec: Implementation Details} records implementation details. Section~\ref{sec: Connections with Existing Studies} discusses additional related works.

\medskip
\noindent\textbf{Notation.}
For $d\ge2$, let $\Delta_{d-1}=\{x\in\mathbb{R}^d: x_k\ge0,\ \sum_{k=1}^d x_k=1\}$ denote the $(d-1)$-dimensional probability simplex, and let $\Delta_{d-1}^\circ=\{x\in\Delta_{d-1}: x_k>0\ \text{for all }k\}$ denote its interior.
\subsection{More identifiability results}\label{app:strict_id}
We present additional identifiability results mentioned in Section~\ref{sec: identifiablity}. These results are adapted from existing work \citep{liu2025exploratory,lee2025deep}. Throughout, identifiability is understood up to the equivalence relation $\sim_{\calK}$ defined in Section~\ref{sec: identifiablity}.

\begin{definition}\label{def:strict identifiability}
Let $(\MTheta^\star,\calG^\star,\MQ^\star)\in\Omega_{K}(\MTheta, \mathcal{G}, \MQ)$ be the true parameter triple of the discrete causal representation learning framework.
The framework is strictly identifiable up to the equivalence relation $\sim_{\calK}$ if, for every alternative admissible triple $
(\MTheta',\calG',\MQ')\in\Omega_{K}(\MTheta,\calG,\MQ)$ satisfying the equality of marginal laws $\PP_{\widetilde\MTheta,\widetilde\calG,\widetilde\MQ}=\PP_{\MTheta^\star,\calG^\star,\MQ^\star},$ it necessarily holds that $(\widetilde\MTheta,\widetilde\calG,\widetilde\MQ)\sim_{\calK}(\MTheta^\star,\calG^\star,\MQ^\star).$
Here, $\PP_{\MTheta,\calG,\MQ}$ denotes the marginal distribution of the observable vector $\MX$, defined through \eqref{eq:observed exp fam} and \eqref{eq:marginal probability}.
\end{definition}

Now we show our main results for the identifiability, which is from Proposition 1 in \citep{liu2025exploratory} and the definition of faithfulness.
\begin{theorem}\label{thm1}
Under Assumption~\ref{assm1}, the framework is strictly identifiable if the following hold.
\begin{enumerate}
    \item[(i)] $\MQ^\star$ contains two identity matrices after permuting the rows. Without the loss of generality, suppose that the first $2K$ rows of $\MQ^\star$ are $\big[\MI_{K}, \MI_{K} \big]^\top.$
    \item[(ii)] For any $\bz \neq \bz' \in\{0,1\}^{K}$, there exists $j > 2K$ such that $\eta_j^\star(\bz)\ \neq\ \eta_j^\star(\bz').$
\end{enumerate}
\end{theorem}

When each latent cause affects the observables only through its main effects, without any interaction terms, Assumption~\ref{assm1}(b) could be replaced by a weaker requirement.
\begin{customassumption}\label{assm1'}
\begin{enumerate}
    \item[(a)] $\calG$ is a perfect map of $\p$ and $\p_{\bz}\in(0,1)$ for all $\bz\in\{0,1\}^{K}$.
    \item[(b)]$\sum_{j=1}^J\beta_{j,k}>0$ for $k=1,\dots,K$.
\end{enumerate}
\end{customassumption}

Under a main-effect measurement specification, the condition can be further weakened, which are from Proposition 1, 2 in \cite{lee2025deep} and the definition of faithfulness.
\begin{corollary}\label{cor1}
Suppose the measurement is main-effect only (no interaction terms). Under Assumption~\ref{assm1'}, the framework is strictly identifiable if the following conditions hold.
\begin{enumerate}
    \item[(i)] $\MQ^\star$ contains two identity matrices after permuting the rows. Without the loss of generality, suppose that the first $2K$ rows of $\MQ^\star$ are $\big[\MI_{K}, \MI_{K} \big]^\top.$
    \item[(ii)] For any $\bz \neq \bz' \in\{0,1\}^{K}$, there exists $j > 2K$ such that $\sum_{k=1}^{K} \beta_{j,k}^\star (z_k - z_k^{'}) \neq 0$.
\end{enumerate}
\end{corollary}
\begin{remark}
In many applications of the proposed framework, the number of observed items $J$ is quite large, as is common in modern machine-learning settings. In such regimes, the strict identifiability requirement that $\MQ^\star$ contain two identity blocks is less restrictive than it may initially appear.
\end{remark}
\begin{corollary}\label{cor2}
Suppose the measurement is main-effect only (no interaction terms). Under Assumption~\ref{assm1'} and Assumption~\ref{assm2}, the framework is generic identifiable if the following hold.
\begin{enumerate}
\item[(i)] After a row permutation, we can write $\MQ^\star=[\MQ_1^\top, \MQ_2^\top, \MQ_3^\top]^\top$, where $\MQ_1,\MQ_2\in\{0,1\}^{K\times K}$ have unit diagonals (off–diagonals arbitrary), and $\MQ_3$ has no all-zero column.
\end{enumerate}
\end{corollary}

\subsection{Non-generic identifiability if Condition~(ii) in Theorem~\ref{thm2} is violated}\label{counterexample-2}

In this subsection we construct a concrete counterexample showing that Condition~(ii) in Theorem~\ref{thm2} is indispensable. The example is chosen so that Assumption~\ref{assm1} and Assumption~\ref{assm2} hold, and Condition~(i) of Theorem~\ref{thm2} is satisfied. The only assumption we deliberately violate is Condition~(ii). Nevertheless, we exhibit a positive-measure subset of the parameter space on which the framework is not identifiable, so the framework is not generically identifiable.

We consider a one-layer saturated all-effect Bernoulli--logistic model with $K=4$ latent variables and $J=12$ items. In particular, we take $\text{ParFam}_j$ to be the Bernoulli family and $g_j$ to be the logistic link in \eqref{eq:observed exp fam}.

Let
\[
\MQ^\star=\begin{bmatrix}
\MQ_1^{\top},  \MQ_1^{\top}, \MQ_1^{\top}
\end{bmatrix}^{\top}, \qquad 
\MQ_1=\begin{pmatrix}
    1&0&0&0\\
    0&1&0&0\\
    1&0&1&0\\
    0&0&0&1
\end{pmatrix}.
\]
It is straightforward to verify that this measurement design satisfies Condition~(i) in Theorem~\ref{thm2}. However, Condition~(ii) in Theorem~\ref{thm2} fails, since $\MQ^\star_{:,1}\succeq \MQ^\star_{:,3}$. 

Because we work with Bernoulli--logistic responses, Assumption~\ref{assm2} is automatically satisfied. We also fix a strictly positive latent distribution $(p_\bz)_{\bz\in\{0,1\}^4}$ so that Assumption~\ref{assm1}(a) holds.

We now construct an explicit positive-measure subset of the parameter space on which the framework is not identifiable. Define
\[
\widetilde B
=\bigl\{\bbeta\in\mathbb{R}^{16}: 2\beta_3+\beta_{13}=\beta_{0}\bigr\}
 \cup \bigl\{\bbeta\in\mathbb{R}^{16}: 2\beta_3+\beta_{13}=\beta_{1}\bigr\}
 \cup \bigl\{\bbeta\in\mathbb{R}^{16}: 2\beta_3+\beta_{13}=0\bigr\}.
\]
The set $\widetilde B$ is a finite union of proper algebraic varieties and therefore has Lebesgue measure zero in $\mathbb{R}^{16}$.

Index the items as $j=4m+r$ with $m\in\{0,1,2\}$ and $r\in\{1,2,3,4\}$.  
If $r\in\{1,2,4\}$, select $\bbeta_j$ from
\[
\bigl\{
\bbeta\in\mathbb{R}^{16}:\ 
\beta_{S}\ne 0\ \text{if and only if}\ S\subseteq \{r\},\ \beta_{r}>0
\bigr\},
\]
and if $r=3$, select $\bbeta_j$ from
\[
\bigl\{
\bbeta\in\mathbb{R}^{16}:
\beta_{S}\ne 0\ \text{if and only if}\ S\subseteq \{1,3\},
\ \beta_1+\beta_3+\beta_{13}>0,\ \beta_1+\beta_{13}>0,\ \beta_3+\beta_{13}>0
\bigr\}
\setminus\widetilde B,
\]
which has positive Lebesgue measure. Indeed, each inequality describes an open set in $\mathbb{R}^{16}$ and hence has positive relative measure. Removing $\widetilde B$, a measure-zero set, preserves positive relative measure. By construction, all such choices of $\bbeta_j$ satisfy the monotonicity requirement in Assumption~\ref{assm1}(b).

Next, we define a transformed parameterization $(\MB',\p')$ that induces the same marginal distribution of $\MX$ but cannot be obtained from $(\MB,\p)$ through any latent-coordinate permutation.  

For $j=4m+r$ with $m\in\{0,1,2\}$ and $r\in\{1,2,4\}$, set
\[
\beta'_{j,0}=\beta_{j,0}, \qquad
\beta'_{j,r}=\beta_{j,r},
\]
and set all other entries of $\beta'_j$ to zero.  
For items with indices $j=4m+3$ $(m=0,1,2)$, define
\[
\beta'_{j,0}=\beta_{j,0}+\beta_{j,3}, \qquad
\beta'_{j,1}=\beta_{j,1}-\beta_{j,3}, \qquad
\beta'_{j,3}=-\beta_{j,3}, \qquad
\beta'_{j,13}=2\beta_{j,3}+\beta_{j,13},
\]
and again set all remaining entries of $\beta'_j$ to zero.

Define a permutation $\pi$ of the $2^4$ latent states by
\[
\pi(0000)=0010,\quad \pi(1000)=1000,\quad \pi(0100)=0110,\quad \pi(0010)=0000,
\]
\[
\pi(0001)=0011,\quad\pi(1100)=1100,\quad \pi(1010)=1010,\quad \pi(1001)=1001,
\]
\[
\pi(0110)=0100,\quad\pi(0101)=0111,\quad \pi(0011)=0001,\quad \pi(1110)=1110,
\]
\[
\pi(1101)=1101,\quad \pi(1011)=1011,\quad \pi(0111)=0101,\quad \pi(1111)=1111.
\]
Set the transformed mixing weights by $\pi$:
\[
p'_{0000}=p_{0010},\quad p'_{1000}=p_{1000},\quad p'_{0100}=p_{0110},\quad p'_{0010}=p_{0000},
\]
\[
p'_{0001}=p_{0011},\quad p'_{1100}=p_{1100},\quad p'_{1010}=p_{1010},\quad p'_{1001}=p_{1001},
\]
\[
p'_{0110}=p_{0100},\quad p'_{0101}=p_{0111},\quad p'_{0011}=p_{0001},\quad p'_{1110}=p_{1110},
\]
\[
p'_{1101}=p_{1101},\quad p'_{1011}=p_{1011},\quad p'_{0111}=p_{0101},\quad p'_{1111}=p_{1111}.
\]

For all $\bz$ and $j$,
\[
\eta'_j(\bz)=\eta_j\big(\pi(\bz)\big),\qquad
q'_{j,\bz}=\sigma\!\big(\eta'_j(\bz)\big)=q_{j,\pi(\bz)}.
\]
Consequently, for every $\bx\in\{0,1\}^J$,
\[
\PP'(\MX=\bx)
=\sum_{\bz} p'_\bz \prod_{j=1}^J \bigl(q_{j,\bz}'\bigr)^{\bx_j}\bigl(1-q'_{j,\bz}\bigr)^{1-\bx_j}
=\sum_{\bz} p_{\pi(\bz)} \prod_{j=1}^J q_{j,\pi(\bz)}^{\bx_j}\bigl(1-q_{j,\pi(\bz)}\bigr)^{1-\bx_j}
=\PP(\MX=\bx).
\]
It is straightforward to verify that \(\eta'_j(\bz) >\eta'_j(\bz')\) whenever \(
\bz\succeq q_j\) and \( \bz'\not\succeq q_j,\) for $1\le j\le12$, so the monotonicity condition in Assumption~\ref{assm1}(b) continues to hold under the transformed parameterization.

By construction of $\bbeta_j$, it further follows that $\bbeta'_j$ cannot be obtained from $\bbeta_j$ through any latent-coordinate permutation if $j\equiv3\pmod4$, because the value of $\beta'_{j,13}$ differs from all four entries of the original vector $\bbeta_j$. Therefore, $(\MB,\p)$ and $(\MB',\p')$ are not related by any latent-coordinate permutation but induce the same observable law. Since the set of admissible $(\bbeta_j)_{j=1}^{12}$ has positive Lebesgue measure, the framework is not generically identifiable. In particular, this shows that even when Assumption~\ref{assm1}, Assumption~\ref{assm2}, and Condition~(i) of Theorem~\ref{thm2} all hold, generic identifiability can fail once Condition~(ii) is violated.

\subsection{Proof of Theorem \ref{thm2}}\label{sec: proof of identifiability}

Before presenting the proof of Theorem~\ref{thm2}, we introduce some additional notations.

For each $j$, fix an enumeration of $\calC_j^\mathrm{can}\setminus\{\calX_j\}$ as
\[
\calC_j^\mathrm{can}\setminus\{\calX_j\}=\{T_{1,j},T_{2,j},\dots\},
\qquad (j\in[J]).
\]
For each $t\ge 1$, define the parameter-independent finite discretization
\[
\bar{\mathcal D}_j^{(t)}:=\{T_{1,j},\dots,T_{t,j}\}\cup\{\calX_j\}\subseteq \calC_j^\mathrm{can},
\qquad
\bar{\mathcal D}^{(t)}:=(\bar{\mathcal D}_j^{(t)})_{j\in[J]}.
\]
Then $\kappa_j^{(t)}:=|\bar{\mathcal D}_j^{(t)}|=t+1$, and we index
$\bar{\mathcal D}_j^{(t)}=(S_{1,j}^{(t)},\dots,S_{\kappa_j^{(t)},j}^{(t)})$
with $S_{\kappa_j^{(t)},j}^{(t)}=\calX_j$.

For each $t\ge 1$, define $\MN_1^{(t)}$ to be a $\kappa_1^{(t)}\cdots\kappa_K^{(t)}\times 2^K$ matrix with entries
\[
\MN_1^{(t)}\big((l_1,\dots,l_K),\bz\big)
:=\PP\big(X_1\in S_{l_1,1}^{(t)},\dots,X_K\in S_{l_K,K}^{(t)}\mid \bz\big).
\]
Columns are indexed by $\bz\in\{0,1\}^K$ and rows by $\xi_1=(l_1,\dots,l_K)$ with $l_j\in[\kappa_j^{(t)}]$.
Similarly, let $\MN_2^{(t)}$ be the $\kappa_{K+1}^{(t)}\cdots\kappa_{2K}^{(t)}\times2^K$ matrix whose
$((l_{K+1},\dots,l_{2K}),\bz)$-entry is
\[
\PP\big(X_{K+1}\in S_{l_{K+1},K+1}^{(t)},\dots,X_{2K}\in S_{l_{2K},2K}^{(t)}\mid \bz\big),
\]
and let $\MN_3^{(t)}$ be the $\kappa_{2K+1}^{(t)}\cdots\kappa_{J}^{(t)}\times2^K$ matrix whose
$((l_{2K+1},\dots,l_J),\bz)$-entry is
\[
\PP\big(X_{2K+1}\in S_{l_{2K+1},2K+1}^{(t)},\dots,X_{J}\in S_{l_{J},J}^{(t)}\mid \bz\big).
\]
For brevity, set
\[
\upsilon_1^{(t)}=\prod_{k=1}^K\kappa_k^{(t)},\qquad
\upsilon_2^{(t)}=\prod_{k=K+1}^{2K}\kappa_k^{(t)},\qquad
\upsilon_3^{(t)}=\prod_{k=2K+1}^{J}\kappa_k^{(t)}.
\]
Since $S_{\kappa_j^{(t)},j}^{(t)}=\calX_j$, the last row of each $\MN_a^{(t)}$ equals $\mathbf{1}_{2^K}^\top$.

Define the $3$-way marginal probability tensor $\mathbf{P}_0^{(t)}$ of size
$\upsilon_1^{(t)}\times\upsilon_2^{(t)}\times\upsilon_3^{(t)}$ by
\begin{align*}
\mathbf{P}_0^{(t)}(\xi_1,\xi_2,\xi_3)
&=\PP\big(X_1\in S_{l_1,1}^{(t)},\dots,X_J\in S_{l_J,J}^{(t)}\big)\\
&=\sum_{z}\pi_{\bz}\,
\MN_1^{(t)}\big((l_1,\dots,l_K),\bz\big)\,
\MN_2^{(t)}\big((l_{K+1},\dots,l_{2K}),\bz\big)\,
\MN_3^{(t)}\big((l_{2K+1},\dots,l_{J}),\bz\big).
\end{align*}
Equivalently,
\begin{equation}\label{eq:tensor_decomposition_fixed_t}
\mathbf{P}_0^{(t)}=\big[\MN_1^{(t)}\mathrm{Diag}(\p),\,\MN_2^{(t)},\,\MN_3^{(t)}\big].
\end{equation}

We record two lemmas whose proofs are deferred to the end of the subsection.
The first establishes uniqueness of the tensor decomposition of $\mathbf P_0^{(t)}$ up to a common column permutation.
\begin{lemma}\label{lem:kruskal gid}
Consider a discrete causal representation learning framework with parameters
$(\p^\star,\calG^\star,\MB^\star,\MQ^\star,\bgamma^\star)$ satisfying the conditions of Theorem~\ref{thm2}.
For each $t\ge1$, let $\mathbf P_0^{(t)}$ be the tensor induced by the parameter-independent discretization
$\bar{\mathcal D}^{(t)}$ defined above, with factor matrices $\MN_1^{(t)},\MN_2^{(t)},\MN_3^{(t)}$, so that $\mathbf P_0^{(t)}=[\MN_1^{(t)}\mathrm{Diag}(\p),\MN_2^{(t)},\MN_3^{(t)}]$.
Then there exists a Lebesgue-null set
$\mathcal N_\infty\subset\Omega_K(\MTheta;\calG^\star,\MQ^\star)$
which constrains only $(\MB,\bgamma)$ such that the following holds.

For every $\MTheta\in\Omega_K(\MTheta;\calG^\star,\MQ^\star)\setminus\mathcal N_\infty$,
there exists an integer $t_0=t_0(\MTheta)<\infty$ such that for all $t\ge t_0$
the rank-$2^K$ CP decomposition of $\mathbf P_0^{(t)}$ is unique up to a common column permutation.
Moreover, since $\calX_j\in\bar{\mathcal D}_j^{(t)}$, each $\MN_a^{(t)}$ contains a row equal to
$\mathbf 1_{2^K}^\top$, hence the uniqueness contains no nontrivial scaling ambiguity.
\end{lemma}

The next lemma constrains how the $2^K$ columns can be permuted.

\begin{lemma}\label{lem:sign flip}
Let $(\p,\MB,\calG,\MQ)$ and $(\p',\MB',\calG',\MQ')$ satisfy Assumption~\ref{assm1}, and suppose $\MQ$ also meets the conditions of Theorem~\ref{thm2}. 
Assume there exists a permutation $\mathfrak{S}\in S_{\{0,1\}^K}$ such that
\[
\eta_{j}(\bz)=\eta'_{j}(\mathfrak{S}(\bz))\quad\text{for all }j,\bz.
\]
Then $(\p,\MB,\calG,\MQ)\sim_{\calK}(\p',\MB',\calG',\MQ')$ for $\MB\in \Omega(\MB;\MQ)$, where
\[
\Omega(\MB;\MQ)=\{\MB:\ \beta_{j,S}=0\ \text{whenever }S\nsubseteq K_j,\ 
\beta_{j,\{k\}}\neq 0\ \text{if and only if }k\in K_j\}.
\]
\end{lemma}

Assume $\MQ^\star$ satisfies the conditions of Theorem~\ref{thm2} and that
$\MTheta\in\Omega_K(\MTheta;\calG^\star,\MQ^\star)$.
Suppose there exist alternative parameters $(\widetilde{\MTheta},\widetilde{\calG},\widetilde\MQ)$ such that
\(\PP_{\widetilde{\MTheta},\widetilde{\MQ},\widetilde{\calG}}=\PP_{\MTheta,\MQ^\star,\calG^\star}.
\)
We will show that
$(\MTheta,\calG^\star,\MQ^\star)\sim_K(\widetilde{\MTheta},\widetilde{\calG},\widetilde{\MQ})$
for $\MTheta$ outside a Lebesgue-null set.

By Lemma~\ref{lem:kruskal gid}, there exists a Lebesgue-null set $\calN_\infty$ such that for every $\MTheta\in\Omega_K(\MTheta;\calG^\star,\MQ^\star)\setminus\mathcal N_\infty$, we could find an integer $t_0=t_0(\MTheta)<\infty$
such that for every $t\ge t_0$ the tensor decomposition of $\mathbf P_0^{(t)}$ is unique
up to a common column permutation.
Fix any $t\ge t_0$.
Then
\[
\mathbf{P}_0^{(t)}
=\big[\MN_1^{(t)}\mathrm{Diag}(\p),\,\MN_2^{(t)},\,\MN_3^{(t)}\big]
=\big[\widetilde{\MN}_1^{(t)}\mathrm{Diag}(\widetilde{\p}),\,\widetilde{\MN}_2^{(t)},\,\widetilde{\MN}_3^{(t)}\big],
\]
where the equality holds up to a common permutation of the $2^K$ columns.
Hence there exists a permutation $\mathfrak S^{(t)}\in S_{\{0,1\}^K}$ such that
\[
\MN_a^{(t)}(\cdot,\bz)=\widetilde{\MN}_a^{(t)}(\cdot,\mathfrak S^{(t)}(\bz)),
\qquad a=1,2,3,
\qquad
\p_\bz=\widetilde{\p}_{\mathfrak S^{(t)}(\bz)}.
\]
In particular, for every $j\in[J]$, every $l\in[\kappa_j^{(t)}]$, and every $\bz\in\{0,1\}^K$,
\begin{equation}\label{eq:probability_comparing_t}
\PP_{j,g_j(\eta_j(\bz),\gamma_j)}\big(S_{l,j}^{(t)}\big)
=
{\PP}_{j,g_j(\widetilde\eta_j(\mathfrak S^{(t)}(\bz)),\widetilde\gamma_j))}\big(S_{l,j}^{(t)}\big).
\end{equation}

We now justify the passage from the setwise equalities in \eqref{eq:probability_comparing_t} to equality of
the full conditional laws as probability measures, and simultaneously show that the aligning permutation
stabilizes as $t$ increases.

\begin{lemma}\label{lem:perm_stability}
Fix a countable separating class $\mathcal C_j$ for each $j\in[J]$. Let $\mathcal D_j\subseteq \mathcal D_j^+\subseteq \mathcal C_j$ be two finite collections for each $j$. Construct the corresponding factor matrices $(\MN_1,\MN_2,\MN_3)$ and $(\MN_1^+,\MN_2^+,\MN_3^+)$, and similarly $(\widetilde{\MN}_1,\widetilde{\MN}_2,\widetilde{\MN}_3)$ and $(\widetilde{\MN}_1^+,\widetilde{\MN}_2^+,\widetilde{\MN}_3^+)$ under an alternative parameterization. Assume that both tensors admit unique rank-$2^K$ CP decompositions up to a common column permutation, so that there exist $\mathfrak S,\mathfrak S^+\in S_{\{0,1\}^K}$ satisfying
\[
\MN_a(\cdot,\bz)=\widetilde{\MN}_a(\cdot,\mathfrak S(\bz)),
\qquad
\MN_a^+(\cdot,\bz)=\widetilde{\MN}_a^+(\cdot,\mathfrak S^+(\bz)),
\qquad a=1,2,3.
\]
If $\MN_1$ has full column rank $2^K$, then $\mathfrak S^+=\mathfrak S$.
\end{lemma}
\begin{proof}
Since $\MN_1$ has full column rank $2^K$, its $2^K$ columns are pairwise distinct. Because $\mathcal D_j\subseteq\mathcal D_j^+$ for all $j$, each row event used to define $\MN_1$
(i.e., each product event determined by choosing one set from each $\mathcal D_j$)
also appears among the row events defining $\MN_1^+$.
Thus, for each $\bz\in\{0,1\}^K$, the column $\MN_1(\cdot,\bz)$ is obtained from
$\MN_1^+(\cdot,\bz)$ by restricting to those rows corresponding to product events formed from $\mathcal D$.

Fix $\bz$.
From the two Kruskal conclusions, we have $\MN_1(\cdot,\bz)=\widetilde{\MN}_1(\cdot,\mathfrak S(\bz))$ and $\MN_1^+(\cdot,\bz)=\widetilde{\MN}_1^+(\cdot,\mathfrak S^+(\bz))$.
Restricting the second equality to the rows corresponding to $\mathcal D$ gives $\MN_1(\cdot,\bz)=\widetilde{\MN}_1(\cdot,\mathfrak S^+(\bz))$. Therefore, $\widetilde{\MN}_1(\cdot,\mathfrak S(\bz))=\widetilde{\MN}_1(\cdot,\mathfrak S^+(\bz)).$ Since $\widetilde{\MN}_1$ is a column permutation of $\MN_1$, its columns are also pairwise distinct, so the above equality forces
$\mathfrak S(\bz)=\mathfrak S^+(\bz)$.
As $\bz$ was arbitrary, $\mathfrak S=\mathfrak S^+$.
\end{proof}
Since $\MTheta\notin\mathcal N_\infty$, Lemma~\ref{lem:kruskal gid} implies that
for every $t\ge t_0(\MTheta)$ the Kruskal conclusion holds for $\mathbf P_0^{(t)}$,
and hence the associated aligning permutation $\mathfrak S^{(t)}$ is well-defined. Moreover, for every $t\ge t_0(\MTheta)$, the Kruskal argument in the proof of
Lemma~\ref{lem:kruskal gid} yields $\mathrm{rk}_k(\MN_1^{(t)})=2^K$ and hence
$\MN_1^{(t)}$ has full column rank $2^K$.
Applying Lemma~\ref{lem:perm_stability} to the nested discretizations
$\bar{\mathcal D}^{(t)}\subseteq\bar{\mathcal D}^{(t+1)}$ for $t\ge t_0(\MTheta)$
yields that $\mathfrak S^{(t)}$ is constant in $t$.
Denote the common permutation by $\mathfrak S$.

Now fix $j\in[J]$ and let $S\in\mathcal C_j^\mathrm{can}$ be arbitrary.
Since $\bigcup_{t\ge t_0}\bar{\mathcal D}_j^{(t)}=\mathcal C_j^\mathrm{can}$, there exists $t\ge t_0$ such that $S\in\bar{\mathcal D}_j^{(t)}$.
Therefore \eqref{eq:probability_comparing_t} implies
\[
\PP_{j,g_j(\eta_j(\bz),\gamma_j)}(S)={\PP}_{j,g_j(\widetilde\eta_j(\mathfrak S(\bz)),\widetilde\gamma_j))}(S)
\qquad
\text{for all } \bz\in\{0,1\}^K.
\]
Because $\mathcal C_j^\mathrm{can}$ is separating, we conclude that
\[
\PP_{j,g_j(\eta_j(\bz),\gamma_j)}={\PP}_{j,g_j(\widetilde\eta_j(\mathfrak S(\bz)),\widetilde\gamma_j))}
\qquad\text{as probability measures on }\calX_j,\quad\text{for all }\bz.
\]
By Assumption~\ref{assm2}(ii) and injectivity of $g_j$ in Assumption~\ref{assm2}(iii), this further implies
\[
\eta_{j}(\bz)=\widetilde{\eta}_{j}(\mathfrak S(\bz))
\qquad\text{and}\qquad
\gamma_j=\widetilde\gamma_j,
\qquad\text{for all }j,\bz.
\]
Finally, Lemma~\ref{lem:sign flip} yields
$(\p,\MB,\calG,\MQ)\sim_{\calK}(\widetilde{\p},\widetilde{\MB},\widetilde{\calG},\widetilde{\MQ})$,
completing the proof.

\subsubsection{Proof of Lemma~\ref{lem:kruskal gid}}

Fix $t\ge1$ and consider the tensor $\mathbf P_0^{(t)}$ induced by the parameter-independent discretization
$\bar{\mathcal D}^{(t)}$ constructed at the beginning of this subsection, with factor matrices
$\MN_1^{(t)},\MN_2^{(t)},\MN_3^{(t)}$ satisfying \eqref{eq:tensor_decomposition_fixed_t}.
Write $\mathrm{rk}_k(M)$ for the Kruskal column rank of a matrix $M$.

By Kruskal's theorem, it suffices to prove that there exist a Lebesgue-null set
$\mathcal N_\infty\subset\Omega_K(\MTheta;\calG^\star,\MQ^\star)$ which constrains only $(\MB,\bgamma)$ and,
for every $\MTheta\in\Omega_K(\MTheta;\calG^\star,\MQ^\star)\setminus\mathcal N_\infty$, an integer
$t_0=t_0(\MTheta)<\infty$ such that
\begin{equation}\label{eq:kruskal_target_three_clean}
\mathrm{rk}_k(\MN_1^{(t)})=2^K,\qquad
\mathrm{rk}_k(\MN_2^{(t)})=2^K,\qquad
\mathrm{rk}_k(\MN_3^{(t)})\ge2
\end{equation}
for all $t\ge t_0(\MTheta)$.

We first establish that $\mathrm{rk}_k(\MN_1^{(t)})=2^K$ generically for all sufficiently large $t$.
Let $\mathcal J^{(K)}_{\mathrm{disp}}\subseteq[K]$ denote the indices of items among $\{1,\dots,K\}$ whose response
family genuinely includes an unknown dispersion parameter $\gamma_j$.
For one-parameter families we do not treat $\gamma_j$
as a coordinate.

Define the local parameter block
\[
\Theta_1
:=\Bigl(\ \{\beta_{j,S}: j\in[K],\,S\subseteq K_j\},\ \{\gamma_j: j\in\mathcal J^{(K)}_{\mathrm{disp}}\}\ \Bigr),
\]
which we identify with a vector in a Euclidean space $\RR^{D_1}$ of dimension
$D_1:=\sum_{j=1}^K 2^{|K_j|}+|\mathcal J^{(K)}_{\mathrm{disp}}|$. Throughout we restrict attention to the open connected domain
\[
U_1:=\Bigl\{\Theta_1:\ \gamma_j>0\ \text{for all }j\in\mathcal J^{(K)}_{\mathrm{disp}}\Bigr\}
\subset \RR^{D_1}.
\]

Because the discretization $\bar{\mathcal D}^{(t)}$ is parameter-independent, every entry of $\MN_1^{(t)}$ is
a finite product of terms of the form
\[
\PP_{j,\;g_j(\eta_j(\bz),\gamma_j)}(S),
\qquad S\in\bar{\mathcal D}_j^{(t)},
\]
with $S$ fixed.
For $j\in\mathcal J^{(K)}_{\mathrm{disp}}$, Assumption~\ref{assm2}(i) and Assumption~\ref{assm2}(iii) imply that
$(\eta,\gamma)\mapsto \PP_{j,g_j(\eta,\gamma)}(S)$ is real-analytic on $\RR\times(0,\infty)$.
For $j\notin\mathcal J^{(K)}_{\mathrm{disp}}$, $g_j$ is independent of $\gamma$ and
$\PP_{j,\,g_j(\eta,\gamma)}(S)=\PP_{j,\,g_j(\eta,\gamma_0)}(S)$ for any fixed $\gamma_0\in[0,\infty)$; in particular,
$\eta \mapsto \PP_{j,g_j(\eta,\gamma_0)}(S)$
is real-analytic on $\RR$.
Since each $\eta_j(\bz)$ is a polynomial in the coefficients $\{\beta_{j,S}\}$, it follows that every
entry of $\MN_1^{(t)}$ is real-analytic on the domain $U_1$.
Consequently, $f_{1,t}(\Theta_1):=\det((\MN_1^{(t)})^\top \MN_1^{(t)})$ is a real-analytic function on the domain $U_1\subset\RR^{D_1}$.

Next we describe the projection of $\Omega_K(\MTheta;\calG^\star,\MQ^\star)$ onto $\Theta_1\subseteq U_1$.
Let $E_1=\{(j,k): j\in[K],\,k\in K_j\}$ index the main-effect coefficients $\beta_{j,\{k\}}$ for $j\le K$.
For each sign pattern $\sigma\in\{-1,+1\}^{E_1}$ define the orthant
\[
\mathcal E_1^{(\sigma)}
:=\Bigl\{\Theta_1\subseteq U_1:\ \beta_{j,\{k\}}\sigma_{j,k}>0\ \text{for all }(j,k)\in E_1\Bigr\},
\]
where all remaining coordinates in $\Theta_1$ are unrestricted.
Each $\mathcal E_1^{(\sigma)}$ is an open, connected domain of $\RR^{D_1}$, and
\[
\Pi_1\bigl(\Omega_K(\MTheta;\calG^\star,\MQ^\star)\bigr)
=\bigcup_{\sigma}\mathcal E_1^{(\sigma)},
\]
where $\Pi_1$ denotes the projection onto the coordinates~$\Theta_1$.

Fix any sign pattern $\sigma\in\{-1,+1\}^{E_1}$.
We will show that there exist an explicit parameter point
$\bar\Theta_1^{(\sigma)}\in\mathcal E_1^{(\sigma)}$ and an integer $t_\sigma<\infty$
such that $f_{1,t}\!\left(\bar\Theta_1^{(\sigma)}\right)>0$ for all $t\ge t_\sigma$.
In particular, for every $t\ge t_\sigma$ the restriction of $f_{1,t}$ to $\mathcal E_1^{(\sigma)}$ is a
nontrivial real-analytic function on the open, connected domain
$\mathcal E_1^{(\sigma)}$.
By \cite{mityagin2015the}, the zero set
$\calV_{1,\sigma,t}:=\{\Theta_1\in\mathcal E_1^{(\sigma)}: f_{1,t}(\Theta_1)=0\}$
has Lebesgue measure zero in $\mathcal E_1^{(\sigma)}$.

To construct the point, we now explicitly use condition~(i) of Theorem~\ref{thm2}.
There exists a permutation $\varrho_1$ of $[K]$ such that the permuted $K\times K$ block $\overline\MQ_1:=\MQ_{\varrho_1(1:K),:\ }$
has unit diagonal, equivalently $q_{\varrho_1(r),r}=1$ for all $r\in\{1,\dots,K\}$.
Define the induced bijection $\rho:[K]\to[K]$ by
$\rho(j):=\varrho_1^{-1}(j)$ for $j\in\{1,\dots,K\}$. Then, for every $j\in[K]$, we have $q_{j,\,\rho(j)}=1$, hence $\rho(j)\in K_j$.

Fix a sign pattern $\sigma\in\{-1,+1\}^{E_1}$.
Define a boundary point $\Theta_{1,0}^{(\sigma)}$ by setting all interaction terms to zero and keeping
only the single admissible main effect $\beta_{j,\{\rho(j)\}}$ for each $j\in[K]$
\[
\beta_{j,S}=0 \ \text{for all }j\in[K]\text{ and }|S|\ge2,
\quad
\beta_{j,\{\rho(j)\}}=\sigma_{j,\rho(j)},
\quad
\beta_{j,\{k\}}=0 \ \text{for all }(j,k)\in E_1,\ k\neq \rho(j),
\]
with all remaining coordinates in $\Theta_1$ arbitrary, and $\gamma_j>0$ for
$j\in\mathcal J^{(K)}_{\mathrm{disp}}$.
Because $\rho(j)\in K_j$, the coordinate $\beta_{j,\{\rho(j)\}}$ is indeed part of $\Theta_1$, so this assignment is admissible.
Under $\Theta_{1,0}^{(\sigma)}$, for each $j\in[K]$ the conditional law of $X_j$ depends on $\bz$
only through $\bz_{\rho(j)}$, since $\eta_j(\bz)=\beta_{j,\emptyset}+\beta_{j,\{\rho(j)\}}\bz_{\rho(j)}$.

Fix $j\in[K]$.
Let $\mu_{j,0}$ and $\mu_{j,1}$ denote the two conditional laws of $X_j$ under $\Theta_{1,0}^{(\sigma)}$
corresponding to $\bz_{\rho(j)}=0$ and $\bz_{\rho(j)}=1$, respectively.
Because $\beta_{j,\{\rho(j)\}}\neq0$, we have $\eta_j(0)\neq\eta_j(1)$ in the $\bz_{\rho(j)}$ coordinate.
By Assumption~\ref{assm2}(iii) the map $\eta\mapsto g_j(\eta,\gamma_j)$ is injective for fixed $\gamma_j$,
hence the induced parameters are distinct.
By identifiability in Assumption~\ref{assm2}(ii), it follows that $\mu_{j,0}\neq\mu_{j,1}$ as probability measures.
Since $\mathcal C_j^\mathrm{can}$ is separating, there exists a set $B_j\in\mathcal C_j^\mathrm{can}$ such that $\mu_{j,0}(B_j)\neq \mu_{j,1}(B_j).$

For each $r\in[K]$, define $j_r:=\varrho_1(r)$, so that $\rho(j_r)=r$ and $q_{j_r,r}=1$.
For $j_r$, choose a distinguishing set $B_{j_r}\in\mathcal C_{j_r}$ as above and write
$B_{j_r}=T_{i_r,\,j_r}$
for some $i_r\ge1$.
Define $t_\sigma:=\max_{r\in[K]} i_r  < \infty.$ For every $t\ge t_\sigma$ we have $B_{j_r}\in\bar{\mathcal D}_{j_r}^{(t)}$ for all $r\in[K]$, and moreover $S_{\kappa_{j_r}^{(t)},\,j_r}^{(t)}=\mathcal X_{j_r}\in\bar{\mathcal D}_{j_r}^{(t)}$.

Fix any $t\ge t_\sigma$.
Consider the $2^K\times 2^K$ submatrix of $\MN_1^{(t)}$ obtained by restricting to the $2^K$ rows indexed by
$l_{j_r}\in\{i_r,\ \kappa_{j_r}^{(t)}\} $ ($r=1,\dots,K$),
that is, for each $r$ we use either the event $B_{j_r}$ or the event $\mathcal X_{j_r}$.
Under $\Theta_{1,0}^{(\sigma)}$ and conditional independence of $(X_{j_1},\dots,X_{j_K})$ given $\bz$,
this submatrix factorizes as a Kronecker product:
\[
\MN_{1,\mathrm{sub}}^{(t)}\bigl(\Theta_{1,0}^{(\sigma)}\bigr)
=
\bigotimes_{r=1}^K
\begin{pmatrix}
\mu_{j_r,0}(B_{j_r}) & \mu_{j_r,1}(B_{j_r})\\
1 & 1
\end{pmatrix}.
\]
Each $2\times2$ factor has nonzero determinant
$\mu_{j_r,0}(B_{j_r})-\mu_{j_r,1}(B_{j_r})\neq0$, hence the Kronecker product is invertible.
Therefore $\MN_1^{(t)}\bigl(\Theta_{1,0}^{(\sigma)}\bigr)$ has full column rank $2^K$.

Next, we perturb $\Theta_{1,0}^{(\sigma)}$ into the open orthant $\mathcal E_1^{(\sigma)}$ while preserving
full column rank of $\MN_1^{(t)}$.
For $\varepsilon>0$, define $\Theta_{1,\varepsilon}^{(\sigma)}$ by keeping all coordinates of
$\Theta_{1,0}^{(\sigma)}$ unchanged except setting, for every $(j,k)\in E_1$ with $k\neq \rho(j)$, $\beta_{j,\{k\}}=\sigma_{j,k}\varepsilon$,
and leaving $\beta_{j,\{\rho(j)\}}=\sigma_{j,\rho(j)}$.
Then $\Theta_{1,\varepsilon}^{(\sigma)}\in\mathcal E_1^{(\sigma)}$ for every $\varepsilon>0$.
Since the determinant of the fixed $2^K\times2^K$ submatrix $\MN_{1,\mathrm{sub}}^{(t)}(\Theta_1)$ is continuous in
$\Theta_1$ on $U_1$ and is nonzero at $\Theta_{1,0}^{(\sigma)}$, there exists $\varepsilon_\sigma(t)>0$ such that for all
$\varepsilon\in(0,\varepsilon_\sigma(t))$ the submatrix remains invertible, and hence
$\MN_1^{(t)}\bigl(\Theta_{1,\varepsilon}^{(\sigma)}\bigr)$ has full column rank $2^K$.
Choose any such $\varepsilon$ and set $\bar\Theta_1^{(\sigma)}:=\Theta_{1,\varepsilon}^{(\sigma)}$.
Then $f_{1,t}\!\left(\bar\Theta_1^{(\sigma)}\right)>0$.

Since there are only finitely many sign patterns, define $t_1:=\max_{\sigma\in\{-1,+1\}^{E_1}} t_\sigma  <\infty$. Then for every $t\ge t_1$ and every sign pattern $\sigma$, the restriction of $f_{1,t}$ to $\mathcal E_1^{(\sigma)}$
is a nontrivial real-analytic function, and hence the set
$\calV_{1,\sigma,t}:=\{\Theta_1\in\mathcal E_1^{(\sigma)}: f_{1,t}(\Theta_1)=0\}$
has Lebesgue measure zero in $\mathcal E_1^{(\sigma)}$.

For each fixed $t\ge t_1$, define
\[
\calV_{1,t}
:=
\Pi_1\bigl(\Omega_K(\MTheta;\calG^\star,\MQ^\star)\bigr)\cap\bigcup_{\sigma} \calV_{1,\sigma,t}.
\]
Then $\calV_{1,t}$ has Lebesgue measure zero in
$\Pi_1\bigl(\Omega_K(\MTheta;\calG^\star,\MQ^\star)\bigr)$.
Define
\[
\mathcal N_{1,t}^{(1)}
:=\{(\p,\MB,\bgamma)\in\Omega_K(\MTheta;\calG^\star,\MQ^\star):\ \Theta_1\in \calV_{1,t}\}.
\]
Then $\mathcal N_{1,t}^{(1)}$ has Lebesgue measure zero in $\Omega_K(\MTheta;\calG^\star,\MQ^\star)$ and depends
only on $(\MB,\bgamma)$.
For every $t\ge t_1$ and every $\MTheta\notin\mathcal N_{1,t}^{(1)}$ we have $f_{1,t}(\Theta_1)\neq0$, hence $\MN_1^{(t)}$ has full
column rank $2^K$, which implies $\mathrm{rk}_k(\MN_1^{(t)})=2^K$.

An entirely analogous argument yields an integer $t_2<\infty$ and, for each $t\ge t_2$, a Lebesgue-null set
$\mathcal N_{1,t}^{(2)}\subset\Omega_K(\MTheta;\calG^\star,\MQ^\star)$ depending only on $(\MB,\bgamma)$ such that
\[
\mathrm{rk}_k(\MN_2^{(t)})=2^K
\qquad\text{for all }t\ge t_2,\ \MTheta\in\Omega_K(\MTheta;\calG^\star,\MQ^\star)\setminus\mathcal N_{1,t}^{(2)}.
\]

We now prove that $\mathrm{rk}_k(\MN_3^{(t)})\ge 2$ generically. It suffices to verify that
Condition~B of \citet{lee2024new} holds for generic parameters in
$\Omega_K(\MTheta;\calG^\star,\MQ^\star)$.

Fix any pair $\bz\neq\bz'$ and choose an index
$\ell=\ell(\bz,\bz')\in[K]$ such that $z_\ell\neq z_\ell'$.
Since $\MQ_3$ has no all-zero column, there exists an item
$j=j(\bz,\bz')\in\{2K+1,\dots,J\}$ such that $\MQ_{j,\ell}=1$, hence $\ell\in K_j$.

Similarly, let $\mathcal J_{\mathrm{disp}}\subseteq[J]$ denote the indices of items whose response family genuinely
includes an unknown dispersion parameter $\gamma_j$.
Collect the local measurement parameters for item $j$ into the free block
\[\Theta_j
  :=
  \begin{cases}
  \bigl(\{\beta_{j,S}: S\subseteq K_j\},\,\gamma_j\bigr)\in\RR^{\,2^{|K_j|}}\times(0,\infty), & j\in\mathcal J_{\mathrm{disp}},\\[2mm]
  \bigl(\{\beta_{j,S}: S\subseteq K_j\}\bigr)\in\RR^{\,2^{|K_j|}}, & j\notin\mathcal J_{\mathrm{disp}}.
  \end{cases}\]
By the definition of $\Omega_K(\MTheta;\calG^\star,\MQ^\star)$, we have
$\beta_{j,\{k\}}=0$ for all $k\notin K_j$ and $\beta_{j,\{\ell\}}\neq 0$.

For this fixed pair $(\bz,\bz')$, consider the difference
\[
  h_{j,\bz,\bz'}(\Theta_j)
  := \eta_j(\bz)-\eta_j(\bz')
  = \sum_{S\subseteq K_j} \beta_{j,S}
      \Bigl(\prod_{k\in S}z_k - \prod_{k\in S}z_k'\Bigr).
\]
This is a linear function of the coefficients $\{\beta_{j,S}\}$ (and is independent of $\gamma_j$ when
$j\in\mathcal J_{\mathrm{disp}}$).
Moreover, the coefficient of $\beta_{j,\{\ell\}}$ in $h_{j,\bz,\bz'}$ equals
$z_\ell-z_\ell'\neq 0$, hence $h_{j,\bz,\bz'}$ is a nontrivial linear functional.

Therefore the zero set
\[
  E_{j,\bz,\bz'}
  := \bigl\{\Theta_j:\ h_{j,\bz,\bz'}(\Theta_j)=0\bigr\}
\]
has Lebesgue measure zero in the free-coordinate space of $\Theta_j$:
it is an affine hyperplane in $\RR^{2^{|K_j|}}$ when $j\notin\mathcal J_{\mathrm{disp}}$, and it is
an affine hyperplane in the $\beta$-coordinates times $(0,\infty)$ when $j\in\mathcal J_{\mathrm{disp}}$.

Define the corresponding exceptional set in the full parameter space by
\[
  \mathcal N_{\bz,\bz'}^{(3)}
  :=\Bigl\{\MTheta\in\Omega_K(\MTheta;\calG^\star,\MQ^\star):\
  h_{j(\bz,\bz'),\bz,\bz'}(\Theta_{j(\bz,\bz')})=0\Bigr\}.
\]
Since $h_{j(\bz,\bz'),\bz,\bz'}$ is a nontrivial linear functional of the local block
$\Theta_{j(\bz,\bz')}$ (through its $\beta$-coordinates), the set
$\mathcal N_{\bz,\bz'}^{(3)}$ has Lebesgue measure zero in
$\Omega_K(\MTheta;\calG^\star,\MQ^\star)$.

Now fix $\MTheta\in\Omega_K(\MTheta;\calG^\star,\MQ^\star)\setminus \mathcal N_{\bz,\bz'}^{(3)}$.
Then $\eta_j(\bz)\neq\eta_j(\bz')$. Writing
\[
\PP_{j,\bz}=\PP_{j,\theta_{j,\bz}},
\qquad
\theta_{j,\bz}:=g_j\!\bigl(\eta_j(\bz),\gamma_j\bigr)\in H_j^\circ,
\]
injectivity of $\eta\mapsto g_j(\eta,\gamma_j)$ in Assumption~\ref{assm2}(iii) implies
$\theta_{j,\bz}\neq\theta_{j,\bz'}$.
By identifiability of $\mathrm{ParFam}_j$ in Assumption~\ref{assm2}(ii), we have
$\PP_{j,\bz}\neq \PP_{j,\bz'}$ as probability measures.
Since $\calC_j^\mathrm{can}$ is separating, there exists a set $S_{\bz,\bz'}\in\calC_j^\mathrm{can}$ such that
$\PP_{j,\bz}(S_{\bz,\bz'})\neq \PP_{j,\bz'}(S_{\bz,\bz'})$.
By the enumeration $\calC_j^\mathrm{can}=\{T_{1,j},T_{2,j},\dots\}$, we may write
$S_{\bz,\bz'}=T_{m(\bz,\bz'),\,j}$ for some index $m(\bz,\bz')\ge 1$.
Consequently, for every $t\ge m(\bz,\bz')$, we have
$S_{\bz,\bz'}\in\bar{\mathcal D}_j^{(t)}$, and therefore Condition~B holds for this pair
$(\bz,\bz')$.

Since there are finitely many pairs $(\bz,\bz')$ with $\bz\neq\bz'$,
define $t_3:=\max_{\bz\neq\bz'} m(\bz,\bz')  < \infty$ and $\mathcal N_1^{(3)}:=\bigcup_{\bz\neq\bz'} \mathcal N_{\bz,\bz'}^{(3)}$.
Then $\mathcal N_1^{(3)}$ has Lebesgue measure zero in $\Omega_K(\MTheta;\calG^\star,\MQ^\star)$ and constrains only
the measurement parameters. Moreover, for every $t\ge t_3$ and every
$\MTheta\in\Omega_K(\MTheta;\calG^\star,\MQ^\star)\setminus\mathcal N_1^{(3)}$,
Condition~B holds for the discretization $\bar{\mathcal D}^{(t)}$.

Finally, Condition~B implies that for every $\bz\neq\bz'$ there exists a row of $\MN_3^{(t)}$
(using $S_{\bz,\bz'}$ for item $j$ and $\calX$ for all other items in block~3) on which the
$\bz$-th and $\bz'$-th columns differ, hence all columns are pairwise distinct.
Together with the fact that $\MN_3^{(t)}$ contains the all-$\calX$ row
$\mathbf 1_{2^K}^\top$, this rules out collinearity of any two columns and yields
\[
\mathrm{rk}_k(\MN_3^{(t)})\ge 2,
\qquad\text{for all }t\ge t_3,\ 
\MTheta\in\Omega_K(\MTheta;\calG^\star,\MQ^\star)\setminus\mathcal N_1^{(3)}.
\]

For each $t\ge \max\{t_1,t_2,t_3\}$ define
$\mathcal N_1^{(t)}:=\mathcal N_{1,t}^{(1)}\cup\mathcal N_{1,t}^{(2)}\cup\mathcal N_1^{(3)}$.
By construction $\mathcal N_1^{(t)}$ has Lebesgue measure zero in
$\Omega_K(\MTheta;\calG^\star,\MQ^\star)$ and depends only on $(\MB,\bgamma)$, not on $\p$.

Now define the global exceptional set
\[
\mathcal N_\infty
:=\Big(\bigcup_{t\ge \max\{t_1,t_2,t_3\}}\mathcal N_{1,t}^{(1)}\Big)\ \cup\
\Big(\bigcup_{t\ge \max\{t_1,t_2,t_3\}}\mathcal N_{1,t}^{(2)}\Big)\ \cup\ \mathcal N_1^{(3)}.
\]
Since $\{\mathcal N_{1,t}^{(1)}\}_{t\ge \max\{t_1,t_2,t_3\}}$ and
$\{\mathcal N_{1,t}^{(2)}\}_{t\ge \max\{t_1,t_2,t_3\}}$ are countable families of Lebesgue-null sets, the union
$\mathcal N_\infty$ is also Lebesgue-null, and it still constrains only $(\MB,\bgamma)$.

Fix $\MTheta\in\Omega_K(\MTheta;\calG^\star,\MQ^\star)\setminus\mathcal N_\infty$.
Define $t_0(\MTheta):=\max\{t_1, t_2, t_3\}.$
Then for all $t\ge t_0(\MTheta)$ we have simultaneously
\[
\mathrm{rk}_k(\MN_1^{(t)})=2^K,\qquad
\mathrm{rk}_k(\MN_2^{(t)})=2^K,\qquad
\mathrm{rk}_k(\MN_3^{(t)})\ge2.
\]
This completes the proof.

\subsubsection{Proof of Lemma~\ref{lem:sign flip}}\label{subsubsection: key}
To avoid ambiguity, we fix the identification
\[
f_b:\{0,1\}^K \longrightarrow \mathcal P([K]),\qquad 
f_b(v):=\{\,k\in[K]: v_k=1\,\}.
\]
Conversely, for a subset $T\subseteq[K]$, we write $(f_b)^{-1}(T)\in\{0,1\}^K$ for its indicator vector.

A permutation $\pi\in S_K$ acts on subsets by $\pi(T):=\{\pi(i): i\in T\}\subseteq[K]$,
and induces the corresponding action on vectors by permuting coordinates:
\[
(\pi\cdot v)_k= v_{\pi^{-1}(k)}\qquad (v\in\{0,1\}^K,\ k\in[K]).
\]

We first establish that $\MQ \sim_{\calK} \MQ'$. Equivalently, there exists a column permutation $\pi\in S_K$ such that
\[
\MQ'_{j,:}=\pi\cdot \MQ_{j,:}\qquad\text{for all }j\in\{1,\dots,J\}.
\]

In fact, for $(\boeta_{j}(\bz))_{1\leq j\le J,\ \bz\in\{0,1\}^K}\in\RR^{2^K\times J}$ and $(\boeta'_{j}(\bz))_{1\leq j\le J,\ \bz\in\{0,1\}^K}\in\RR^{2^K\times J}$ that differ only by a row permutation, the counts of coordinatewise maximal rows are preserved. In the single–index case $S=\{j\}$, a row indexed by $\bz\in\{0,1\}^K$ is maximal in the $2^K\times 1$ submatrix $\boeta_{:,\,\{j\}}$ if and only if
\[
\bz \succeq (f_b)^{-1}(K_j),
\]
hence the number of maximal rows equals $2^{\,K-|K_j|}$. The same conclusion holds for $\boeta'$, so $2^{\,K-|K_j|}=2^{\,K-|K_j'|}$ and therefore $|K_j|=|K_j'|$.

For any distinct $j,j'$, a row is coordinatewise maximal in the $2^K\times 2$ submatrix $\boeta_{:,\,\{j,j'\}}$ if and only if
\[
\bz \succeq (f_b)^{-1}(K_j)\quad\text{and}\quad \bz \succeq (f_b)^{-1}(K_{j'}),
\]
which is equivalent to $\bz \succeq (f_b)^{-1}\!\big(K_j\cup K_{j'}\big)$. Consequently, the number of maximal rows is $2^{\,K-|K_j\cup K_{j'}|}$, and by row–permutation invariance we obtain
\[
|K_j\cup K_{j'}|=|K_j'\cup K_{j'}'|.
\]

More generally, for any finite $S\subseteq\{1,\ldots,J\}$, a row is coordinatewise maximal in $\boeta_{:,\,S}$ if and only if
\[
\bz \succeq (f_b)^{-1}\!\Big(\,\bigcup_{j\in S}K_j\Big),
\]
so the number of maximal rows is $2^{\,K-\big|\cup_{j\in S}K_j\big|}$, which is invariant under row permutations. Hence
\[
\Big|\bigcup_{j\in S}K_j\Big|=\Big|\bigcup_{j\in S}K_j'\Big|\qquad\text{for all }S\subseteq\{1,\ldots,J\}.
\]

For $T\subseteq [J]$, write
\[
N_T \;:=\; \big|\bigcap_{j\in T} K_j\big|
\qquad\text{and}\qquad
N_T' \;:=\; \big|\bigcap_{j\in T} K_j'\big|.
\]
It is straightforward to check
\[
\bigg|\bigcup_{j\in S} K_j\bigg| \;=\; \sum_{\emptyset\neq T\subseteq S} (-1)^{|T|+1} N_T
\quad\text{for all }S\subseteq [J],
\]
and \[\bigg|\bigcup_{j\in S} K'_j\bigg| \;=\; \sum_{\emptyset\neq T\subseteq S} (-1)^{|T|+1} N'_T
\quad\text{for all }S\subseteq [J].\] Since the union cardinalities match for all $S$, we deduce that
\[
N_T \;=\; N_T' \qquad \text{for all } T\subseteq [J].
\]

Next, for each pattern $v\in\{0,1\}^J$ with support $\text{supp}(v):=\{j: v_j=1\}$, define
\[
M_v \;:=\; \big|\{\,k\in [K]: \MQ_{:,k}=v\,\}\big|,
\qquad
M_v' \;:=\; \big|\{\,k\in [K]: \MQ'_{:,k}=v\,\}\big|.
\]
The intersection counts decompose as
\[
N_T \;=\; \sum_{v:\, \text{supp}(v)\supseteq T} M_v
\qquad\text{for all } T\subseteq [J],
\]
and likewise for $N_T'$ in terms of $M_v'$. Similarly, we have
\[
M_v \;=\; M_v' \qquad \text{for all } v\in\{0,1\}^J.
\]
Hence the multisets of columns of $\MQ$ and $\MQ'$ coincide. Equivalently, there exists a column permutation $\pi\in S_K$ such that
\[
\MQ'_{j,:}=\pi\cdot \MQ_{j,:}\qquad\text{for all }j\in\{1,\dots,J\}.
\]

Next, we have the following claim.

\begin{claim}
Suppose that $\MQ$ satisfies the conditions of Theorem~\ref{thm2}, and that $\MQ'_{j,:} = \pi \cdot \MQ_{j,:}$ for all $j$, then the admissible column permutation $\mathfrak{S} \in S_{2^K}$ (acting on the $2^K$ latent states) such that
\(
\eta_{j,\bz}=\eta'_{j,\mathfrak{S}(\bz)}\ (\forall \ j,\bz)
\) can be restricted to the right coset $(\ZZ_2)^K \pi$. 
In other words, every admissible $\mathfrak{S}$ can be written in the form
\[
f_b\big(\mathfrak{S}((f_b)^{-1}(T))\big)
= \pi(T) \,\Delta\, A,
\qquad T \subseteq [K],
\]
for some subset $A \subseteq [K]$, where $A \Delta B = (A\setminus B)\cup(B\setminus A)$ for any two sets $A$ and $B$ 
 
If we further suppose that both $(\MB,\MQ)$ and $(\MB',\MQ')$ satisfy Assumption~\ref{assm1}(b), then the only admissible permutation is $\pi$. 
\end{claim}

A detailed proof of this reduction is given in the next subsection.

Based on this claim, we directly conclude that $(\p,\MB,\calG,\MQ)\sim_{\calK}(\p',\MB',\calG',\MQ')$. 

\subsubsection{Proof of the Claim}
We now state a lemma, which is essentially an equivalent reformulation of the first part of the claim in Subsection~\ref{subsubsection: key}.  However, for clarity of exposition, we still restate it in a slightly different language.

Let $\mathcal P([K])$ denote the family of all subsets of $[K]=\{1,\dots,K\}$.  
All $2^K$-dimensional vectors and $2^K\times 2^K$ matrices below are indexed by $\mathcal P([K])$ in lexicographic order.   
Define
\[
Y_{Q,T}=\mathbf 1\{T\subseteq Q\},\qquad Q,T\subseteq[K].
\]
Given $\bbeta=(\beta_T)_{T\subseteq[K]}\in\mathbb R^{2^K}$, set 
\[
\boeta=\MY\bbeta,\qquad \text{so that }\ \eta_Q=\sum_{T\subseteq Q}\beta_T.
\]

Let $f:\mathcal P([K])\to\mathcal P([K])$ be a bijection, and let $\MP_f$ be its associated permutation matrix acting on coordinates indexed by $\mathcal P([K])$:
\[
(\MP_f v)_Q = v_{f^{-1}(Q)},\qquad v\in\mathbb R^{2^K}.
\]
Hence $\MP_f$ simply reorders the $2^K$ coordinates of any vector according to $f$, and thus $\MP_f\in S_{2^K}$.  
We then define
\[
\boeta'=\MP_f\boeta,\quad 
\bbeta'=\MY^{-1}\MP_f\MY\bbeta.
\]

For any $Q \subseteq[K]$, define
\[
\Pi_{Q}=\mathrm{Diag}\big(\mathbf{1}\{\,Q'\subseteq Q\,\}\big)_{Q'\subseteq[K]}\in\mathbb R^{2^K\times 2^K},\]
\begin{equation}\label{eq: E_Q}E_{Q}=\mathrm{Im}(\Pi_{Q})=\big\{\x\in\mathbb R^{2^K}:x_{Q'}=0\text{ if }Q'\nsubseteq Q\big\}.
\end{equation}

Let $\mathcal F=\{Q_1,\dots,Q_l\}\subseteq\mathcal P([K])$ be a family of subsets of $[K]$, and let $\pi\in S_K$ be a permutation on $\{1,\dots,K\}$.  
For each subset $Q\subseteq[K]$, write
\[
\pi(Q):=\{\pi(i):i\in Q\}\subseteq[K],
\]
so that $\pi$ acts naturally on $\mathcal P([K])$. We define the set of row permutations preserving the corresponding subspaces:
\[
G_\pi(Q)=\big\{\MP_f:\ \MY^{-1}\MP_f\MY E_{Q}\subseteq E_{\pi(Q)}\big\},
\]
and their intersection over the family $\mathcal F$:
\[
H_\pi(\mathcal F)=\bigcap_{Q\in\mathcal F}G_\pi(Q)=\bigcap_{i=1}^{l}G_\pi(Q_i).
\]

To describe coordinate structure, define the signature map
\[
\phi_{\mathcal F}:[K]\to\{0,1\}^{l},\qquad 
\phi_{\mathcal F}(i):=(\mathbf 1\{\,i\in Q_j\,\})_{1\leq j\le l}.
\]

\begin{lemma}\label{lem:main}
Assume that for every $i\in[K]$ there exists some $Q\in\mathcal F$ such that $i\notin Q$,  
and that for any distinct $i,j\in[K]$, neither $\phi_{\mathcal F}(i)\preceq \phi_{\mathcal F}(j)$ nor $\phi_{\mathcal F}(j)\preceq \phi_{\mathcal F}(i)$ holds.   
Then the intersection $H_\pi(\mathcal F)$ coincides with the coset of $(\mathbb Z_2)^K\rtimes S_K$ corresponding to $\pi$, namely
\[
H_\pi(\mathcal F)\ =\ (\mathbb Z_2)^K\pi,
\]
where $(\mathbb Z_2)^K$ denotes the subgroup of coordinatewise bit-flip operations 
\[
T\ \mapsto\ T\ \Delta\ A,\qquad A\subseteq[K].
\]
Equivalently, every bijection $f\in H_\pi(\mathcal F)$ is uniquely of the form
\[
f(T)\ =\ \pi(T)\ \Delta\ A,\qquad  A\subseteq[K].
\]
\end{lemma}
Now we are ready to give the proof of Lemma \ref{lem:main}. We need three propositions.

\begin{proposition}\label{prop:single-Q}
Let $Q\subseteq[K]$ and, for each $U\in\mathcal P(Q)$, define the $Q$-blocks
\[
B^{Q}(U)\ :=\ \{\,T\subseteq[K]:\ T\cap Q=U\,\}\ \subseteq\mathcal P([K]).
\]
Recall that $f:\mathcal P([K])\to\mathcal P([K])$ is a bijection and $\MP_f$ is its associated permutation matrix acting by $(\MP_f v)_Q=v_{f^{-1}(Q)}$. Then the following are equivalent:
\begin{align*}
\text{\emph{(a)}}\quad& \MP_f\in G_\pi(Q).\\
\text{\emph{(b)}}\quad& \exists\ \text{a bijection }g:\mathcal P(Q)\to\mathcal P(\pi(Q))\ \text{ such that }\ 
f\big(B^{Q}(U)\big)=B^{\pi(Q)}\big(g(U)\big)\ \ \ \forall\,U\in\mathcal P(Q).
\end{align*}
\end{proposition}

\begin{proof}
Recall $E_Q=\big\{\x\in\mathbb R^{2^K}:x_{Q'}=0\text{ if }Q'\nsubseteq Q\big\}$(See \eqref{eq: E_Q}). For $\bbeta\in E_Q$ and any $S\subseteq[K]$,
\[
(\MY\beta)_S=\eta_S=\sum_{T\subseteq S}\beta_T=\sum_{T\subseteq S\cap Q}\beta_T,
\]
so the image can be written as
\[
Y_Q:=\MY E_Q
=\Big\{\,y\in\RR^{2^K}:\ \exists\,h:\mathcal P(Q)\to\RR\ \text{such that }\
y_S=h(S\cap Q)\ \ \forall\,S\subseteq[K]\Big\}.
\]
Similarly,
\[
Y_{\pi(Q)}
=\Big\{\,y\in\RR^{2^K}:\ \exists\,\tilde h:\mathcal P(\pi(Q))\to\RR\ \text{such that }\
y_S=\tilde h(S\cap \pi(Q))\ \ \forall\,S\subseteq[K]\Big\}.
\]
Since $\MY$ is invertible, $\MY^{-1}\MP_f\MY E_Q\subseteq E_{\pi(Q)}$ if and only if $\MP_fY_Q\subseteq Y_{\pi(Q)}.$

Thus $\MP_f\in G_\pi(Q)$ if and only if $\MP_f$ maps vectors that are constant on each $B^{Q}(U)$ to vectors that are constant on each $B^{\pi(Q)}(V)$. This holds if and only if $f$ maps each block $B^{Q}(U)$, as a set, onto some block $B^{\pi(Q)}(V)$. Because $\{B^{Q}(U)\}_{U\in\mathcal P(Q)}$ is a partition and $f$ is a bijection, these images define a unique bijection $g:\mathcal P(Q)\to\mathcal P(\pi(Q))$ with
\[
f\big(B^{Q}(U)\big)=B^{\pi(Q)}\big(g(U)\big)\qquad\forall\,U\in\mathcal P(Q).
\]
This proves \((a)\) and \((b)\) are equivalent.
\end{proof}

\begin{proposition}\label{prop:lower}
For any $\mathcal F$,
$(\mathbb Z_2)^K\pi \subseteq H_\pi(\mathcal F).$
\end{proposition}
\begin{proof}
Fix $A\subseteq[K]$ and define $\tau_A:\mathcal P([K])\to\mathcal P([K])$ by
\[
\tau_A(T):=\pi(T)\ \Delta\ A.
\]
We claim that for every $Q\subseteq[K]$ and every $U\in\mathcal P(Q)$,
\[
\tau_A\big(B^{Q}(U)\big)\ =\ B^{\pi(Q)}\!\big(\,\pi(U)\ \Delta\ (A\cap \pi(Q))\,\big).
\]
Indeed, if $T\in B^{Q}(U)$ then $T\cap Q=U$, and by the identity 
$(X\Delta Y)\cap Z=(X\cap Z)\Delta(Y\cap Z)$ we have
\[
\big(\pi(T)\Delta A\big)\cap\pi(Q) = \big(\pi(T)\cap \pi(Q)\big) \Delta \big(A\cap \pi(Q)\big)
 = \pi(T\cap Q) \Delta\big(A\cap \pi(Q)\big)
 = \pi(U) \Delta \big(A\cap \pi(Q)\big),
\]
which is independent of the particular $T$ in the block and depends only on $U$.
Thus $\tau_A$ maps the block $B^{Q}(U)$ onto the block 
$B^{\pi(Q)}\!\big(\pi(U)\Delta(A\cap \pi(Q))\big)$, so the associated permutation
matrix $P_{\tau_A}$ satisfies $P_{\tau_A}\in G_\pi(Q)$ by Proposition~\ref{prop:single-Q}.
Since $Q\subseteq[K]$ was arbitrary, we have $$P_{\tau_A}\in \bigcap_{Q\in\mathcal F}G_\pi(Q)=H_\pi(\mathcal F).$$

Finally, the set $\{\tau_A: A\subseteq[K]\}$ is exactly the right coset $(\mathbb Z_2)^K\pi$, hence $(\mathbb Z_2)^K\pi\subseteq H_\pi(\mathcal F)$.
\end{proof}
\begin{proposition}\label{prop:upper}
Assume that for any distinct $i,j\in[K]$, neither $\phi_{\mathcal F}(i)\preceq \phi_{\mathcal F}(j)$ nor $\phi_{\mathcal F}(j)\preceq \phi_{\mathcal F}(i)$ holds, then we have
$$ H_\pi(\mathcal F)\subseteq (\mathbb Z_2)^K\pi.$$
\end{proposition}
\begin{proof}
Throughout we identify a permutation matrix $\MP_f$ with its underlying bijection 
$f:\mathcal P([K])\to\mathcal P([K])$ via $(\MP_f v)_S=v_{f^{-1}(S)}$.
Fix $i\in[K]$ and define, for $T\subseteq[K]$,
\[
D_i(T):=f(T)\ \Delta\ f\bigl(T\Delta\{\pi^{-1}(i)\}\bigr).
\]
We now state the following proposition, which characterizes $D_i(T)$ and will be used below.
\begin{proposition}\label{prop:Di-characterization}
Under the assumptions of Proposition~\ref{prop:upper}, for every $\MP_f\in H_\pi(\mathcal F)$, every $i\in[K]$, and every $T\subseteq[K]$, one has
\[
D_i(T)=\{\,i\,\}.
\]
\end{proposition}

With Proposition~\ref{prop:Di-characterization}, take any $j\in[K]$ and $T\subseteq[K]$. 
Substituting $i=\pi(j)$ in the definition of $D_i(T)$ yields
\[
f\bigl(T\Delta\{j\}\bigr)
\;=\;
f(T)\ \Delta\ \{\pi(j)\}.
\]
Let $A:=f(\varnothing)$, we will prove that $f(T)=A\ \Delta\ \pi(T)$ by induction on $m:=|T|$.

\begin{enumerate}
    \item{Base case $m=0$.} For $T=\varnothing$, we have $f(\varnothing)=A=A\ \Delta\ \pi(\varnothing)$.
    \item{Inductive step.}
Assume $f(S)=A\ \Delta\ \pi(S)$ holds for all $S\subseteq[K]$ with $|S|=m$. Let $T=S\cup\{j\}$ with $j\notin S$. By Proposition~\ref{prop:Di-characterization} with $i=\pi(j)$, 
\[
f(S\Delta\{j\})=f(S)\ \Delta\ \{\pi(j)\}.
\]
Since $S\Delta\{j\}=S\cup\{j\}=T$ and $\pi(T)=\pi(S\cup\{j\})=\pi(S)\ \Delta\ \{\pi(j)\}$, we obtain
\[
f(T)=f(S)\ \Delta\ \{\pi(j)\}
=\bigl(A\ \Delta\ \pi(S)\bigr)\ \Delta\ \{\pi(j)\}
= A\ \Delta\ \bigl(\pi(S)\ \Delta\ \{\pi(j)\}\bigr)
= A\ \Delta\ \pi(T).
\]
This completes the induction. 
\end{enumerate}
 Hence for every $\MP_f\in H_\pi(\mathcal F)$, the underlying bijection has the form $f(T)=A\Delta\pi(T)$. In other words, $H_\pi(\mathcal F)\subseteq(\mathbb Z_2)^K\pi$. 

\end{proof}
Combining Proposition~\ref{prop:upper} with Proposition~\ref{prop:lower}, we complete the proof of Lemma~\ref{lem:main}

Based on this lemma, we can obtain the sum of each column of $\MB'$. Let $\mathbf 1$ be the all-ones vector. For any $S\subseteq[K]$ and $j\in\{1,\dots,J\}$, using $\mathbf 1^\top \MY^{-1}=e_{[K]}^\top$, it follows that
\[
\sum_S \beta'_{j,S}
=\mathbf 1^\top \bbeta'
=e_{[K]}^\top \MP_f \MY\bbeta_j
=e_{\,f([K])}^\top \MY\bbeta_j
=\sum_{U\subseteq f([K])}\beta_{j,U}.
\]
Since $f([K])=\pi([K])\Delta A=[K]\setminus A$, we obtain:

$$\sum_{S}\beta'_{j,S}
=\sum_{U\subseteq [K]\setminus A}\ \beta_{j,U}.$$

Now since both $(\MB,\MQ)$ and $(\MB',\MQ')$ both satisfy Assumption~\ref{assm1}(b), we must have 
\[
\sum_S \beta'_{j,S}= \sum_S \beta_{j,S}\quad j=1,\dots,J.
\]This forces $A\cap K_j=\varnothing$ for all $j=1,\dots,J$. Consequently, $A=\varnothing$ and $f(T)=\pi(T)$ is the only admissible permutation.

\subsubsection{Proof of Proposition~\ref{prop:Di-characterization}}
We first claim that, for all $f\in H_\pi(\mathcal F)$ and all $i,T$, \[ D_i(T)\ \subseteq\ \bigcap_{\substack{Q\in\mathcal F\\ \pi^{-1}(i)\notin Q}} (\pi(Q))^{c} =\{\,j\in[K]:\ \phi_\calF(\pi^{-1}(j))\preceq \phi_\calF(\pi^{-1}(i))\}. \]Note that by the assumption that for every $i\in[K]$ there must exist some $Q\in\mathcal F$ such that $i\notin Q$, the index set 
$\{\,Q\in\mathcal F:\ \pi^{-1}(i)\notin Q\,\}$
is nonempty. Hence the intersection is taken over a nonempty family

In fact, let $Q\in\mathcal F$ satisfy $\pi^{-1}(i)\notin Q$. Then $T$ and $T\Delta\{\pi^{-1}(i)\}$ lie in the same $Q$-block:
\[
T\cap Q \;=\; (T\Delta\{\pi^{-1}(i)\})\cap Q.
\]
Since $f\in H_\pi(\mathcal F)$, Proposition~\ref{prop:single-Q} implies that $f$ maps each $Q$-block onto a $\pi(Q)$-block. Hence
\[
f(T)\cap \pi(Q) \;=\; f\bigl(T\Delta\{\pi^{-1}(i)\}\bigr)\cap \pi(Q),
\]
and therefore
\[
D_i(T)\cap \pi(Q)\;=\;\varnothing.
\]
As this holds for every $Q\in\mathcal F$ with $\pi^{-1}(i)\notin Q$, we obtain
\begin{equation}\label{eq: ast}
D_i(T)\ \subseteq\ \bigcap_{\substack{Q\in\mathcal F\\ \pi^{-1}(i)\notin Q}}\big(\pi(Q)\big)^c.
\end{equation}

Next we rewrite the right-hand side using signature vectors. Fix $j\in[K]$, $$j\in \bigcap_{\{Q\in\mathcal F:\,\pi^{-1}(i)\notin Q\}}(\pi(Q))^{c}$$ means precisely that for every $Q\in\mathcal F$ with $\pi^{-1}(i)\notin Q$ we have $j\notin \pi(Q)$. 
By the definition of the signature map, the condition $\pi^{-1}(i)\notin Q$ is the same as $\phi_{\mathcal F}(\pi^{-1}(i))_Q=0$, and $j\notin \pi(Q)$ is the same as $\phi_{\mathcal F}(\pi^{-1}(j))_Q=0$. 
Therefore the preceding sentence says: for all $Q\in\mathcal F$, 
$\phi_{\mathcal F}(\pi^{-1}(j))\preceq \phi_{\mathcal F}(\pi^{-1}(i))$. 
Hence
\[
\bigcap_{\substack{Q\in\mathcal F\\ \pi^{-1}(i)\notin Q}}\big(\pi(Q)\big)^c
=\bigl\{\,j\in[K]:\ \phi_{\mathcal F}(\pi^{-1}(j))\preceq \phi_{\mathcal F}(\pi^{-1}(i))\,\bigr\},
\]
and the claim follows from \eqref{eq: ast}.

By assumption, for each $i$ we have
\[
\bigl\{\,j\in[K]:\ \phi_{\cal F}(\pi^{-1}(j))\preceq \phi_{\cal F}(\pi^{-1}(i))\,\bigr\}
\;=\;\{\,i\,\}.
\]
Therefore, by \eqref{eq: ast} we already proved above,
\[
D_i(T)\ \subseteq\
\bigl\{\,j:\ \phi_{\cal F}(\pi^{-1}(j))\preceq \phi_{\cal F}(\pi^{-1}(i))\,\bigr\}
\;=\;\{i\}.
\]
Since $f$ is a bijection and $T\neq T\Delta\{\pi^{-1}(i)\}$, we have
$D_i(T)\neq\varnothing$, hence $D_i(T)=\{i\}$.

\subsection{Extension to the polytomous-attribute case}\label{app:poly_extension}
We extend the binary-latent-attribute framework in Section~\ref{sec: model} to the case where each latent attribute is polytomous with possibly different numbers of categories.
Fix integers \(M_k\ge 2\) for \(k\in[K]\) and let \(\MZ=(Z_1,\ldots,Z_K)\) take values in \(\prod_{k=1}^K [M_k]\) where \([M_k]=\{0,1,\ldots,M_k-1\}\).
The latent law \(\p\) is generated by a categorical Bayesian network on a latent DAG \(\calG\).
Conditionally on \(\MZ\), the items are independent and each \(X_j\mid \MZ\) follows the same item-specific family and link specification as in \eqref{eq:observed exp fam}.

Similar to the binary case, we still use \(\MQ=(q_{j,k})\in\{0,1\}^{J\times K}\) to encode the bipartite measurement graph, where \(q_{j,k}=1\) means that the latent variable \(Z_k\) is a direct cause of the observed variable \(X_j\), and set \(K_j:=\{k\in[K]:q_{j,k}=1\}\).
For any latent configuration vector \(\bu\in \prod_{k=1}^K [M_k]\), define \(\mathrm{supp}(\bu):=\{k\in[K]:u_k\ge 1\}\).
For each \(j\in[J]\), introduce coefficients \(\{\beta_{j,\bu}\}\) indexed by \(\bu\in \prod_{k=1}^K [M_k]\), and define the linear predictor for every latent state \(\bz\in \prod_{k=1}^K [M_k]\) by \(\eta_j(\bz):=
\sum_{\bu: \mathrm{supp}(\bu)\subseteq K_j, \bu\preceq \bz}\beta_{j,\bu}\). Equivalently, \(\eta_j(\bz)\) is a linear combination of the coefficients \(\{\beta_{j,\bu}\}\), and a term \(\beta_{j,\bu}\) contributes to \(\eta_j(\bz)\) only when \(\mathrm{supp}(\bu)\subseteq K_j\) and \(\bu\preceq \bz\).
Collect \(\bbeta_j=(\beta_{j,\bu})_{\bu\in \prod_{k=1}^K [M_k]}\in\RR^{r}\) and let \(\MB=(\bbeta_1^\top,\ldots,\bbeta_J^\top)^\top\in\RR^{J\times r}\).

With this parameterization in place, \(\MQ\) still admits a direct causal interpretation in the polytomous-attribute regime: if \(q_{j,k}=0\), then varying \(Z_k\) while holding the other latents fixed does not change \(\eta_j(\bz)\) and hence does not affect the conditional law of \(X_j\); whereas if \(q_{j,k}=1\), increasing \(Z_k\) from level \(s\) to \(s+1\) activates new contributions in \(\eta_j(\bz)\), including a level-\((s+1)\) main-effect contribution of \(Z_k\) and potentially additional interaction contributions with other latent causes in \(K_j\) that become available only after \(Z_k\) reaches level \(s+1\).
Because these newly activated effects are assigned their own parameters rather than being constrained to scale linearly across levels, the causal effect of \(Z_k\) on \(X_j\) can vary across levels and across configurations of the other latents, making the model highly flexible.

Similar to \cite{he2023a}, to encode the intrinsic ordering of levels, we also use cumulative threshold indicators
$\calI_{k,u}(\MZ)=\mathbf 1\{Z_k\ge u\}$ for $u\in\{1,\ldots,M_k-1\}$ and form the Kronecker feature map
$$\Phi(\MZ)=\otimes_{k=1}^K (1,\calI_{k,1}(\MZ),\ldots,\calI_{k,M_k-1}(\MZ))^\top\in\RR^{r}.$$
Then each item $j$ has a predictor
\[\eta_j(\MZ)=\langle \bbeta_j,\Phi(\MZ)\rangle=
\sum_{\bu\in\calZ} \beta_{j,\bu} \prod_{k=1}^K \calI_{k,u_k}(\MZ),
\] where $\bbeta_j:=\bigl(\beta_{j,\bu}\bigr)_{\bu\in\calZ}\in\RR^{r}$. 

We further impose the structural restriction that coefficients vanish outside the relevant coordinates,
\[
\beta_{j,\bu}=0
\qquad\text{whenever}\qquad
\mathrm{supp}(\bu)\nsubseteq K_j,
\]
and require a full main-effect ladder whenever $q_{j,k}=1$.
Specifically, for each $k\in[K]$ and each threshold $v\in\{1,\ldots,M_k-1\}$ define the main-effect index
\(
\bu(k,v)=(u_h)_{h\in[K]}
\)
with $u_k=v$ and $u_h=0$ for $h\neq k$.
We assume
\[
q_{j,k}=1
\ \Longleftrightarrow\
\beta_{j,\bu(k,v)}\neq 0\ \ \text{for all }v\in\{1,\ldots,M_k-1\},
\qquad (j\in[J],\ k\in[K]),
\]
while interaction coordinates $\bu$ with $|\mathrm{supp}(\bu)|\ge 2$ and $\mathrm{supp}(\bu)\subseteq K_j$ are allowed to be nonzero.

Because the latent variables may have unequal numbers of categories, admissible relabelings must preserve numbers of categories.
Define the permutation group
\begin{equation}\label{eq:poly_S_KM}
S_{K,\mathbf M}
:=
\Bigl\{\varpi\in S_K:\ M_{\varpi(k)}=M_k\ \ \forall k\in[K]\Bigr\}.
\end{equation}
Given $\varpi\in S_{K,\mathbf M}$, define the induced bijection $\sigma_{\varpi}:\calZ\to\calZ$ by
\begin{equation}\label{eq:poly_sigma_varpi}
\sigma_{\varpi}(z_1,\ldots,z_K)
:=
\bigl(z_{\varpi^{-1}(1)},\ldots,z_{\varpi^{-1}(K)}\bigr),
\end{equation}
and let $P_{\varpi}$ be the induced $r\times r$ Kronecker permutation matrix such that
\begin{equation}\label{eq:poly_perm_matrix}
\Phi\bigl(\sigma_{\varpi}(\bz)\bigr)=P_{\varpi}\,\Phi(\bz),
\qquad \bz\in\calZ.
\end{equation}
We say that $(\p,\calG,\MQ,\{\bbeta_j\}_{j\in[J]},\bgamma)$ and $(\p',\calG',\MQ',\{\bbeta'_j\}_{j\in[J]},\bgamma')$
are equivalent, denoted $\sim_{K,\mathbf M}^{\mathrm{ord}}$, if and only if $\bgamma=\bgamma'$ and there exists
$\varpi\in  S_{K,\mathbf M}$ such that, with $\sigma=\sigma_{\varpi}$,
\begin{equation}\label{eq:poly_equiv_Q}
p_{\bz}=p'_{\sigma(\bz)}\quad\forall \bz\in\calZ,
\qquad
\bbeta'_j=P_{\varpi}\,\bbeta_j\quad\forall j\in[J],
\qquad
q'_{j,k}=q_{j,\varpi^{-1}(k)}\quad\forall (j,k)\in[J]\times[K],
\end{equation}
and the relabeled DAG $\varpi(\calG)$ is Markov equivalent to $\calG'$.

Let $\MTheta:=(\p,\MB,\bgamma)$.
Define the admissible parameter space
\begin{align*}
&\Omega_{K}(\MTheta;\calG,\MQ):=
\Bigl\{
\MTheta:\ 
\calG \ \text{is a perfect map of }\p,\ \beta_{j,\bu}=0
\ \text{if}\
\mathrm{supp}(\bu)\nsubseteq K_j,
\\&q_{j,k}=1\
\text{iff}\
\beta_{j,\bu(k,v)}\neq 0\ \ \text{for all }v\in\{1,\ldots,M_k-1\}
\Bigr\},\\
&\Omega_{K}(\MTheta,\calG,\MQ)
:=
\Bigl\{
(\MTheta,\calG,\MQ):\ \MTheta\in\Omega_{K}(\MTheta;\calG,\MQ)
\Bigr\}.
\end{align*}

\begin{definition}\label{def:generic identifiability-polytomous-Q}
Let $(\MTheta^\star,\calG^\star,\MQ^\star)\in\Omega_{K}(\MTheta,\calG,\MQ)$ be the true parameter triple.
The framework is \emph{generically identifiable} up to $\sim_{K,\mathbf M}^{\mathrm{ord}}$ if
\begin{equation}\label{eq:generic_ident_Q}
\Bigl\{
\MTheta\in\Omega_{K}(\MTheta;\calG^\star,\MQ^\star):
\ \exists\,(\widetilde{\MTheta},\widetilde{\calG},\widetilde{\MQ})\not\sim_{K,\mathbf M}^{\mathrm{ord}}(\MTheta,\calG^\star,\MQ^\star)
\ \text{such that }\PP_{\widetilde{\MTheta},\widetilde{\calG},\widetilde{\MQ}}=\PP_{\MTheta,\calG^\star,\MQ^\star}
\Bigr\}
\end{equation}
is a Lebesgue-null subset of $\Omega_{K}(\MTheta;\calG^\star,\MQ^\star)$.
\end{definition}

Now we are ready to state our identifiability result in polytomous-attribute case. Conditions in Theorem~\ref{thm:poly_generic_ident_delta} are parallel to those in Theorem~\ref{thm2}, and Assumption~\ref{assm:poly_monotone_binary_style} is the ordered-level analogue of the monotonicity in Assumption~\ref{assm1}(b). 
With these in place, we obtain generic identifiability for the polytomous-attribute extension of DCRL. It is worth pointing out that, if $M_k=2$ for all $k$ (i.e., the binary-attribute case), then all assumptions and conditions in Theorem~\ref{thm:poly_generic_ident_delta} reduce to Theorem~\ref{thm2} exactly.
\begin{assumption}\label{assm:poly_monotone_binary_style}
For each item $j$, define the top-set $\mathcal M_j:=
\{\MZ\in\calZ: Z_k=M_k-1\ \text{for all }k\in K_j\}$. Assume $\max_{\MZ\in\calZ\setminus\mathcal M_j} \eta_j(\MZ)<
\min_{\MZ\in\mathcal M_j} \eta_j(\MZ)$
for $j\in[J]$.
Moreover, for each item $j$, each $k\in K_j$, and each threshold $u\in[M_k-1]$, define $\mathcal T^{(k,u)}_{1,j}:=\{\MZ\in\calZ: Z_k\ge u,\ Z_h=M_h-1 \ \forall h\in K_j\setminus\{k\}\}$ and $\mathcal T^{(k,u)}_{0,j}:=\bigl\{\MZ\in\calZ: Z_k\le u-1,\ Z_h=M_h-1\ \ \forall h\in K_j\setminus\{k\}\bigr\}$,
and assume $\max_{\MZ\in\mathcal T^{(k,u)}_{0,j}} \eta_j(\MZ)<\min_{\MZ\in\mathcal T^{(k,u)}_{1,j}} \eta_j(\MZ).$
\end{assumption}

\begin{theorem}\label{thm:poly_generic_ident_delta}
Under Assumption~\ref{assm1}(a), Assumption~\ref{assm2}, and Assumption~\ref{assm:poly_monotone_binary_style}, the polytomous-attribute DCRL is generically identifiable if the following conditions hold.
\begin{enumerate}
\item[(i)]
For each \(k\in[K]\), let \(m_k=\lceil \log_2 M_k\rceil\) and \(\tilde d=\sum_{k=1}^K m_k\).
After a row permutation, \(\MQ^\star=[\MQ_1^\top,\MQ_2^\top,\MQ_3^\top]^\top\) with \(\MQ_1,\MQ_2\in\{0,1\}^{\tilde d\times K}\). For each \(a\in\{1,2\}\), \(\MQ_a=[\MQ_{a,1}^\top,\ldots,\MQ_{a,K}^\top]^\top\), where \(\MQ_{a,k}\in\{0,1\}^{m_k\times K}\) and every row of \(\MQ_{a,k}\) has a \(1\) in column \(k\). Other entries are arbitrary. Moreover, \(\MQ_3\) has no all-zero column.

\item[(ii)]For any $p\neq q$, neither $\MQ_{:,p}\succeq \MQ_{:,q}$ nor $\MQ_{:,q}\succeq \MQ_{:,p}$.
\end{enumerate}
\end{theorem}
Condition (i) implies that at least $2\sum_{k=1}^K(\lceil \log_2 M_k\rceil)+1$ items are required to achieve generic identifiability. For the binary-attribute case, condition (i) reduces to $J\ge 2K+1$, matching the binary generic-identifiability requirement in Theorem~\ref{thm2}. The logarithmic order \(m_k=\lceil \log_2 M_k\rceil\) here is optimal up to constants and the ceiling since this theorem is required to hold uniformly over the general response families. It suffices to consider the binary-response submodel contained in our framework. Fix a coordinate \(k\), and consider any collection of \(l\) items with \(q_{j,k}=1\). Write their joint response as \(\MX\in\{0,1\}^l\).
Choose a latent DAG \(\calG\) in which \(Z_k\) is a root so that its marginal \(\p^{(k)}\) can vary freely on an open set, and fix all remaining latent parameters as well as parameters \((\MB,\bgamma)\).
Then for each \(s\in\{0,\dots,M_k-1\}\), the conditional law of \(\MX\) given \(Z_k=s\) is a vector \(\bv_s\in\RR^{2^l}\), and the marginal law of \(Y\) is the mixture
$\PP(\MX=\cdot)=\sum_{s=0}^{M_k-1}   p^{(k)}_s \bv_s$,
so the map \(  \p^{(k)}\mapsto \PP(\MX=\cdot)\) is linear.
If \(2^l<M_k\), this linear map cannot be injective on any open set: there exists \(0\ne\bh\in\RR^{M_k}\) with \(\sum_s h_s=0\) such that \(\sum_s h_s \bv_s=0\), and hence for every interior \(  \p^{(k)}\) and all sufficiently small \(\varepsilon\), the two distinct marginals \(  \p^{(k)}\) and \(  \p^{(k)}+\varepsilon \bh\) induce the same law of \(\MX\). Thus the set of parameters yielding non-identifiability contains an open subset, so generic identifiability fails in this binary submodel unless \(2^l\ge M_k\). Consequently, the logarithmic scaling in condition~(i) is sharp in magnitude for the general-response setting, reflecting that polytomous attributes with more categories require more items to achieve identifiability.


For condition (i) in Theorem~\ref{thm:poly_generic_ident_delta}, we can equivalently rewrite it as follows:
\begin{enumerate}
\item[(i)]
The item index set can be partitioned as $\calJ_1=\{1,\ldots,\tilde d\}$, $\calJ_2=\{\tilde d+1,\ldots,2\tilde d\}$, and $\calJ_3=\{2\tilde d+1,\ldots,J\}$,
and there exist bijections: $\rho_1:\calJ_1\to\{(k,s):k\in[K],\ s\in[m_k]\}$ and $\rho_2:\calJ_2\to\{(k,s):k\in[K],\ s\in[m_k]\}$.
For each $a\in\{1,2\}$ and each $j\in\calJ_a$, write $\rho_a(j)=(k(j),s(j))$.
$q_{j,k(j)}=1$ for all $j\in\calJ_1\cup\calJ_2$,
and for each $k\in[K]$ there exists $j_k\in\calJ_3$ such that
$q_{j_k,k}=1.$
\end{enumerate}
We will use \(\rho_1\) and \(\rho_2\) later to index the items in the first two blocks.
\subsection{Proof of Theorem~\ref{thm:poly_generic_ident_delta}}

\subsubsection{Kruskal reduction via a parameter-independent refinement scheme}
Fix an enumeration of $\calC_j^\mathrm{can}\setminus\{\calX_j\}$ as
\[
\calC_j^\mathrm{can}\setminus\{\calX_j\}=\{T_{1,j},T_{2,j},\dots\},
\qquad (j\in[J]).
\]

For each $t\ge 1$, define the parameter-independent finite discretization
\[
\bar{\mathcal D}_j^{(t)}:=\{T_{1,j},\dots,T_{t,j}\}\cup\{\calX_j\}\subseteq \calC_j^\mathrm{can},
\qquad
\bar{\mathcal D}^{(t)}:=(\bar{\mathcal D}_j^{(t)})_{j\in[J]}.
\]
Let $\kappa_j^{(t)}:=|\bar{\mathcal D}_j^{(t)}|=t+1$ and index
\[
\bar{\mathcal D}_j^{(t)}=(S^{(t)}_{1,j},\dots,S^{(t)}_{\kappa_j^{(t)},j})
\qquad\text{with}\qquad S^{(t)}_{\kappa_j^{(t)},j}=\calX_j.
\]

For each $t\ge 1$, define factor matrices $\MN_1^{(t)},\MN_2^{(t)},\MN_3^{(t)}$ whose columns are indexed by $\bz\in\calZ$ as follows.

For $\xi_1=(\ell_j)_{j\in\calJ_1}$ with $\ell_j\in[\kappa_j^{(t)}]$ for all $j\in\calJ_1$, set
\[
\MN_1^{(t)}(\xi_1,\bz)
:=
\PP\Bigl(\bigcap_{j\in\calJ_1}\{X_j\in S_{\ell_j,j}^{(t)}\}\,\Big|\, \bz\Bigr).
\]
For $\xi_2=(\ell_j)_{j\in\calJ_2}$ with $\ell_j\in[\kappa_j^{(t)}]$ for all $j\in\calJ_2$, set
\[
\MN_2^{(t)}(\xi_2,\bz)
:=
\PP\Bigl(\bigcap_{j\in\calJ_2}\{X_j\in S_{\ell_j,j}^{(t)}\}\,\Big|\, \bz\Bigr).
\]
For $\xi_3=(\ell_j)_{j\in\calJ_3}$ with $\ell_j\in[\kappa_j^{(t)}]$ for all $j\in\calJ_3$, set
\[
\MN_3^{(t)}(\xi_3,\bz)
:=
\PP\Bigl(\bigcap_{j\in\calJ_3}\{X_j\in S_{\ell_j,j}^{(t)}\}\,\Big|\, \bz\Bigr).
\]
Define the tensor $\mathbf P_0^{(t)}$ by
\[
\mathbf P_0^{(t)}(\xi_1,\xi_2,\xi_3)
=
\PP\Bigl(\bigcap_{j=1}^J \{X_j\in S^{(t)}_{\ell_j,j}\}\Bigr),
\qquad
\xi_a=(\ell_j)_{j\in\calJ_a},
\]
where $(\ell_j)_{j=1}^J$ is the concatenation of $\xi_1,\xi_2,\xi_3$ in the natural item order.
By conditional independence given $\bz$ and the law of total probability,
\begin{align}\label{eq:poly_tensor_decomp_paramfree_fixed}
\mathbf P_0^{(t)}(\xi_1,\xi_2,\xi_3)
&=
\sum_{\bz\in\calZ}p_{\bz}\,
\MN_1^{(t)}(\xi_1,\bz)\,
\MN_2^{(t)}(\xi_2,\bz)\,
\MN_3^{(t)}(\xi_3,\bz),
\\
\mathbf P_0^{(t)}
&=
\big[\,\MN_1^{(t)}\mathrm{Diag}(\p),\,\MN_2^{(t)},\,\MN_3^{(t)}\,\big].
\nonumber
\end{align}
Since $S^{(t)}_{\kappa_j^{(t)},j}=\calX_j$, each $\MN_a^{(t)}$ contains a row equal to $\bone_r^\top$.

\begin{lemma}\label{lem:poly_kruskal_gid_paramfree}
Consider a discrete causal representation learning framework with parameters
$(\p^\star,\calG^\star,\MB^\star,\MQ^\star,\bgamma^\star)$ satisfying the conditions of Theorem~\ref{thm:poly_generic_ident_delta}.
For each $t\ge1$, let $\mathbf P_0^{(t)}$ be the tensor induced by the parameter-independent discretization
$\bar{\mathcal D}^{(t)}$ defined above, with factor matrices $\MN_1^{(t)},\MN_2^{(t)},\MN_3^{(t)}$, so that $\mathbf P_0^{(t)}=[\MN_1^{(t)}\mathrm{Diag}(\p),\MN_2^{(t)},\MN_3^{(t)}]$. Then there exists a Lebesgue-null set
$\mathcal N_\infty\subset\Omega_K(\MTheta;\calG^\star,\MQ^\star)$
which constrains only $(\MB,\bgamma)$ such that the following holds.

For every $\MTheta\in\Omega_K(\MTheta;\calG^\star,\MQ^\star)\setminus\mathcal N_\infty$,
there exists an integer $t_0=t_0(\MTheta)<\infty$ such that for all $t\ge t_0$
the rank-$r$ CP decomposition of $\mathbf P_0^{(t)}$ is unique up to a common column permutation.
Moreover, since $\calX_j\in\bar{\mathcal D}_j^{(t)}$, each $\MN_a^{(t)}$ contains a row equal to
$\mathbf 1_{r}^\top$, hence the uniqueness contains no nontrivial scaling ambiguity.
\end{lemma}

\begin{proof}
The argument is parallel to the proof of Lemma~\ref{lem:kruskal gid} in the binary-attribute setting. The only places changed are the following.
First, the latent state space is $\calZ=\prod_{k=1}^K [M_k]$ of size $r=\prod_{k=1}^K M_k$,
so every instance of $2^K$ in Lemma~\ref{lem:kruskal gid} is replaced by $r$ here.
Second, the witness construction that certifies $\mathrm{rank}(\MN_1^{(t)})=r$ uses Condition (i) in Theorem~\ref{thm:poly_generic_ident_delta} to obtain a Kronecker product block built from
$m_k=\lceil\log_2 M_k\rceil$ items per coordinate $k$, rather than the $2\times2$ blocks
in the binary argument.

Finally, the discussion of dispersion parameters $\gamma$ is exactly the same as in Lemma~\ref{lem:kruskal gid}.
When an item family is one-parameter, we  do not treat $\gamma_j$
as a free coordinate.
When $\gamma_j$ is genuinely present, Assumption~\ref{assm2} ensures that
$(\eta,\gamma)\mapsto \PP_{j,g_j(\eta,\gamma)}(S)$ is real-analytic on $\RR\times(0,\infty)$ for each fixed
$S\in\bar{\mathcal D}_j^{(t)}$. For simplicity, we only treat the second case here.

Fix $t\ge 1$.
By Kruskal's theorem, uniqueness holds provided that
\[
\mathrm{rk}_k(\MN_1^{(t)})+\mathrm{rk}_k(\MN_2^{(t)})+\mathrm{rk}_k(\MN_3^{(t)})\ \ge\ 2r+2.
\]
Thus it suffices to establish, outside a Lebesgue-null exceptional set,
\begin{equation}\label{eq:poly_kruskal_target_paramfree}
\mathrm{rk}_k(\MN_1^{(t)})=r,\qquad \mathrm{rk}_k(\MN_2^{(t)})=r,\qquad \mathrm{rk}_k(\MN_3^{(t)})\ge 2,
\end{equation}
for all sufficiently large $t$.

\smallskip
\emph{(i) Generic full column rank for $\MN_1^{(t)}$ and $\MN_2^{(t)}$.}

For each $i\in\calJ_1$ define
\begin{equation}\label{eq:UiQ_def}
\calU_i^{\MQ}
:=
\{\mathbf u\in\calZ:\ \mathrm{supp}(\mathbf u)\subseteq K_i\}.
\end{equation}
Let
\begin{equation}\label{eq:I1Q_def}
\calI_1^{\MQ}
:=
\{(i,\mathbf u):\ i\in\calJ_1,\ \mathbf u\in\calU_i^{\MQ}\},
\qquad
q_1:=|\calI_1^{\MQ}|.
\end{equation}
As before, the block parameter vector $\Theta_1$ collects exactly the free coefficients
$\bbeta_1:=(\beta_{i,\mathbf u})_{(i,\mathbf u)\in\calI_1^{\MQ}}\in\RR^{q_1}$ and the dispersion parameters
$\bgamma_1\in(0,\infty)^{|\calJ_1^\gamma|}$, where $\calJ_1^\gamma$ denote the indices of items among the first block whose response family genuinely includes an unknown dispersion parameter $\gamma_j$. In other words, we identify
\begin{equation}\label{eq:Theta1_chart_Q}
\Theta_1=(\bbeta_1,\bgamma_1)\ \in\ \widetilde\Omega_1(\MQ)
:=
\RR^{q_1}\times(0,\infty)^{|\calJ_1^\gamma|}.
\end{equation}
Fix $t$ and view $\MN_1^{(t)}$ as a function of $\Theta_1$.
By Assumption~\ref{assm2}, for each fixed row index $\xi_1$ and column $\bz$, the entry $\MN_1^{(t)}(\xi_1,\bz)$ is a finite product of analytic maps of the form
\[
(\eta,\gamma)\ \longmapsto\ \PP_{j,\; g_j(\eta,\gamma)}(S),
\qquad S\in\bar{\mathcal D}_j^{(t)},
\]
evaluated at analytic functions of the local item parameters in $\Theta_1$.
Hence every entry of $\MN_1^{(t)}$ is real-analytic in the ambient Euclidean coordinates of $\Theta_1$, and so is $f_{1,t}(\Theta_1):=\det((\MN_1^{(t)}(\Theta_1))^\top \MN_1^{(t)}(\Theta_1))$.

Therefore $f_{1,t}$ is real-analytic on the open connected domain $\widetilde\Omega_1(\MQ)$, and hence either
$f_{1,t}\equiv 0$ on $\widetilde\Omega_1(\MQ)$, or else its zero set is Lebesgue-null \citep{mityagin2015the}.

The feasible block-$\calJ_1$ set in $\Omega_K(\MTheta;\calG^\star,\MQ^\star)$still imposes the nonzero ladder constraint for every main effect.
Accordingly define
\begin{equation}\label{eq:Omega1Q_def}
\Omega_1(\MQ)
:=
\Bigl\{\Theta_1\in\widetilde\Omega_1(\MQ):\
\beta_{i,\mathbf u(k,v)}\neq 0\ \ \forall i\in\calJ_1,\ \forall k\in K_i,\ \forall v\in[M_k-1]\Bigr\}.
\end{equation}
In $\Theta_1$-coordinates, $\Omega_1(\MQ)$ is obtained from $\widetilde\Omega_1(\MQ)$ by removing finitely many coordinate hyperplanes
$\{\beta_{i,\mathbf u(k,v)}=0\}$, hence $\Omega_1(\MQ)$ is dense in $\widetilde\Omega_1(\MQ)$ and has full Lebesgue measure in it.

Consequently, it suffices to construct a single witness point
$\Theta_{1,\mathrm{wit}}\in\widetilde\Omega_1(\MQ)$ such that $f_{1,t}(\Theta_{1,\mathrm{wit}})>0$ for all sufficiently large $ t$.
Indeed, that implies $f_{1,t}\not\equiv 0$ on $\widetilde\Omega_1(\MQ)$, hence
$\{\Theta_1\in\widetilde\Omega_1(\MQ):\ f_{1,t}(\Theta_1)=0\}$ is Lebesgue-null.
Therefore,
\begin{equation}\label{eq:null_intersection_Q}
\{\Theta_1\in\Omega_1(\MQ): f_{1,t}(\Theta_1)=0\}
=
\{\Theta_1\in\widetilde\Omega_1(\MQ): f_{1,t}(\Theta_1)=0\}\cap \Omega_1(\MQ),
\end{equation}
is Lebesgue-null in $\Omega_1(\MQ)$ as well.

For each $k\in[K]$, fix an injective map
\[
c_k:\{0,1,\dots,M_k-1\}\to\{0,1\}^{m_k},
\qquad
c_k(a)=(c_{k,1}(a),\ldots,c_{k,m_k}(a)).
\]
For each $i\in\calJ_1$, choose numbers $b_{i,0}\neq 0$ and $b_{i,1}\neq 0$ and define a function on $[M_{k(i)}]$ by
\[
f_i(a):=b_{i,0}+b_{i,1}\,c_{k(i),s(i)}(a),
\qquad a\in[M_{k(i)}].
\]
Since $q_{i,k(i)}=1$ for $i\in\calJ_1$, all main-effect coordinates
$\{\mathbf u(k(i),v): v\in[d_{k(i)}]\}$ are admissible.
We define the witness coefficients by
\begin{equation}\label{eq:bit_witness_coeffs}
\beta_{i,\mathbf 0}:=f_i(0),
\qquad
\beta_{i,\mathbf u(k(i),v)}:=f_i(v)-f_i(v-1)\quad (v\in[d_{k(i)}]),
\qquad
\beta_{i,\mathbf v}:=0\ \ \text{for all remaining }\mathbf v\in\calZ.
\end{equation}
Then for all $\bz\in\calZ$,
\begin{equation}\label{eq:bit_witness_eta}
\eta_i(\bz)=f_i(z_{k(i)})=b_{i,0}+b_{i,1}\,c_{k(i),s(i)}(z_{k(i)}).
\end{equation}

For each $i\in\calJ_1$, define two probability measures on $\calX_i$ by
\[
\mu_{i,0}(\cdot):=\PP_{i,\; g_i(b_{i,0},\gamma_i)}(\cdot),
\qquad
\mu_{i,1}(\cdot):=\PP_{i,\; g_i(b_{i,0}+b_{i,1},\gamma_i)}(\cdot).
\]
Since $b_{i,1}\neq 0$, Assumption~\ref{assm2}(iii) and Assumption~\ref{assm2}(ii) imply $\mu_{i,0}\neq \mu_{i,1}$.
Since $\calC_i^\mathrm{can}$ is separating, we can choose $B_i\in\calC_i^\mathrm{can}$ such that $\mu_{i,0}(B_i)\neq \mu_{i,1}(B_i)$.

For each $k\in[K]$ and each $s\in[m_k]$, define the unique index $i_{k,s}\in\calJ_1$ by
\[
i_{k,s}:=\rho_1^{-1}(k,s),
\qquad
B_{k,s}:=B_{i_{k,s}},
\]
and set
\[
x_{k,s,0}:=\mu_{i_{k,s},0}(B_{k,s}),
\qquad
x_{k,s,1}:=\mu_{i_{k,s},1}(B_{k,s}),
\qquad s\in[m_k],
\]
so that $x_{k,s,0}\neq x_{k,s,1}$.

Choose $t_\star$ large enough so that $B_i\in\bar{\mathcal D}_i^{(t_\star)}$ for all $i\in\calJ_1$.
Fix any $t\ge t_\star$.
Consider the submatrix $\MN_{1,\mathrm{sub}}^{(t)}$ of $\MN_1^{(t)}$ obtained by restricting to the $2^{\tilde d}$ rows
indexed by $\varepsilon=(\varepsilon_{k,s})_{k\in[K],\,s\in[m_k]}\in\{0,1\}^{\tilde d}$, where for each $i_{k,s}$ we select
\[
S^{(t)}_{\ell_{i_{k,s}},\,i_{k,s}}
=
\begin{cases}
B_{k,s}, & \varepsilon_{k,s}=1,\\
\calX_{i_{k,s}}, & \varepsilon_{k,s}=0,
\end{cases}
\]
and for columns we keep all $\bz\in\calZ$.
Under \eqref{eq:bit_witness_eta} and conditional independence across $j\in\calJ_1$ given $\bz$,
\begin{equation}\label{eq:MN1sub_bit_form}
\MN_{1,\mathrm{sub}}^{(t)}(\varepsilon,\bz)
=
\prod_{k=1}^K\prod_{s=1}^{m_k}
(1-\varepsilon_{k,s}+\varepsilon_{k,s}x_{k,s,\;c_{k,s}(z_k)}),
\qquad \varepsilon\in\{0,1\}^{\tilde d},\ \bz\in\calZ.
\end{equation}

For each $k\in[K]$, define the $2^{m_k}\times 2^{m_k}$ matrix $H_k^{\mathrm{full}}$ indexed by
$\varepsilon^{(k)}\in\{0,1\}^{m_k}$ and $b\in\{0,1\}^{m_k}$ via
\[
H_k^{\mathrm{full}}(\varepsilon^{(k)},b):=\prod_{s=1}^{m_k} (1-\varepsilon^{(k)}_s+\varepsilon^{(k)}_sx_{k,s,b_s}).
\]
Then
\[
H_k^{\mathrm{full}}
=
\bigotimes_{s=1}^{m_k}
\begin{pmatrix}
1 & 1\\
x_{k,s,0} & x_{k,s,1}
\end{pmatrix},
\qquad
\det(H_k^{\mathrm{full}})=\prod_{s=1}^{m_k}(x_{k,s,1}-x_{k,s,0})^{2^{m_k-1}}\neq 0,
\]
so $H_k^{\mathrm{full}}$ is invertible.
Let $H_k$ be the $2^{m_k}\times M_k$ submatrix obtained by restricting the columns of $H_k^{\mathrm{full}}$ to the subset
$\{c_k(a):a\in[M_k]\}\subseteq\{0,1\}^{m_k}$.
Since $c_k$ is injective, these $M_k$ columns are linearly independent, hence $\mathrm{rank}(H_k)=M_k$.

By \eqref{eq:MN1sub_bit_form}, the matrix $\MN_{1,\mathrm{sub}}^{(t)}$ is the Kronecker product of these blocks,
\[
\MN_{1,\mathrm{sub}}^{(t)}=\bigotimes_{k=1}^K H_k,
\]
so we have
\[
\mathrm{rank}(\MN_{1,\mathrm{sub}}^{(t)})=\prod_{k=1}^K \mathrm{rank}(H_k)=\prod_{k=1}^K M_k=r.
\]
Therefore, at the coefficient choice \eqref{eq:bit_witness_coeffs}, the full matrix $\MN_1^{(t)}$ has full column rank $r$
for every $t\ge t_\star$.

The similar real-analytic argument applied to block $\calJ_2$ yields a Lebesgue-null set on which $\mathrm{rk}_k(\MN_2^{(t)})=r$ for all $t\ge t_\star$.

\smallskip
\emph{(ii) Eventual $\mathrm{rk}_k(\MN_3^{(t)})\ge 2$.}

Fix $\bz\neq\bz'$ and choose $k$ with $z_k\neq z_k'$.
Let $j_k\in\calJ_3$ be as in Condition (i) in Theorem~\ref{thm:poly_generic_ident_delta}, so that $q_{j_k,k}=1$.
Set $u_\star:=\min\{z_k,z_k'\}+1\in[M_k-1]$ so that $\calI_{k,u_\star}(\bz)\neq \calI_{k,u_\star}(\bz')$.
Then $\beta_{j_k,\mathbf u(k,u_\star)}$ is a free coordinate in the chart for item $j_k$ (subject only to $\beta_{j_k,\mathbf u(k,u_\star)}\neq 0$),
so $\eta_{j_k}(\bz)=\eta_{j_k}(\bz')$ defines a proper affine hyperplane in that coefficient space.
Hence, outside the union of these hyperplanes over all pairs $\bz\neq\bz'$, we have $\eta_{j_k}(\bz)\neq \eta_{j_k}(\bz')$ for every $\bz\neq\bz'$. By Assumption~\ref{assm2}, this implies the corresponding conditional laws differ, hence for all large enough $t$ the columns of $\MN_3^{(t)}$ are pairwise distinct, and together with the all-$\calX$ row this yields $\mathrm{rk}_k(\MN_3^{(t)})\ge 2$.

\smallskip
\emph{(iii) Kruskal uniqueness for all large $t$.}
Combining the three blocks yields \eqref{eq:poly_kruskal_target_paramfree} for all sufficiently large $t$ outside a Lebesgue-null set, hence uniqueness up to a common column permutation.
\end{proof}

Fix $\Theta$ outside a Lebesgue-null set to be specified and take $t$ large enough so that Lemma~\ref{lem:poly_kruskal_gid_paramfree} applies.
Let $\MN_a^{\prime(t)}$ denote the analogous factor matrices constructed from $\Theta'$ using the same discretizations and the same item blocks $\calJ_1,\calJ_2,\calJ_3$.
If
\[
\mathbf P_0^{(t)}
=
\big[\,\MN_1^{(t)}\mathrm{Diag}(\p),\,\MN_2^{(t)},\,\MN_3^{(t)}\,\big]
=
\big[\,\MN_1^{\prime(t)}\mathrm{Diag}(\p'),\,\MN_2^{\prime(t)},\,\MN_3^{\prime(t)}\,\big],
\]
then there exists a permutation $\mathfrak S^{(t)}\in  S_r$ such that
\begin{equation}\label{eq:poly_align_MN23_paramfree}
\MN_a^{(t)}(\cdot,\bz)=\MN_a^{\prime(t)}(\cdot,\mathfrak S^{(t)}(\bz))
\ \ (a=2,3),
\qquad
\big(\MN_1^{(t)}\mathrm{Diag}(\p)\big)(\cdot,\bz)=\big(\MN_1^{\prime(t)}\mathrm{Diag}(\p')\big)(\cdot,\mathfrak S^{(t)}(\bz)).
\end{equation}
Let $\xi_{1,\mathrm{all}}$ be the row index of $\MN_1^{(t)}$ selecting $\calX_j$ for every $j\in\calJ_1$.
Then $\MN_1^{(t)}(\xi_{1,\mathrm{all}},\bz)=1$ and $\MN_1^{\prime(t)}(\xi_{1,\mathrm{all}},\bz)=1$ for all $\bz$, hence evaluating \eqref{eq:poly_align_MN23_paramfree} at $\xi_{1,\mathrm{all}}$ yields
\begin{equation}\label{eq:poly_pi_align_paramfree}
p_{\bz}=p'_{\mathfrak S^{(t)}(\bz)}\qquad\forall \bz\in\calZ.
\end{equation}
Dividing the last identity in \eqref{eq:poly_align_MN23_paramfree} columnwise by $p_{\bz}$ yields
\begin{equation}\label{eq:poly_align_MN1_paramfree}
\MN_1^{(t)}(\cdot,\bz)=\MN_1^{\prime(t)}(\cdot,\mathfrak S^{(t)}(\bz))
\qquad\forall \bz\in\calZ.
\end{equation}
In particular, for every $j\in[J]$, every $S\in\bar{\mathcal D}_j^{(t)}$, and every $\bz\in\calZ$,
\begin{equation}\label{eq:poly_prob_compare_paramfree}
\PP_{j,g_j(\eta_j(\bz),\gamma_j))}(S)=\PP_{j,g_j(\eta'_j(\mathfrak S^{(t)}(\bz)),\gamma'_j))}(S).
\end{equation}

\begin{lemma}[Permutation stability under refinement]\label{lem:poly_perm_stability_paramfree}
Let $\bar{\mathcal D}^{(t)}\subseteq \bar{\mathcal D}^{(t+1)}$ be the nested discretizations above, and let
$\mathfrak S^{(t)},\mathfrak S^{(t+1)}$ be the aligning permutations obtained from Lemma~\ref{lem:poly_kruskal_gid_paramfree} at levels $t$ and $t+1$.
If $\MN_1^{(t)}$ has full column rank $r$, then $\mathfrak S^{(t+1)}=\mathfrak S^{(t)}$.
\end{lemma}

\begin{proof}
Because $\bar{\mathcal D}^{(t)}\subseteq\bar{\mathcal D}^{(t+1)}$, every row event defining $\MN_1^{(t)}$ also appears among the rows defining $\MN_1^{(t+1)}$.
Thus the column $\MN_1^{(t)}(\cdot,\bz)$ is obtained from $\MN_1^{(t+1)}(\cdot,\bz)$ by restricting to a subset of rows.
Using \eqref{eq:poly_align_MN1_paramfree} at levels $t$ and $t+1$ and restricting the level $t+1$ identity to the level $t$ rows yields
\[
\MN_1^{\prime(t)}(\cdot,\mathfrak S^{(t)}(\bz))=\MN_1^{\prime(t)}(\cdot,\mathfrak S^{(t+1)}(\bz)).
\]
Since $\MN_1^{(t)}$ has full column rank, its columns are pairwise distinct, and the same holds for $\MN_1^{\prime(t)}$ because it is a column permutation of $\MN_1^{(t)}$.
Therefore $\mathfrak S^{(t)}(\bz)=\mathfrak S^{(t+1)}(\bz)$ for all $\bz$.
\end{proof}

On a generic set where $\MN_1^{(t)}$ has full column rank for all sufficiently large $t$, Lemma~\ref{lem:poly_perm_stability_paramfree} implies that $\mathfrak S^{(t)}$ is constant for all large $t$.
Denote the common permutation by $\mathfrak S\in S_r$.

Now fix any $j\in[J]$ and any set $S\in\calC_j^\mathrm{can}$.
Since $\bigcup_{t\ge 1}\bar{\mathcal D}_j^{(t)}=\calC_j^\mathrm{can}$, there exists $t$ large enough with $S\in\bar{\mathcal D}_j^{(t)}$.
Then \eqref{eq:poly_prob_compare_paramfree} implies
\[
\PP_{j,g_j(\eta_j(\bz),\gamma_j))}(S)=\PP_{j,g_j(\eta_j'(\mathfrak S(\bz)),\gamma_j'))}(S)
\qquad\text{for all }\bz\in\calZ.
\]
Because $\calC_j^\mathrm{can}$ is separating, we conclude
\[
\PP_{j,g_j(\eta_j(\bz),\gamma_j))}=\PP_{j,g_j(\eta'_j(\mathfrak S(\bz)),\gamma'_j))}
\qquad\text{as probability measures on }\calX_j,
\qquad\forall j\in[J],\ \bz\in\calZ.
\]
Finally, by Assumption~\ref{assm2}(ii) and injectivity of $g_j$ in Assumption~\ref{assm2}(iii), we obtain
\begin{equation}\label{eq:step1_tables_identified}
\eta_j(\bz)=\eta'_j(\mathfrak S(\bz)),
\qquad
\gamma_j=\gamma'_j,
\qquad \forall j\in[J],\ \bz\in\calZ,
\end{equation}
and combining with \eqref{eq:poly_pi_align_paramfree} (stabilized to $\mathfrak S$) yields
\begin{equation}\label{eq:step1_pi_identified}
p_{\bz}=p'_{\mathfrak S(\bz)}
\qquad \forall \bz\in\calZ.
\end{equation}

\subsubsection{Recovering $\MQ$ and restricting admissible relabelings (yielding Corollary~\ref{cor:poly_moran_indeterminacy})}

\begin{lemma}\label{lem:poly_disentangle_core}
Fix integers $M_k\ge 2$ and $\calZ=\prod_{k=1}^K[M_k]$.
Let $\MQ=(q_{j,k})\in\{0,1\}^{J\times K}$ and define $K_j=\{k:q_{j,k}=1\}$.
For each item $j$, let $\eta_j:\calZ\to\RR$ satisfy the structural restriction
\[
\eta_j(\MZ)=\eta_j(\MZ')\qquad\text{whenever }\MZ_{K_j}=\MZ'_{K_j}.
\]
Assume the generic injectivity condition
\begin{equation}\label{eq:poly_generic_inj}
\eta_j(\MZ)\neq \eta_j(\MZ')\qquad\text{whenever }\MZ_{K_j}\neq \MZ'_{K_j}.
\end{equation}
Let $(\MQ',\{\eta'_j\}_{j=1}^J)$ be another design/predictor pair with the analogous property.
Suppose there exists a bijection $\sigma:\calZ\to\calZ$ such that
\begin{equation}\label{eq:poly_eta_rowperm}
\eta_j(\MZ)=\eta'_j\big(\sigma(\MZ)\big)\qquad\text{for all }j\in[J],\ \MZ\in\calZ.
\end{equation}
Assume moreover the no-containment condition for $\MQ$
\begin{equation}\label{eq:no_contain}
\text{for any }p\neq q,\ \text{neither }\MQ_{:,p}\succeq \MQ_{:,q}\ \text{nor }\MQ_{:,q}\succeq \MQ_{:,p}.
\end{equation}

Then the following conclusions hold.

\begin{enumerate}
\item[(a)] Let $C_k:=\{j\in[J]:q_{j,k}=1\}$ denote the support set of column $k$ of $\MQ$.
For any $\MZ,\MZ'\in\calZ$, define the coordinate-difference set
\[
S(\MZ,\MZ'):=\{k\in[K]:Z_k\neq Z_k'\}
\]
and the item-difference pattern
\[
D(\MZ,\MZ'):=\{j\in[J]:\eta_j(\MZ)\neq \eta_j(\MZ')\}.
\]
Then
\begin{equation}\label{eq:D_union}
D(\MZ,\MZ')=\bigcup_{k\in S(\MZ,\MZ')} C_k .
\end{equation}

\item[(b)] The collection of inclusion-minimal nonempty sets among
$\{D(\MZ,\MZ'):\MZ\neq\MZ'\}$ equals $\{C_k:k\in[K]\}$.
Consequently, the multiset of columns of $\MQ$ is identified from $\{\eta_j(\MZ)\}$,
hence $\MQ$ is identified up to a column permutation.

\item[(c)] There exists a permutation $\varpi\in S_K$ such that
\[
C_k = C'_{\varpi(k)}\qquad\text{for all }k\in[K],
\]
equivalently $\MQ' = \MQ\Pi$ for the permutation matrix $\Pi$ of $\varpi^{-1}$.

\item[(d)] The relabeling $\sigma$ must lie in 
\[
\Bigl(\prod_{k=1}^K  S_{M_k}\Bigr)\rtimes  S_{K,\mathbf M},
\qquad
 S_{K,\mathbf M}:=\{\varpi\in S_K:\ M_{\varpi(k)}=M_k\ \forall k\}.
\]
More explicitly, with $\varpi$ from {\rm (c)}, there exist permutations $\tau_k\in S_{M_k}$ such that
for all $\MZ=(Z_1,\ldots,Z_K)\in\calZ$ and all $k\in[K]$,
\[
\bigl(\sigma(\MZ)\bigr)_{\varpi(k)}=\tau_k(Z_k).
\]
Equivalently, for all $(Z_1,\ldots,Z_K)\in\calZ$,
\[
\sigma(Z_1,\ldots,Z_K)
=
\bigl(\tau_{\varpi^{-1}(1)}(Z_{\varpi^{-1}(1)}),\ldots,\tau_{\varpi^{-1}(K)}(Z_{\varpi^{-1}(K)})\bigr).
\]

\end{enumerate}
\end{lemma}

\begin{proof}
We prove (a)--(d) in order.

\smallskip
\emph{Proof of (a).}
Fix $\MZ,\MZ'\in\calZ$.
For any item $j$, by the defining restriction of $K_j$ we have
$\eta_j(\MZ)=\eta_j(\MZ')$ whenever $\MZ_{K_j}=\MZ'_{K_j}$.
Conversely, by \eqref{eq:poly_generic_inj},
$\eta_j(\MZ)\neq \eta_j(\MZ')$ whenever $\MZ_{K_j}\neq \MZ'_{K_j}$.
Hence
\[
\eta_j(\MZ)\neq \eta_j(\MZ')
\Longleftrightarrow
\MZ_{K_j}\neq \MZ'_{K_j}
\Longleftrightarrow
K_j\cap S(\MZ,\MZ')\neq\varnothing.
\]
Therefore
\[
D(\MZ,\MZ')
=\{j:K_j\cap S(\MZ,\MZ')\neq\varnothing\}
=\bigcup_{k\in S(\MZ,\MZ')} \{j:k\in K_j\}
=\bigcup_{k\in S(\MZ,\MZ')} C_k,
\]
which is \eqref{eq:D_union}.

\smallskip
\emph{Proof of (b).}
Let $\mathcal D:=\{D(\MZ,\MZ'):\MZ\neq\MZ'\}$.
By (a), every element of $\mathcal D$ is a union of some subcollection of $\{C_k\}_{k=1}^K$.
Fix $k\in[K]$ and choose $\MZ,\MZ'$ that differ only at coordinate $k$.
Then $S(\MZ,\MZ')=\{k\}$ and (a) gives $D(\MZ,\MZ')=C_k\in\mathcal D$,
so every $C_k$ appears in $\mathcal D$.

Now take any nonempty $D\in\mathcal D$.
Then $D=\bigcup_{k\in S}C_k$ for some nonempty $S\subseteq[K]$.
If $|S|\ge 2$ then for any $k_0\in S$ we have $C_{k_0}\subseteq D$.
By \eqref{eq:no_contain}, the inclusion is strict, hence $D$ is not inclusion-minimal.
Thus the inclusion-minimal nonempty elements of $\mathcal D$ are exactly $\{C_k:k\in[K]\}$.
Knowing all $C_k$ recovers $\MQ$ up to a permutation of columns.

\smallskip
\emph{Proof of (c).}
From \eqref{eq:poly_eta_rowperm}, for any $\MZ,\MZ'\in\calZ$ and any $j\in[J]$,
\[
\eta_j(\MZ)\neq \eta_j(\MZ')
\Longleftrightarrow
\eta'_j\big(\sigma(\MZ)\big)\neq \eta'_j\big(\sigma(\MZ')\big),
\]
hence $D(\MZ,\MZ')=D'(\sigma(\MZ),\sigma(\MZ'))$.
Because $\sigma$ is a bijection on $\calZ$, the map $(\MZ,\MZ')\mapsto(\sigma(\MZ),\sigma(\MZ'))$
is a bijection on $\{(\MZ,\MZ'):\MZ\neq \MZ'\}$, and therefore
\[
\mathcal D=\{D(\MZ,\MZ'):\MZ\neq\MZ'\}
=
\{D'(\tilde\MZ,\tilde\MZ'):\tilde\MZ\neq\tilde\MZ'\}
=\mathcal D'
\]
as \emph{sets of subsets} of $[J]$.

By part (b) (which uses \eqref{eq:no_contain} for $\MQ$), the inclusion-minimal nonempty elements of $\mathcal D$
are exactly $\{C_k:k\in[K]\}$. Since $\mathcal D=\mathcal D'$, the inclusion-minimal nonempty elements of $\mathcal D'$
are also exactly $\{C_k:k\in[K]\}$.

Fix $k\in[K]$ and choose $\MZ,\MZ'\in\calZ$ that differ only at coordinate $k$.
Then $D(\MZ,\MZ')=C_k$, so $D'(\sigma(\MZ),\sigma(\MZ'))=C_k$.
Applying part (a) to $(\MQ',\eta')$ gives
\[
D'\big(\sigma(\MZ),\sigma(\MZ')\big)
=\bigcup_{h\in S'\big(\sigma(\MZ),\sigma(\MZ')\big)} C'_h.
\]
Every $C'_h$ belongs to $\mathcal D'$ (take two states that differ only at coordinate $h$),
so each $C'_h$ appearing in the above union is a nonempty element of $\mathcal D'$ and satisfies
$C'_h\subseteq \bigcup_{h\in S'} C'_h = C_k$.
But $C_k$ is inclusion-minimal among the nonempty elements of $\mathcal D'$,
hence no nonempty element of $\mathcal D'$ can be a \emph{proper} subset of $C_k$.
Therefore every $C'_h$ appearing in the union must equal $C_k$.
In particular, there exists at least one index $\varpi(k)\in[K]$ such that $C'_{\varpi(k)}=C_k$.

If $\varpi(k_1)=\varpi(k_2)$ then $C_{k_1}=C'_{\varpi(k_1)}=C'_{\varpi(k_2)}=C_{k_2}$.
Under \eqref{eq:no_contain}, the sets $\{C_k\}_{k=1}^K$ are pairwise distinct, so $k_1=k_2$.
Thus $\varpi$ is injective and hence a permutation of $[K]$.
Consequently $C_k=C'_{\varpi(k)}$ for all $k\in[K]$, equivalently $\MQ'=\MQ\Pi$ for the permutation matrix $\Pi$.

\smallskip
\emph{Proof of (d).}

Define the Hamming graph on $\calZ$ whose vertex set is $\calZ$ and where two vertices are adjacent
if they differ in exactly one coordinate.
For each $k\in[K]$, call an edge $\{\MZ,\MZ'\}$ a $k$-edge if $\MZ,\MZ'$ differ only in coordinate $k$.

By (c), $\MQ'$ is a column permutation of $\MQ$, hence \eqref{eq:no_contain} also holds for $\MQ'$.
Now fix a $k$-edge $\{\MZ,\MZ'\}$, i.e. $\MZ,\MZ'$ differ only at coordinate $k$.
Then $D(\MZ,\MZ')=C_k$, so $D'(\sigma(\MZ),\sigma(\MZ'))=C_k=C'_{\varpi(k)}$.
If $\sigma(\MZ)$ and $\sigma(\MZ')$ differed in at least two coordinates, then by part (a) for $(\MQ',\eta')$
the set $D'\big(\sigma(\MZ),\sigma(\MZ')\big)$ would be a union of at least two distinct sets among $\{C'_h\}$.
Under \eqref{eq:no_contain} for $\MQ'$, such a union strictly contains each constituent, hence cannot equal $C'_{\varpi(k)}$.
Therefore $\sigma(\MZ)$ and $\sigma(\MZ')$ differ in exactly one coordinate.
Let that coordinate be $h$.
Then by part (a) for $(\MQ',\eta')$ applied to the pair $\sigma(\MZ),\sigma(\MZ')$ we have $D'\big(\sigma(\MZ),\sigma(\MZ')\big)=C'_h.$ Comparing with $D'(\sigma(\MZ),\sigma(\MZ'))=C'_{\varpi(k)}$ yields $C'_h=C'_{\varpi(k)}$,
and since \eqref{eq:no_contain} implies the supports $\{C'_1,\ldots,C'_K\}$ are pairwise distinct, we conclude $h=\varpi(k)$.
Consequently, $\sigma$ maps every $k$-edge to a $\varpi(k)$-edge.

Fix $k\in[K]$ and fix a context $\mathbf z_{-k}\in\prod_{h\neq k}[M_h]$.
For $t\in[M_k]$, write $\bz(\mathbf z_{-k},t)$ for the latent state whose $-k$ coordinates equal $\mathbf z_{-k}$
and whose $k$th coordinate equals $t$.
Then any two distinct vertices in the fiber
\[
\calF_{k}(\mathbf z_{-k}) := \{\bz(\mathbf z_{-k},t): t\in[M_k]\}
\]
differ in exactly one coordinate (namely $k$), hence are joined by a $k$-edge.
Since $\sigma$ maps $k$-edges to $\varpi(k)$-edges, it follows that for any $t\neq t'$,
the images $\sigma(\bz(\mathbf z_{-k},t))$ and $\sigma(\bz(\mathbf z_{-k},t'))$
differ in exactly one coordinate (namely $\varpi(k)$).
In particular, all vectors $\{\sigma(\bz(\mathbf z_{-k},t)):t\in[M_k]\}$ agree on coordinates outside $\varpi(k)$.

Therefore there exists a map $\tau_{k,\mathbf z_{-k}}:[M_k]\to[M_{\varpi(k)}]$ such that
\begin{equation}\label{eq:tau_context_def}
\bigl(\sigma(\bz(\mathbf z_{-k},t))\bigr)_{\varpi(k)}=\tau_{k,\mathbf z_{-k}}(t)
\qquad\text{for all }t\in[M_k].
\end{equation}
Because $\sigma$ is injective, the points $\sigma(\bz(\mathbf z_{-k},t))$ are all distinct as $t$ varies,
hence $\tau_{k,\mathbf z_{-k}}$ is injective. Thus $M_{\varpi(k)}\ge M_k$.
Applying the same argument to $\sigma^{-1}$ (which maps $\varpi(k)$-edges back to $k$-edges) yields $M_k\ge M_{\varpi(k)}$.
Consequently $M_{\varpi(k)}=M_k$ and $\tau_{k,\mathbf z_{-k}}\in   S_{M_k}$ is a permutation.

Next we show that $\tau_{k,\mathbf z_{-k}}$ does not depend on the context $\mathbf z_{-k}$.
Fix $t\in[M_k]$ and take two contexts $\mathbf z_{-k},\mathbf b_{-k}$.
In the Hamming graph on $\calZ$, there is a path from $\bz(\mathbf z_{-k},t)$ to $\bz(\mathbf b_{-k},t)$
that changes only coordinates in $[K]\setminus\{k\}$.
Along this path, each step is an $h$-edge for some $h\neq k$, hence its image under $\sigma$ is a $\varpi(h)$-edge.
Since $\varpi$ is a permutation, $\varpi(h)\neq\varpi(k)$ whenever $h\neq k$,
so the $\varpi(k)$-coordinate remains constant along the image path.
Therefore $(\sigma(\bz(\mathbf z_{-k},t)))_{\varpi(k)}=(\sigma(\bz(\mathbf b_{-k},t)))_{\varpi(k)}$.
By \eqref{eq:tau_context_def}, this implies $\tau_{k,\mathbf z_{-k}}(t)=\tau_{k,\mathbf b_{-k}}(t)$ for every $t$,
hence $\tau_{k,\mathbf z_{-k}}$ is the same permutation for all contexts.
Denote this common permutation by $\tau_k\in   S_{M_k}$.

We have shown that for every $\MZ=(Z_1,\ldots,Z_K)\in\calZ$ and every $k\in[K]$, $(\sigma(\MZ))_{\varpi(k)}=\tau_k(Z_k)$. Equivalently, writing $m=\varpi(k)$ and $k=\varpi^{-1}(m)$,
$(\sigma(\MZ))_{m}=\tau_{\varpi^{-1}(m)}(Z_{\varpi^{-1}(m)})$. Therefore, for all $(Z_1,\ldots,Z_K)\in\calZ$,
\[
\sigma(Z_1,\ldots,Z_K)
=
\bigl(\tau_{\varpi^{-1}(1)}(Z_{\varpi^{-1}(1)}),\ldots,\tau_{\varpi^{-1}(K)}(Z_{\varpi^{-1}(K)})\bigr).
\]
Finally, since $M_{\varpi(k)}=M_k$ for all $k$, we have $\varpi\in   S_{K,\mathbf M}$.
This proves that $\sigma$ lies in $\Bigl(\prod_{k=1}^K    S_{M_k}\Bigr)\rtimes    S_{K,\mathbf M}$.
\end{proof}
We now explain explicitly why Lemma~\ref{lem:poly_disentangle_core} implies Corollary~\ref{cor:poly_moran_indeterminacy}. Fix any $\xi\in A(\calF,\calP)$ and any $f\in\calF$. By definition of the indeterminacy set, the transformed pair $(f\circ\xi^{-1},\,\xi_{\#}p)$ also belongs to $(\calF,\calP)$. Write
\[
f' := f\circ\xi^{-1}.
\]
Since $f\in\calF$, there exists a binary matrix $\MQ$ such that $(\MQ,\{f_j\}_{j=1}^J)$ satisfies the structural restriction
\[
f_j(\bz)=f_j(\bz')\qquad\text{whenever }\bz_{K_j}=\bz'_{K_j},
\]
the generic injectivity condition
\[
f_j(\bz)\neq f_j(\bz')\qquad\text{whenever }\bz_{K_j}\neq \bz'_{K_j},
\]
and the subset condition on the columns of $\MQ$. Likewise, since $f'\in\calF$, there exists another binary matrix $\MQ'$ such that $(\MQ',\{f'_j\}_{j=1}^J)$ satisfies the analogous properties. Now set
\[
\eta_j:=f_j,\qquad \eta'_j:=f'_j,\qquad \sigma:=\xi.
\]
Then for every $j\in[J]$ and every $\bz\in\calZ$,
\[
\eta_j(\bz)=f_j(\bz)=f'_j(\xi(\bz))=\eta'_j(\sigma(\bz)),
\]
so all assumptions of Lemma~\ref{lem:poly_disentangle_core} are satisfied. Therefore Lemma~\ref{lem:poly_disentangle_core}(d) yields
\[
\xi\in \Bigl(\prod_{k=1}^K S_{M_k}\Bigr)\rtimes S_{K,\mathbf M},
\]
which is exactly the conclusion of Corollary~\ref{cor:poly_moran_indeterminacy}.

Note that the corollary is, in one respect, more restrictive in its setup than Lemma~\ref{lem:poly_disentangle_core}. The lemma assumes the subset condition only for the true matrix $\MQ$, whereas Corollary~\ref{cor:poly_moran_indeterminacy} is stated in terms of $\xi\in A(\calF,\calP)$, and hence requires both $f$ and $f\circ\xi^{-1}$ to belong to $\calF$. Consequently, both associated design matrices $\MQ$ and $\MQ'$ must satisfy the defining constraints of $\calF$, including the subset condition. 
\begin{remark}
The argument in Lemma~\ref{lem:poly_disentangle_core} is essentially combinatorial and does not rely on the finiteness of $\calZ$.
In particular, the same proof applies if one replaces $\calZ=\prod_{k=1}^K[M_k]$ by a product $\calZ=\prod_{k=1}^K\calZ_k$ of arbitrary coordinate sets, in which case the symmetry group $\bigl(\prod_{k=1}^K S_{M_k}\bigr)\rtimes S_{K,\mathbf M}$ is replaced by the group of coordinatewise bijections
\[
\Bigl(\prod_{k=1}^K \mathrm{Bij}(\calZ_k)\Bigr)\rtimes S_K,
\]
so that $\sigma$ also decompose into a coordinate permutation composed with per--coordinate relabelings.
However, when $\calZ_k$ are continuous domains (e.g.\ $\calZ_k=\RR$), the injectivity condition \eqref{eq:poly_generic_inj} is typically incompatible with mild regularity of $\eta_j$ as soon as $|K_j|>1$:
indeed, there is no continuous injective map from $\RR^{|K_j|}$ into $\RR$ for $|K_j|>1$.
For this reason we state the lemma in the discrete setting, where \eqref{eq:poly_generic_inj} is a natural condition.
\end{remark}

Fix an item $j\in[J]$ and recall $K_j=\{k\in[K]: q_{j,k}=1\}$ and
$\mathrm{supp}(\bu):=\{k\in[K]: u_k\ge 1\}$.
Define the admissible index set
\begin{equation}\label{eq:ZjQ_index_set}
\calZ_j^{\MQ}
:=
\{\bu\in\calZ:\ \mathrm{supp}(\bu)\subseteq K_j\},
\end{equation}
the corresponding feature vector
\begin{equation}\label{eq:Phi_jQ_def}
\Phi_j^{\MQ}(\MZ)
:=
\Bigl(\ \prod_{k=1}^K \calI_{k,u_k}(\MZ)\ \Bigr)_{\bu\in\calZ_j^{\MQ}}
\in\{0,1\}^{|\calZ_j^{\MQ}|},
\end{equation}
and the free coefficient subvector
\begin{equation}\label{eq:beta_jQ_def}
\bbeta_j^{\MQ}
:=
(\beta_{j,\bu})_{\bu\in\calZ_j^{\MQ}}\in\RR^{|\calZ_j^{\MQ}|}.
\end{equation}
Then we have the reduced representation
\begin{equation}\label{eq:eta_reduced_Q}
\eta_j(\MZ)=\langle \bbeta_j^{\MQ},\ \Phi_j^{\MQ}(\MZ)\rangle.
\end{equation}

\begin{lemma}\label{lem:Q_structural_restriction_for_eta}
Fix item $j$.
If $\MZ_{K_j}=\MZ'_{K_j}$ then $\eta_j(\MZ)=\eta_j(\MZ')$.
\end{lemma}

\begin{proof}
Take any $\bu\in\calZ_j^{\MQ}$.
If $k\notin K_j$, then $\mathrm{supp}(\bu)\subseteq K_j$ forces $u_k=0$, hence $\calI_{k,u_k}(\cdot)=\calI_{k,0}(\cdot)\equiv 1$.
Therefore the basis product $\prod_{k=1}^K \calI_{k,u_k}(\cdot)$ depends only on $\MZ_{K_j}$.
If $\MZ_{K_j}=\MZ'_{K_j}$ then every coordinate of $\Phi_j^{\MQ}(\MZ)$ equals the corresponding coordinate of $\Phi_j^{\MQ}(\MZ')$,
and \eqref{eq:eta_reduced_Q} yields $\eta_j(\MZ)=\eta_j(\MZ')$.
\end{proof}

\begin{lemma}\label{lem:full_ladder_implies_feature_separation_Q}
For any $\MZ,\MZ'\in\calZ$,
\[
\MZ_{K_j}\neq \MZ'_{K_j}\ \Longrightarrow\ \Phi_j^{\MQ}(\MZ)\neq \Phi_j^{\MQ}(\MZ').
\]
Consequently, outside a Lebesgue-null set in the free coordinates $\bbeta_j^{\MQ}$ (equivalently, in the non-structural-zero coordinates of $\bbeta_j$),
\[
\MZ_{K_j}\neq \MZ'_{K_j}\ \Longrightarrow\ \eta_j(\MZ)\neq \eta_j(\MZ').
\]
\end{lemma}

\begin{proof}
Assume $\MZ_{K_j}\neq \MZ'_{K_j}$ and pick $k\in K_j$ such that $Z_k\neq Z_k'$.
Without loss of generality $Z_k<Z_k'$.
Set $v:=Z_k'\in[M_k-1]$.
Then $\calI_{k,v}(\MZ)=0$ and $\calI_{k,v}(\MZ')=1$.
Since $k\in K_j$, the coordinate $\bu(k,v)$ belongs to $\calZ_j^{\MQ}$, and its corresponding feature in \eqref{eq:Phi_jQ_def} is exactly
$\prod_{h=1}^K \calI_{h,\bu(k,v)_h}(\cdot)=\calI_{k,v}(\cdot)$, because $\bu(k,v)_h=0$ for all $h\neq k$ and $\calI_{h,0}\equiv 1$.
Hence the $\bu(k,v)$-coordinate of $\Phi_j^{\MQ}(\MZ)$ differs from that of $\Phi_j^{\MQ}(\MZ')$, proving
$\Phi_j^{\MQ}(\MZ)\neq \Phi_j^{\MQ}(\MZ')$.

For the generic injectivity, fix any distinct pair $(\MZ,\MZ')$ with $\MZ_{K_j}\neq \MZ'_{K_j}$.
By the first part, $\Phi_j^{\MQ}(\MZ)-\Phi_j^{\MQ}(\MZ')\neq 0$, so the equality
$\eta_j(\MZ)=\eta_j(\MZ')$ is the proper affine hyperplane
\[
\bigl\{\bbeta_j^{\MQ}:\ \langle \bbeta_j^{\MQ},\ \Phi_j^{\MQ}(\MZ)-\Phi_j^{\MQ}(\MZ')\rangle=0\bigr\}
\]
in $\RR^{|\calZ_j^{\MQ}|}$ (using \eqref{eq:eta_reduced_Q}).
Since $|\calZ|<\infty$, the union of these hyperplanes over all such pairs is Lebesgue-null.
Intersecting with the constraint set (which only removes finitely many coordinate hyperplanes
$\{\beta_{j,\bu(k,v)}=0\}$ and thus does not change nullness) yields the claim.
\end{proof}

From Step~1 there exists a stabilized bijection $\sigma:\calZ\to\calZ$ such that
\begin{equation}\label{eq:step2_global_eta_align_again}
\eta_j(\bz)=\eta'_j(\sigma(\bz))
\qquad\forall j\in[J],\ \forall \bz\in\calZ.
\end{equation}

Apply Lemma~\ref{lem:Q_structural_restriction_for_eta} and Lemma~\ref{lem:full_ladder_implies_feature_separation_Q}
to every item $j$ and intersect the corresponding generic sets over $j\in[J]$.
On this intersection, the pair $(\MQ,\{\eta_j\}_{j=1}^J)$ satisfies the hypotheses of Lemma~\ref{lem:poly_disentangle_core}
with $K_j=\{k:\ q_{j,k}=1\}$.
The same holds for $(\MQ',\{\eta'_j\}_{j=1}^J)$.

Applying Lemma~\ref{lem:poly_disentangle_core} to $(\MQ,\{\eta_j\})$ and $(\MQ',\{\eta'_j\})$ yields an permutation $\varpi\in   S_{K,\mathbf M}$ such that $\MQ'=\MQ\Pi$, where $\Pi$ is the permutation matrix of $\varpi^{-1}$, and the stabilized relabeling $\sigma$ must lie in $(\prod_{k=1}^K    S_{M_k})\rtimes    S_{K,\mathbf M}$.
Equivalently, there exist permutations $\tau_k\in\mathfrak S_{M_k}$ such that
\begin{equation}\label{eq:step2_sigma_wreath_form}
\sigma(z_1,\ldots,z_K)
=
\bigl(\tau_1(z_{\varpi^{-1}(1)}),\ldots,\tau_K(z_{\varpi^{-1}(K)})\bigr).
\end{equation}

At this stage we have identified the coordinate permutation $\varpi$. In the next step we remove the within-coordinate relabelings $\{\tau_k\}$, thereby concluding $\tau_k=\mathrm{Id}$ and hence
$\sigma=\sigma_{\varpi}$.

\subsubsection{Eliminate within-coordinate relabelings}
Recall that $\sigma_{\varpi}$ for the coordinate permutation on $\calZ$ is defined by $\sigma_{\varpi}(z_1,\ldots,z_K)=(z_{\varpi^{-1}(1)},\ldots,\\z_{\varpi^{-1}(K)})$, and write $\tau$ for the within-coordinate map $\tau(z_1,\ldots,z_K):=(\tau_1(z_1),\ldots,\tau_K(z_K))$,
we have the factorization $\sigma=\tau\circ\sigma_{\varpi}$.

Define the relabeled alternative predictors
\[
\widetilde\eta'_j(\bz):=\eta'_j\big(\sigma_{\varpi}(\bz)\big),\qquad (j\in[J],\ \bz\in\calZ),
\]and define the relabeled stabilized permutation $\widetilde\sigma:=\sigma_{\varpi}^{-1}\circ\sigma$.
Then $\sigma=\tau\circ\sigma_{\varpi}$ implies $\widetilde\sigma$ has within-coordinate form
\[
\widetilde\sigma(z_1,\ldots,z_K)=(\widetilde\tau_1(z_1),\ldots,\widetilde\tau_K(z_K))
\qquad \text{for some }\widetilde\tau_k\in   S_{M_k},
\]
and the Step~2 alignment $\eta_j(\bz)=\eta'_j(\sigma(\bz))$ becomes
\begin{equation}\label{eq:eta_align_after_absorb_varpi}
\eta_j(\bz)=\widetilde\eta'_j\big(\widetilde\sigma(\bz)\big)\qquad \forall j\in[J],\ \forall \bz\in\calZ.
\end{equation}

Similarly, relabel the alternative measurement matrix by permuting coordinate indices according to $\varpi$.
Let $\Pi$ denote the permutation matrix of $\varpi^{-1}$ so that Step~2 yields $\MQ'=\MQ\,\Pi$.
Define $\widetilde\MQ':=\MQ'\Pi^\top$.
Then $\widetilde\MQ'=\MQ$ and hence $K_j^{\widetilde\MQ'}=K_j$ for all $j$.
Moreover, Assumption~\ref{assm:poly_monotone_binary_style} is invariant under this coordinate relabeling.

Therefore, to prove that within-coordinate relabelings are trivial, it suffices to work with the relabeled alternative model
$(\widetilde\MQ',\{\widetilde\eta'_j\})$ and the relabeled stabilized permutation $\widetilde\sigma$.
For notational simplicity, we may assume
\[
\varpi=\mathrm{Id},
\qquad
\MQ'=\MQ,
\qquad
\sigma(z_1,\ldots,z_K)=(\tau_1(z_1),\ldots,\tau_K(z_K)),
\]
with $\tau_k\in   S_{M_k}$.

Thus we are reduced to the within-coordinate form
\begin{equation}\label{eq:sigma_within_only_fixed}
\sigma(z_1,\ldots,z_K)=(\tau_1(z_1),\ldots,\tau_K(z_K)),
\qquad \tau_k\in   S_{M_k},
\end{equation}
and we show $\tau_k=\mathrm{Id}$ for every $k$.

\begin{lemma}\label{lem:kill_within_coord_perm_binary_style}
Assume Assumption~\ref{assm:poly_monotone_binary_style} holds for both
$(\MQ,\{\eta_j\})$ and $(\MQ',\{\eta'_j\})$.
Assume also that $\MQ'=\MQ$ and that $\sigma$ has the within-coordinate form \eqref{eq:sigma_within_only_fixed}.

Fix $k\in[K]$ and let $j_k\in\calJ_3$ be an anchor item in the true design (Condition(i) in Theorem~\ref{thm:poly_generic_ident_delta})
satisfying $q_{j_k,k}=1$.
Suppose the alignment for this item holds
\begin{equation}\label{eq:eta_align_item_jk_fixed}
\eta_{j_k}(\bz)=\eta'_{j_k}(\sigma(\bz))\qquad \forall \bz\in\calZ.
\end{equation}
Then $\tau_k=\mathrm{Id}$.
\end{lemma}

\begin{proof}
Write $j:=j_k$ and $R:=K_j$.
Since $\MQ'=\MQ$, we also have $K'_j=K_j=R$.

Recall the top-context block
\[
\mathcal M_j
:=
\bigl\{\bz\in\calZ:\ z_h=M_h-1\ \text{for all }h\in R\bigr\},
\qquad
m:=|\mathcal M_j|=\prod_{h\notin R} M_h.
\]
Define $\mathcal M'_j$ analogously for the alternative model. Since $K'_j=K_j$, we have $\mathcal M'_j=\mathcal M_j$.

By Assumption~\ref{assm:poly_monotone_binary_style} in the true model, every $\bz\in\mathcal M_j$
has $\eta_j(\bz)$ strictly larger than every $\bz\notin\mathcal M_j$.
Equivalently, $\mathcal M_j$ is the unique subset $S\subseteq\calZ$ with $|S|=m$ such that
\[
\max_{\bz\notin S}\eta_j(\bz)\ <\ \min_{\bz\in S}\eta_j(\bz).
\]
The same uniqueness statement holds for $\eta'_j$ with $\mathcal M'_j$ by applying
Assumption~\ref{assm:poly_monotone_binary_style} to the alternative model.

Now use the alignment \eqref{eq:eta_align_item_jk_fixed}: for any $\bz,\bz'\in\calZ$,
\[
\eta_j(\bz)>\eta_j(\bz')
\Longleftrightarrow
\eta'_j(\sigma(\bz))>\eta'_j(\sigma(\bz')).
\]
Hence $\sigma$ maps the set of states attaining the top $m$ values of $\eta_j$ onto the set of states attaining the top $m$
values of $\eta'_j$.
By uniqueness of the corresponding size-$m$ separated block, this implies
$\sigma(\mathcal M_j)=\mathcal M'_j=\mathcal M_j$.

Since $\sigma$ has the coordinatewise form \eqref{eq:sigma_within_only_fixed}, for any $h\in R$ and any $\bz\in\mathcal M_j$, $(\sigma(\bz))_h=\tau_h(M_h-1)$. But $\sigma(\bz)\in\mathcal M_j$ forces $(\sigma(\bz))_h=M_h-1$ for all $h\in R$.
Therefore
\begin{equation}\label{eq:tau_fix_top_fixed}
\tau_h(M_h-1)=M_h-1
\qquad\forall h\in R.
\end{equation}

Next, for $t\in\{0,1,\ldots,M_k-1\}$, define a top-context state $\bz^{(t)}\in\calZ$ by
\[
z_k=t,\qquad z_h=M_h-1\ \ (h\in R\setminus\{k\}),
\qquad \text{and arbitrary outside }R.
\]
By Lemma~\ref{lem:Q_structural_restriction_for_eta}, the value of $\eta_j(\bz^{(t)})$ depends only on $\bz^{(t)}_R$,
so the choice outside $R$ is irrelevant.
Define the one-dimensional profile
$g(t)=\eta_j(\bz^{(t)})$ for $t\in\{0,1,\ldots,M_k-1\}$.
We claim that $g$ is strictly increasing.
Fix any $u\in[M_k-1]$.
Since $k\in R=K_j$ (because $q_{j,k}=1$), Assumption~\ref{assm:poly_monotone_binary_style} gives
\[
\max_{\MZ\in\mathcal T^{(k,u)}_{0,j}} \eta_j(\MZ)\ <\ \min_{\MZ\in\mathcal T^{(k,u)}_{1,j}} \eta_j(\MZ).
\]
Taking $\MZ=\bz^{(t)}$ with $t\le u-1$ in $\mathcal T^{(k,u)}_{0,j}$ and with $t\ge u$ in $\mathcal T^{(k,u)}_{1,j}$ yields $g(u-1)<g(u)$, hence $g(0)<g(1)<\cdots<g(M_k-1)$.

The same ladder also holds in the alternative model along coordinate $k$. Define $\bz'^{(t)}$ in the alternative model analogously by setting
$z_k=t$, $z_h=M_h-1$ for $h\in R\setminus\{k\}$, and arbitrary outside $R$,
and set $g'(t):=\eta'_j(\bz'^{(t)})$.
Since $K'_j=K_j=R$, Assumption~\ref{assm:poly_monotone_binary_style} applied to the alternative model yields, for each $u\in[M_k-1]$, $g'(u-1)<g'(u)$ and therefore $g'(0)<g'(1)<\cdots<g'(M_k-1)$.

By \eqref{eq:tau_fix_top_fixed}, for every $h\in R\setminus\{k\}$ we have $\tau_h(M_h-1)=M_h-1$, hence $\sigma(\bz^{(t)})=\bz'^{(\tau_k(t))}$ for $t\in\{0,1,\ldots,M_k-1\}$. Using \eqref{eq:eta_align_item_jk_fixed}, we obtain
\[
g(t)=\eta_j(\bz^{(t)})=\eta'_j(\sigma(\bz^{(t)}))=\eta'_j(\bz'^{(\tau_k(t))})=g'(\tau_k(t)).
\]
If $s<t$, then we have $g(s)<g(t)$, so $g'(\tau_k(s))<g'(\tau_k(t))$.
Since $g'$ is strictly increasing, it follows that $\tau_k(s)<\tau_k(t)$ for all $s<t$.
Thus $\tau_k$ is a strictly increasing permutation of $\{0,1,\ldots,M_k-1\}$, hence $\tau_k=\mathrm{Id}$.
\end{proof}

\smallskip
Applying Lemma~\ref{lem:kill_within_coord_perm_binary_style} to each $k\in[K]$ yields $\tau_k=\mathrm{Id}$ for all $k$,
so the stabilized relabeling reduces to $\sigma=\sigma_{\varpi}$.

\subsubsection{Recover $\{\beta_j\}$ after aligning by $\sigma_{\varpi}$.}

For each item $j$, define the $r$-vector of identified linear predictors
\[
\eta_j:=\bigl(\eta_j(\bz)\bigr)_{\bz\in\calZ}\in\RR^r,
\qquad
\boeta'_j:=\bigl(\eta'_j(\bz)\bigr)_{\bz\in\calZ}\in\RR^r.
\]
After aligning by $\sigma:=\sigma_{\varpi}$, we have
\begin{equation}\label{eq:eta_table_alignment_vector}
\boeta_j = Q_{\sigma}\,\boeta'_j,
\end{equation}
where $Q_{\sigma}$ is the $r\times r$ permutation matrix acting on the state index $\bz\in\calZ$.

List $\MZ\in\calZ$ in lexicographic order and set $F=(\Phi(\MZ)^\top)_{\MZ\in\calZ}\in\RR^{r\times r}$. With the same one-coordinate ordering, let $F_k$ be the $M_k\times M_k$ matrix
\[
F_k(t,s)=
\begin{cases}
1, & s=0,\\
\mathbf 1\{t\ge s\}, & s\in\{1,\ldots,M_k-1\},
\end{cases}
\qquad t\in\{0,\ldots,M_k-1\}.
\]
Then $F_k$ is lower triangular with ones on the diagonal, hence invertible, and the lexicographic construction gives $F = F_1\otimes F_2\otimes \cdots \otimes F_K$, so $F$ is invertible.

Therefore, the saturated coefficients satisfy
\begin{equation}\label{eq:beta_from_eta}
\boeta_j = F\,\bbeta_j,
\qquad
\boeta'_j = F\,\bbeta'_j,
\qquad\text{equivalently}\qquad
\bbeta_j = F^{-1}\boeta_j,
\ \ \bbeta'_j = F^{-1}\boeta'_j.
\end{equation}

Moreover, the coordinate relabeling $\sigma=\sigma_{\varpi}$ acts on the design vectors by $\Phi(\sigma(\bz)) = P_{\varpi}\Phi(\bz)$, which implies the matrix identity
\begin{equation}\label{eq:F_perm_identity}
Q_{\sigma}\,F = F\,P_{\varpi}^{\top}.
\end{equation}
Indeed, for each $\bz\in\calZ$,
\[
(Q_{\sigma}F)_{\bz,:}
=
F_{\sigma^{-1}(\bz),:}
=
\Phi(\sigma^{-1}(\bz))^\top
=
\bigl(P_{\varpi}^{-1}\Phi(\bz)\bigr)^\top
=
\Phi(\bz)^\top P_{\varpi}^{\top}
=
(FP_{\varpi}^{\top})_{\bz,:}.
\]

Combining \eqref{eq:eta_table_alignment_vector}--\eqref{eq:F_perm_identity} yields
\[
\bbeta_j
=
F^{-1}\boeta_j
=
F^{-1}Q_{\sigma}\boeta'_j
=
F^{-1}Q_{\sigma}F\,\bbeta'_j
=
F^{-1}F P_{\varpi}^{\top}\bbeta'_j
=
P_{\varpi}^{\top}\bbeta'_j,
\]
hence
\begin{equation}\label{eq:beta_transformation}
\bbeta'_j = P_{\varpi}\,\bbeta_j,
\qquad j\in[J],
\end{equation}
which completes the proof.

\subsection{Proof of Theorem \ref{cor: local consistency}}\label{sec: proof of cor4}
We first recall two other properties of a scoring criterion \citep{chickering2002optimal}: decomposability and consistency.

First, decomposability requires that the score can be expressed as $S(G,\calD)=\sum_{i=1}^ns(R_i,\\\mathbf{Pa}_i^G)$, where each local term depends only on $R_i$ and its parent set $\MPa_i^G$.

Second, a score is said to be consistent if, in the limit as $N\to\infty$, we have $S(H,\mathcal D)>S(G,\mathcal D)$ whenever $p^\star\in\mathcal M(H)$ and $p^\star\notin\mathcal M(G)$, and $S(H,\mathcal D)<S(G,\mathcal D)$ whenever $p^\star\in\mathcal M(H)\cap\mathcal M(G)$ and $G$ contains fewer parameters than $H$.

By Lemma 7 in \cite{chickering2002optimal}, a decomposable consistent score is locally consistent. Since BDeu is decomposable, we are ready to introduce another useful concept, namely $\{c_N\}$-consistency, which is a rate-robust version of consistency and will assist our proof.
\begin{definition}\label{def:c_N consistency}
Let $p^\star=P_{\Mtheta^\star}$ and $\calD_N=\{\MR_{N,1},\dots,\MR_{N,N}\}$ be i.i.d.\ from some $p_N=P_{\Mtheta_N}$. We say a score $S$ is $\{c_N\}$-consistent if $\|\Mtheta_N-\Mtheta^\star\|=O_p(1/c_N)$ with $c_N\to\infty$, and as $N\to\infty$ the following hold:
(i) if $p^\star\in\calM(H)$ and $p^\star\notin\calM(G)$, then $S(H,\calD_N)>S(G,\calD_N)$ with probability $\to1$;
(ii) if $p^\star\in\calM(H)\cap\calM(G)$ and $G$ contains fewer parameters than $H$, then $S(G,\calD_N)>S(H,\calD_N)$ with probability $\to1$.
\end{definition}

We have the following lemma.
\begin{lemma}
Decomposability together with $\{c_N\}$-consistency implies $\{c_N\}$-local consistency.
\end{lemma}
\begin{proof}

Fix a DAG $G$ and an addition $i\to j$, and write $G'=G+(i\to j)$. By decomposability,\[S(G',D_N)-S(G,D_N)= s\big(R_j,\MPa^G_j\cup\{R_i\}\big) - s\big(R_j,\MPa_j^G\big),\]so the score change depends only on the local family at $X_j$.

To analyze this local change, construct a convenient comparison pair $(H,H')$ as follows. Choose a total order $\tau$ of the vertices in which every node in $\MPa^G_j$ comes before $i$, $i$ comes before $j$, and $j$ comes before all remaining nodes. Let $H'$ be the complete (tournament) DAG consistent with $\tau$ (i.e., orient every pair $u\prec_\tau v$ as $u\to v$). Then $\MPa^{H'}_j=\MPa^G_j\cup\{R_i\}$. Define $H$ by deleting the single edge $i\to j$ from $H'$. Deleting an edge preserves acyclicity and yields $\MPa^{H}_j=\MPa^G_j$, and $H'=H+(i\to j)$. Since $H$ and $H'$ differ only at the family of $R_j$, decomposability gives\[S(H',D_N)-S(H,D_N)= s\big(R_j,\MPa^G_j\cup\{R_i\}\big) - s\big(R_j,\MPa_j^G\big)= S(G',D_N)-S(G,D_N).\]

Now apply $\{c_N\}$-consistency to the global comparison between $H$ and $H'$. In the dependence case $R_j \not\!\perp\!\!\!\perp_{p^\star} R_i\mid \MPa_j^G$, the complete DAG $H'$ imposes no conditional-independence constraints and therefore contains $p^\star$, whereas $H$ enforces the false constraint $R_j \perp\!\!\!\perp R_i \mid \MPa^G_j$ and thus excludes $p^\star$. By $\{c_N\}$-consistency, $S(H',D_N)>S(H,D_N)$ with probability $\to 1$, hence $S(G',D_N)>S(G,D_N)$ with probability $\to 1$.

In the independence case $R_j \perp\!\!\!\perp_{p^\star} R_i\mid \MPa_j^G$, both $H$ and $H'$ contain $p^\star$, but $H'$ has strictly more parameters (one extra parent for $R_j$). By the second clause of $\{c_N\}$-consistency, $S(H,D_N)>S(H',D_N)$ with probability $\to 1$, hence $S(G,D_N)>S(G',D_N)$ with probability $\to 1$. This is exactly $\{c_N\}$-local consistency.
\end{proof}

Now it suffices to show BDeu is $\{c_N\}$-consistent for discrete causal graphical models, where $c_N=\omega(\sqrt{\frac{N}{\log N}})$.

Denote the finite state space by $\calZ$. By Assumption~\ref{assm1}(a) there exists $\varepsilon \in (0, \tfrac12\min_{\bz}\p^\star_\bz).$ Define the high-probability event $\mathcal E_N:=\{\|\p_N-\p^\star\|_\infty<\varepsilon\}$, for which  $\PP(\mathcal E_N)\to1$ since $\|\p_N-\p^\star\|=O_p(\tfrac{1}{c_N})$ while $c_N\to\infty$. Since $\PP(\calE_N)\to 1$, it suffices to establish all subsequent asymptotic statements on $\calE_N$. 

On $\mathcal E_N$ we have
$\min_\bz (\p_N)_\bz \ge \varepsilon$. Fix any baseline $\bz_0$ and enumerate the remaining $d = |\calZ|-1$ states as
$\bz_1,\dots,\bz_d$. Define $\calT_\bz(\p):= \log\!\Big(\frac{\p_\bz}{\p_{\bz_0}}\Big)$ for $ \bz_1,\dots,\bz_d$. Set $\Mtheta^\star:=\calT(\p^\star)$ and ${\Mtheta}_N:=\calT(\p_N)$, where $\calT$ is a $C^\infty$ diffeomorphism from $\Delta^\circ_{d}$ onto $\RR^d$. The whole line segment
$[\p^\star,\p_N] := \big\{\p^\star+t(\p_N-\p^\star):t\in[0,1]\big\}$ is contained in the convex, compact set
$K_\varepsilon := \big\{\p: \p_\bz\ge\varepsilon \text{ for all }\bz\text{ and }  \sum_{\bz}\p_\bz=1\big\}$,
which lies strictly inside the positive simplex.

The map $\calT$ is a composition of an affine map and the coordinatewise logarithm on the open set
$\{\p_\bz>0,\ \sum_\bz\p_\bz=1\}$, so $\calT$ is continuously differentiable there and the Jacobian
$\nabla \calT(\p)$ is a continuous function of $\p$. Since $K_\varepsilon$ is compact, we have
\[
L_\varepsilon\ :=\ \sup_{\p\in K_\varepsilon}\ \big\|\nabla \calT(\p)\big\|_{\mathrm{op}}\ <\ \infty.
\]
Therefore $\calT$ is $L_\varepsilon$-Lipschitz on $K_\varepsilon$. Applying the integral form of Taylor’s theorem along the segment $[\p^\star,\p_N]$, we obtain
\[
\|{\Mtheta}_N-\Mtheta^\star\|_2
 \le
L_\varepsilon\|\p_N-\p^\star\|_2=\ O_p(\tfrac{1}{c_N}).
\]


On the event $\calE_N$, every coordinate of $\p^\star$ and $\p_N$ is strictly positive, so both
$\Mtheta^\star = \calT(\p^\star)$ and ${\Mtheta}_N = \calT(\p_N)$ lie in the natural parameter
space associated with the saturated multinomial family on the finite state space $\calZ$.
For $\bz\in\calZ$ define the sufficient statistic
\[
  \phi_j(\bz) := \mathbf 1\{\bz = \bz_j\}, \qquad j=1,\dots,d,
  \qquad
  \phi(\bz) := \bigl(\phi_1( \bz),\dots,\phi_d(\bz)\bigr)^\top \in \RR^d.
\]
With respect to the counting measure on $\calZ$, the saturated multinomial family can be written in
canonical exponential-family form as
\[
  p_\Mtheta(\bz)
  \;=\;
  h(\bz)\exp\!\bigl\{\langle \Mtheta,\phi(\bz)\rangle - A(\Mtheta)\bigr\},
  \qquad \bz\in\calZ,
\]
where we take $h(\bz)\equiv 1$ and $\Mtheta=(\theta_1,\dots,\theta_d)^\top\in\MXi=\RR^d$, and the
log-partition function is
\[
  A(\Mtheta)
  \;:=\;
  \log\Bigl(1 + \sum_{j=1}^d e^{\theta_j}\Bigr).
\]
This saturated multinomial family is regular and minimal. The mean map
$\Mmu(\Mtheta):=\nabla_\Mtheta A(\Mtheta)$ has components
\(
  \mu_j(\Mtheta) = \PP_\Mtheta(\MZ=\bz_j)
\),
and the Fisher information
\(
  I(\Mtheta):=\nabla_\Mtheta^2 A(\Mtheta)
\)
is continuous on $\MXi=\RR^d$. In particular, at the true parameter
$\Mtheta^\star = \calT(\p^\star)$ we have $I(\Mtheta^\star)\succ 0$.

By continuity of $I(\Mtheta)$, we can choose a convex open neighborhood (for example, a small open ball) $U_\Mtheta\subset\MXi$ with $\Mtheta^\star\in U_\Mtheta$ such that
\[
  \lambda_+ \succeq I(\Mtheta) \succeq \lambda_- > 0
  \qquad (\Mtheta\in U_\Mtheta).
\]
Then $A$ is $\lambda_-$–strongly convex on $U_\Mtheta$, so $\nabla A$ is injective and, by the inverse function theorem, the gradient map $\nabla A: U_\Mtheta \to U_\Mmu:=\nabla A(U_\Mtheta)$ is a $C^\infty$ diffeomorphism onto its image, and $\Mmu^\star:=\Mmu(\Mtheta^\star)\in U_\Mmu$. We can further pick $s>0$ such that $B(\Mmu^\star,s)\subseteq U_\Mmu$.
All model comparisons are taken over feasible sets $\MXi_M=\{\Mtheta: \calT^{-1}(\Mtheta)\in\calM(M)\}$ with $ M\in\{G,H\}.$ Write $f(\Mtheta;\Mmu):=\langle\Mmu,\Mtheta\rangle-A(\Mtheta)$ and $V_M(\Mmu):=\sup_{\Mtheta\in\MXi_M}f(\Mtheta;\Mmu)$.


Recall that $\MR_{N,1},\dots,\MR_{N,N}\overset{\mathrm{i.i.d.}}{\sim}\p_{N}$. Define $\hat{\Mmu}_N:=N^{-1}\sum_{i=1}^N \phi(\MR_{N,i})$. Since $\|\Mtheta_N-\Mtheta^\star\|=O_p(\frac{1}{c_N})$ and $U_\Mtheta$ is open, there exists $\tau>0$ such that $\overline{B(\Mtheta^\star,\tau)}\subseteq U_\Mtheta$ and $\mathbb{P}(\Mtheta_N\in \overline{B(\Mtheta^\star,\tau)})\to1$. 

\begin{lemma}\label{lem:V-cont}
$V_M$ is continuous at $\Mmu^\star$.
\end{lemma}
\begin{proof}
For each $\Mtheta\in \MXi_M$,
the map $\Mmu\mapsto\langle\Mmu,\Mtheta\rangle-A(\Mtheta)$ is affine. As a pointwise supremum of affine functions, $V_M$ is convex. Furthermore,
\[
V_M(\mu) = \sup_{\Mtheta\in \MXi_M}\{\langle\Mmu,\Mtheta\rangle-A(\Mtheta)\}
 \le
\sup_{\Mtheta\in\text{int}(\MXi)}\{\langle\Mmu,\Mtheta\rangle-A(\Mtheta)\}
 = A^\ast(\Mmu).
\]
If $\Mmu\in U_\Mmu$, then $\sup_{\Mtheta\in\text{int}(\MXi)}\{\langle\Mmu,\Mtheta\rangle-A(\Mtheta)\}=\langle\Mmu,(\nabla A)^{-1}(\Mmu)\rangle-A((\nabla A)^{-1}(\Mmu))<\infty$, thus $A^\ast(\Mmu)<\infty$ for all
$\Mmu\in U_\Mmu$. Thus $V_M$ is finite on $U_\Mmu$. As \(V_M\) is convex and finite on the open convex set \(B(\Mmu^\star,s)\), it is continuous throughout \(B(\Mmu^\star,s)\), in particular at \(\Mmu^\star\).
\end{proof}
\begin{lemma}\label{lem: convergence-rate}
    $\|\hat\Mmu_N-\Mmu^\star\|=O_p(N^{-1/2}+c_N^{-1})$.
\end{lemma}
\begin{proof}
Let $\calF_N:=\{\Mtheta_N\in\overline{B(\Mtheta^\star,\tau)}\}$. Because $I(\Mtheta)=\nabla^2_{\Mtheta}A(\Mtheta)$ is continuous on $U_{\Mtheta}$ and satisfies $\lambda_- \MI_d \preceq I(\Mtheta)\preceq \lambda_+\MI_d$ for all $\Mtheta\in\overline{B(\Mtheta^\star,\tau)}$, we have $\sup_{\Mtheta\in\overline{B(\Mtheta^\star,\tau)}} \operatorname{tr} I(\Mtheta)\le d\lambda_+<\infty$.

For any $t>0$, write the probability with the intersection as $$\PP\big(\sqrt{N}\|\hat{\Mmu}_N-\Mmu(\Mtheta_N)\|>t, \calF_N\big)=\E\big[\mathbf 1_{\calF_N}\PP\big(\sqrt{N}\|\hat{\Mmu}_N-\Mmu(\Mtheta_N)\|>t\mid \Mtheta_N\big)\big].$$ On $\{\Mtheta_N=\Mtheta\in \calF_N\}$ we have $\mathrm{Var}(\hat{\Mmu}_N)=I(\Mtheta)/N$, hence $\E\big[N\|\hat{\Mmu}_N-\Mmu(\Mtheta_N)\|^2\mid \Mtheta_N\big]=\operatorname{tr} I(\Mtheta_N)$. Chebyshev’s inequality yields $$\PP\big(\sqrt{N}\|\hat{\Mmu}_N-\Mmu(\Mtheta_N)\|>t, \calF_N\big)\le t^{-2}\E\big[\mathbf 1_{\calF_N}\operatorname{tr} I(\Mtheta_N)\big]\le d\lambda_+/t^2,$$ equivalently $\PP\big(\|\hat{\Mmu}_N-\Mmu(\Mtheta_N)\|>t/\sqrt{N}, \calF_N\big)\le d\lambda_+/t^2$ for all $t>0$. 

By continuity of $\nabla^2 A$ on the compact $\overline{B(\Mtheta^\star,\tau)}$, there exists $L<\infty$ with $\|\nabla^2A(\Mtheta)\|_{\mathrm{op}}\le L$ on this set, so $\nabla A$ is $L$-Lipschitz there. Consequently, for any $\eta>0$ we have $$\PP\!\big(\|\Mmu(\Mtheta_N)-\Mmu^\star\|>\eta,\ \calF_N\big)\le \PP\!\big(L\|\Mtheta_N-\Mtheta^\star\|>\eta\big)=\PP\!\big(\|\Mtheta_N-\Mtheta^\star\|>\eta/L\big).$$ Taking $\eta=K'/c_N$ gives $\PP\!\big(\|\Mmu(\Mtheta_N)-\Mmu^\star\|>K'/c_N,\ \calF_N\big)=o(1)$ because $\|\Mtheta_N-\Mtheta^\star\|=O_p(c_N^{-1})$.

For any $K,K'>0$, 
\begin{align*}
&\PP\!\big(\|\hat{\Mmu}_N-\Mmu^\star\|>\frac{K}{\sqrt{N}}+\frac{K'}{c_N}\big)
\\\le& \PP\!\big(\|\hat{\Mmu}_N-\Mmu(\Mtheta_N)\|>\frac{K}{\sqrt{N}},\calF_N\big)
+\PP\!\big(\|\Mmu(\Mtheta_N)-\Mmu^\star\|>\frac{K'}{c_N},\calF_N\big)
+\PP(\calF_N^c).
\end{align*}
Using the two bounds above together with $\PP(\calF_N^c)\to0$ gives $\limsup_{N\to\infty}\PP\!\big(\|\hat{\Mmu}_N-\Mmu^\star\|>K/\sqrt{N}+K'/c_N\big)\le d\lambda_+/K^{2}$. Since the bound holds for all $K,K'>0$, take $K'=K$ to get, for every $K>0$, $\limsup_{N\to\infty}\PP\!\Big(\|\hat{\Mmu}_N-\Mmu^\star\|>K\big(N^{-1/2}+c_N^{-1}\big)\Big)\ \le\ \frac{d\lambda_+}{K^2}.$
Given $\varepsilon>0$, choose $K_\varepsilon\ge\sqrt{d\lambda_+/\varepsilon}$. Then there exists $N_\varepsilon$ such that for all $N\ge N_\varepsilon$,
\[
\PP\!\Big(\|\hat{\Mmu}_N-\Mmu^\star\|>K_\varepsilon\big(N^{-1/2}+c_N^{-1}\big)\Big)\le \varepsilon,
\]
which is exactly $\|\hat{\Mmu}_N-\Mmu^\star\|=O_p\!\big(N^{-1/2}+c_N^{-1}\big)$.
\end{proof}

\subsubsection*{Case 1: $\Mtheta^\star\in\MXi_H\setminus\MXi_G$}
Because $I\succ0$ on $\text{int}(\MXi)$, $f(\cdot;\Mmu^\star)$ is strictly concave on $\text{int}(\MXi)$ and has a unique maximizer $\Mtheta^\star=(\nabla A)^{-1}(\Mmu^\star)$ on $\text{int}(\MXi)$. Furthermore, since $\Mtheta^\star \in \MXi_H$ but $\Mtheta^\star \notin \MXi_G$, we define
\begin{equation}\label{eq:I}
\epsilon_0 
  := \sup_{\Mtheta \in \MXi_H} f(\Mtheta;\Mmu^\star) 
   - \sup_{\Mtheta \in \MXi_G} f(\Mtheta;\Mmu^\star)
  = f(\Mtheta^\star;\Mmu^\star) - \sup_{\Mtheta \in \MXi_G} f(\Mtheta;\Mmu^\star).
\end{equation}
Since $A$ is $\lambda_-$–strongly convex on $U_\Mtheta$, $f(\Mtheta;\Mmu^\star)=\langle\Mmu^\star,\Mtheta\rangle-A(\Mtheta)$ is $\lambda_-$–strongly concave on $U_\Mtheta$. 

Note that the independence constraints of \(G\) are given by polynomial equalities in the joint probabilities, so the set
\(\calM(G):=\{\p:\text{independence constraints of }G\text{ hold};\sum_x \p_x=1\}\) is algebraic and hence Euclidean closed. Therefore, \(\calM(G)\cap \Delta_{d}^\circ\) is closed in
\(\Delta_{d}^\circ\) under the relative topology. Since \(\calT\) is a homeomorphism,
\(\MXi_G\) is closed in \(\MXi=\mathbb{R}^{d}\). 

Because $\MXi_G$ is relatively closed in $\mathrm{int}(\MXi)$ and excludes $\Mtheta^\star$, we have $\delta=\inf_{\Mtheta\in\MXi_G}\|\Mtheta^\star-\Mtheta\|_2>0$. Pick any $r\in(0,\delta)$ with $\overline{B(\Mtheta^\star,r)}\subset U_\Mtheta$. For any $\Mtheta\in\MXi_G$, set $t=\frac{r}{\|\Mtheta-\Mtheta^\star\|}\in(0,1)$ and $\Mtheta_r=\Mtheta^\star+t(\Mtheta-\Mtheta^\star)$. Then $\|\Mtheta_r-\Mtheta^\star\|=r$ and $\Mtheta_r\in \overline{B(\Mtheta^\star,r)}\subset U_\Mtheta$. By concavity of $f(\cdot;\Mmu^\star)$ on the convex set $\text{int}(\MXi)$, $f(\Mtheta_r;\Mmu^\star) \ge (1-t)f(\Mtheta^\star;\Mmu^\star)+tf(\Mtheta;\Mmu^\star)$,
hence $f(\Mtheta;\Mmu^\star)\le f(\Mtheta_r;\Mmu^\star)$. By $\lambda_-$–strong concavity on $U_\Mtheta$,
\[
f(\Mtheta^\star;\Mmu^\star)-f(\Mtheta_r;\Mmu^\star)\ \ge\ \frac{\lambda_-}{2}\,\|\Mtheta_r-\Mtheta^\star\|^2\ =\ \frac{\lambda_-}{2}\,r^2.
\]
Therefore, we have
\[
f(\Mtheta;\Mmu^\star)\ \le\ f(\Mtheta^\star;\Mmu^\star)-\frac{\lambda_-}{2}\,r^2\qquad(\forall\,\Mtheta\in\MXi_G).
\]
It follows that $\epsilon_0\ge(\lambda_-/2)\,r^2>0$.

By Lemma~\ref{lem:V-cont} and Lemma~\ref{lem: convergence-rate}, with probability
$\to1$,
\[
\sup_{\Mtheta \in \MXi_H} f(\Mtheta;\hat\Mmu_N)-\sup_{\Mtheta \in \MXi_G} f(\Mtheta;\hat\Mmu_N)\ \ge\ \epsilon_0/2,
\]
hence
\begin{align*}
S(H,\mathcal D)-S(G,\mathcal D)&= \tfrac12(d_\calG-d_\calH)\log N\ +\ N\!\Big[\sup_{\Mtheta \in \MXi_H} f(\Mtheta;\hat\Mmu_N)-\sup_{\Mtheta \in \MXi_G} f(\Mtheta;\hat\Mmu_N)\Big]+O_p(1)
\\&>\tfrac12(d_\calG-d_\calH)\log N+\tfrac{N\epsilon_0}{2}+O_p(1)
\\&>0.
\end{align*}
We now make explicit why the BDeu (BD) score admits a BIC expansion with an $O_p(1)$ remainder under the triangular-array sampling scheme $\MR_{N,1},\dots,\MR_{N,N}\overset{\mathrm{i.i.d.}}{\sim}p_N$, where $p_N$ is random. Recall that when $p_N$ is fixed, for discrete causal graphical models, Theorem~18.1 of \citet{koller2009probabilistic} gives
\[
S_{\text{BDeu}}(\mathcal D\mid G)=\ell(\hat\Mtheta;\mathcal D)-\tfrac12 d_G\log N+O(1).
\]
However, since $p_N$ is random here, we need to show why this remainder becomes $O_p(1)$.

Fix a candidate DAG $M$ on the discrete vector $\MZ=(Z_v)_{v\in V}$.
For each node $v\in V$, write $r_v$ and $q_v$ for the number of states of $Z_v$ and the number of parent configurations of $\MPa_M(v)$, respectively. Index the parent configurations of $\MPa_M(v)$ by $u\in\{1,\dots,q_v\}$ and the states of $Z_v$ by
$k\in\{1,\dots,r_v\}$.
Given the dataset $\mathcal D:=\{\MR_{N,1},\dots,\MR_{N,N}\}$ with $\MR_{N,i}\in\calZ$,
define the empirical counts
\[
N_{vuk}
:=
\sum_{i=1}^N \mathbf 1\Big\{ (\MR_{N,i})_{\MPa_M(v)}=u,\ (\MR_{N,i})_v=k\Big\},
\qquad
N_{vu}
:=
\sum_{k=1}^{r_v} N_{vuk}
=
\sum_{i=1}^N \mathbf 1\Big\{ (\MR_{N,i})_{\MPa_M(v)}=u\Big\}.
\]
Write $\theta_{vuk}:=\PP(Z_v=k\mid \MPa_M(v)=u)$ and $\theta_{vu}:=(\theta_{vu1},\dots,\theta_{vur_v})\in\Delta_{r_v-1}$,
and assume independent Dirichlet priors:
for each $(v,u)$,
\[
\theta_{vu}\sim\mathrm{Dir}(\alpha_{vu1},\dots,\alpha_{vur_v}),
\qquad
\alpha_{vuk}>0,
\qquad
\alpha_{vu}:=\sum_{k=1}^{r_v}\alpha_{vuk}.
\]
For BDeu, one uses the special choice $\alpha_{vuk}\equiv \alpha/(q_v r_v)$, hence $\alpha_{vu}\equiv \alpha/q_v$.

Now we write
\begin{align}
\log P(\mathcal D\mid M)
&=
\sum_{v\in V}\sum_{u=1}^{q_v}
\Bigg[
\log\Gamma(\alpha_{vu})
-\log\Gamma(\alpha_{vu}+N_{vu})
+\sum_{k=1}^{r_v}\Big(\log\Gamma(\alpha_{vuk}+N_{vuk})-\log\Gamma(\alpha_{vuk})\Big)
\Bigg].
\\&=
C_M(\alpha)
+
\sum_{v\in V}\sum_{u=1}^{q_v}
\Bigg[
\sum_{k=1}^{r_v}\log\Gamma(\alpha_{vuk}+N_{vuk})
-\log\Gamma(\alpha_{vu}+N_{vu})
\Bigg],
\label{eq:BD-log-split}\end{align}
where $C_M(\alpha)=
\sum_{v\in V}\sum_{u=1}^{q_v}
[
\log\Gamma(\alpha_{vu})
-\sum_{k=1}^{r_v}\log\Gamma(\alpha_{vuk})
]$
depends only on $M$ and the hyperparameters.

We now analyze one fixed CPT row $(v,u)$.
To simplify notation in this row, write
\[
r:=r_v,\quad
n_k:=N_{vuk},\quad
n:=N_{vu}=\sum_{k=1}^r n_k,\quad
a_k:=\alpha_{vuk},\quad
a:=\alpha_{vu}=\sum_{k=1}^r a_k.
\]
In what follows we work on the event $\min_{1\le k\le r} n_k\ge 1$, so that $\log n_k$, $\log\hat\theta_k$, and $\log(1+a_k/n_k)$ are well-defined. We will show this event holds with probability $\to1$ on $\calE_N$ later (see \eqref{eq:EA-prob}).

Define 
$\hat\theta_k:=\frac{n_k}{n}$, then the row log-likelihood at the MLE is
\[
\ell_{vu}(\hat\theta;\mathcal D)
:=
\sum_{k=1}^r n_k\log\hat\theta_k
=
\sum_{k=1}^r n_k\log\Big(\frac{n_k}{n}\Big).
\]

We apply Stirling's expansion with an explicit remainder:
for all $x\ge 1$,
\begin{equation}\label{eq:stirling}
\log\Gamma(x)
=
\Bigl(x-\tfrac12\Bigr)\log x-x+\tfrac12\log(2\pi)+\eta(x),
\qquad
|\eta(x)|\le \frac{1}{12x}.
\end{equation}
Apply \eqref{eq:stirling} to each $\log\Gamma(a_k+n_k)$ and to $\log\Gamma(a+n)$.
We obtain
\begin{align}
\sum_{k=1}^r\log\Gamma(a_k+n_k)-\log\Gamma(a+n)
&=
\sum_{k=1}^r\Bigl(a_k+n_k-\tfrac12\Bigr)\log(a_k+n_k)
-\Bigl(a+n-\tfrac12\Bigr)\log(a+n)
\notag\\
&\quad
+\frac{r-1}{2}\log(2\pi)
+\sum_{k=1}^r\eta(a_k+n_k)-\eta(a+n).
\label{eq:row-after-stirling}
\end{align}

Define
\begin{align}
T_1
&:=
\sum_{k=1}^r\Bigl(a_k+n_k-\tfrac12\Bigr)\log n_k
-\Bigl(a+n-\tfrac12\Bigr)\log n.
\label{eq:T1-def}
\\T_2
&:=
\sum_{k=1}^r\Bigl(a_k+n_k-\tfrac12\Bigr)\log\Big(1+\frac{a_k}{n_k}\Big)
-\Bigl(a+n-\tfrac12\Bigr)\log\Big(1+\frac{a}{n}\Big).
\label{eq:T2-def}
\end{align}
Then we have:
\begin{align}
T_1
&=
\sum_{k=1}^r\Bigl(a_k+n_k-\tfrac12\Bigr)\Bigl(\log n+\log\Big(\frac{n_k}{n}\Big)\Bigr)
-\Bigl(a+n-\tfrac12\Bigr)\log n
\\
&=
\Biggl(\sum_{k=1}^r(a_k+n_k-\tfrac12)-(a+n-\tfrac12)\Biggr)\log n
+
\sum_{k=1}^r\Bigl(a_k+n_k-\tfrac12\Bigr)\log\Big(\frac{n_k}{n}\Big).
\\
&=
-\frac{r-1}{2}\log n
+
\sum_{k=1}^r n_k\log\Big(\frac{n_k}{n}\Big)
+
\sum_{k=1}^r\Bigl(a_k-\tfrac12\Bigr)\log\Big(\frac{n_k}{n}\Big)
\notag\\
&=
-\frac{r-1}{2}\log n
+
\ell_{vu}(\hat\theta;\mathcal D)
+
\sum_{k=1}^r\Bigl(a_k-\tfrac12\Bigr)\log\hat\theta_k.
\label{eq:T1-final}
\end{align}

Combining \eqref{eq:row-after-stirling} with \eqref{eq:T1-final} and \eqref{eq:T2-def} gives 
\begin{align}
\sum_{k=1}^r\log\Gamma(a_k+n_k)-\log\Gamma(a+n)
&=
\ell_{vu}(\hat\theta;\mathcal D)
-\frac{r-1}{2}\log n
+\sum_{k=1}^r\Bigl(a_k-\tfrac12\Bigr)\log\hat\theta_k
\notag\\
&\quad
+T_2
+\frac{r-1}{2}\log(2\pi)
+\sum_{k=1}^r\eta(a_k+n_k)-\eta(a+n).
\label{eq:row-identity}
\end{align}

We now return to the original indices.
For each row $(v,u)$, define
\[
\hat\theta_{vuk}:=\frac{N_{vuk}}{N_{vu}},
\qquad
\ell_{vu}(\hat\theta;\mathcal D):=\sum_{k=1}^{r_v}N_{vuk}\log\hat\theta_{vuk},
\]
and define
\begin{align}
T_{2,vu}
&:=
\sum_{k=1}^{r_v}\Bigl(\alpha_{vuk}+N_{vuk}-\tfrac12\Bigr)\log\Big(1+\frac{\alpha_{vuk}}{N_{vuk}}\Big)
-\Bigl(\alpha_{vu}+N_{vu}-\tfrac12\Bigr)\log\Big(1+\frac{\alpha_{vu}}{N_{vu}}\Big),
\label{eq:T2vu}
\\
E_{vu}
&:=
\sum_{k=1}^{r_v}\eta(\alpha_{vuk}+N_{vuk})-\eta(\alpha_{vu}+N_{vu}).
\label{eq:Evu}
\end{align}
Then \eqref{eq:row-identity} becomes, for every $(v,u)$,
\begin{align}
\sum_{k=1}^{r_v}\log\Gamma(\alpha_{vuk}+N_{vuk})-\log\Gamma(\alpha_{vu}+N_{vu})
&=
\ell_{vu}(\hat\theta;\mathcal D)
-\frac{r_v-1}{2}\log N_{vu}
+\sum_{k=1}^{r_v}\Bigl(\alpha_{vuk}-\tfrac12\Bigr)\log\hat\theta_{vuk}
\notag\\
&\quad
+T_{2,vu}
+\frac{r_v-1}{2}\log(2\pi)
+E_{vu}.
\label{eq:row-final}
\end{align}

Sum \eqref{eq:row-final} over all $v$ and $u$ and substitute into \eqref{eq:BD-log-split}.
This yields
\begin{align}
\log P(\mathcal D\mid M)
&=
\ell(\hat\theta_M;\mathcal D)
-\frac12\sum_{v\in V}\sum_{u=1}^{q_v}(r_v-1)\log N_{vu}
+R'_{M,N},
\label{eq:global-before-logN}
\\
\ell(\hat\theta_M;\mathcal D)
&:=
\sum_{v\in V}\sum_{u=1}^{q_v}\ell_{vu}(\hat\theta;\mathcal D)
=
\sum_{v\in V}\sum_{u=1}^{q_v}\sum_{k=1}^{r_v}N_{vuk}\log\Big(\frac{N_{vuk}}{N_{vu}}\Big),
\notag\\
R'_{M,N}
&:=
C_M(\alpha)
+\frac12\sum_{v\in V}\sum_{u=1}^{q_v}(r_v-1)\log(2\pi)
+\sum_{v\in V}\sum_{u=1}^{q_v}\sum_{k=1}^{r_v}\Bigl(\alpha_{vuk}-\tfrac12\Bigr)\log\hat\theta_{vuk}
\notag\\
&\quad
+\sum_{v\in V}\sum_{u=1}^{q_v}T_{2,vu}
+\sum_{v\in V}\sum_{u=1}^{q_v}E_{vu}.
\label{eq:Rprime-def}
\end{align}

Define the model dimension
$d_M=\sum_{v\in V}\sum_{u=1}^{q_v}(r_v-1)$.
\begin{align}
-\frac12\sum_{v,u}(r_v-1)\log N_{vu}
&=
-\frac12\sum_{v,u}(r_v-1)\log N
-\frac12\sum_{v,u}(r_v-1)\log\Big(\frac{N_{vu}}{N}\Big)
\notag\\
&=
-\frac12 d_M\log N
-\frac12\sum_{v,u}(r_v-1)\log\Big(\frac{N_{vu}}{N}\Big).
\label{eq:split-log}
\end{align}
Substitute \eqref{eq:split-log} into \eqref{eq:global-before-logN} to obtain
\begin{align}
\log P(\mathcal D\mid M)
&=
\ell(\hat\theta_M;\mathcal D)-\frac12 d_M\log N + R_{M,N},
\label{eq:BD-BIC}
\\
R_{M,N}
&:=
R'_{M,N}
-\frac12\sum_{v\in V}\sum_{u=1}^{q_v}(r_v-1)\log\Big(\frac{N_{vu}}{N}\Big).
\label{eq:R-def}
\end{align}

We now prove that $R_{M,N}=O_p(1)$ under our sampling. Fix $(v,u,k)$ and define the cylinder set
\(A_{vuk}:=\{z\in\calZ:\ z_{\MPa_M(v)}=u,\ z_v=k\}.
\) On $\calE_N$ we then have
\[
p_N(A_{vuk})
=
\sum_{z\in A_{vuk}}(p_N)_z
\ge
\min_{z\in A_{vuk}}(p_N)_z
\ge
\varepsilon.
\]
Conditional on $p_N$, the count $N_{vuk}$ is binomial:
\[
N_{vuk}
=
\sum_{i=1}^N \mathbf 1\{\MR_{N,i}\in A_{vuk}\}
\ \Big|\ p_N
\sim
\mathrm{Bin}\big(N,\ p_N(A_{vuk})\big).
\]
Using the Chernoff bound, 
\begin{equation}\label{eq:chernoff-cell}
\PP\!\Big(N_{vuk}\le \tfrac12 N\,p_N(A_{vuk})\ \Big|\ p_N\Big)
\le
\exp\Big(-\frac{N\,p_N(A_{vuk})}{8}\Big)
\le
\exp\Big(-\frac{\varepsilon N}{8}\Big)
\qquad\text{on }\calE_N.
\end{equation}
Since $p_N(A_{vuk})\ge \varepsilon$ on $\calE_N$, 
\begin{equation}\label{eq:chernoff-cell-eps}
\PP\!\Big(N_{vuk}\le \tfrac{\varepsilon}{2}N\ \Big|\ p_N\Big)
\le
\exp\Big(-\frac{\varepsilon N}{8}\Big)
\qquad\text{on }\calE_N.
\end{equation}

Let $m(M):=\sum_{v\in V} q_v r_v$ be the total number of triples $(v,u,k)$.
Define the event
\[
\calA_N(M):=\Big\{\min_{v\in V}\min_{1\le u\le q_v}\min_{1\le k\le r_v} N_{vuk}\ \ge\ \tfrac{\varepsilon}{2}N\Big\}.
\]
By a union bound and \eqref{eq:chernoff-cell-eps},
\begin{align}
\PP\big(\calA_N(M)^c\mid p_N\big)
&\le
\sum_{v,u,k}\PP\!\Big(N_{vuk}\le \tfrac{\varepsilon}{2}N\ \Big|\ p_N\Big)
\le
m(M)\exp\Big(-\frac{\varepsilon N}{8}\Big)
\qquad\text{on }\calE_N.
\label{eq:unionbound}
\end{align}
Hence
\begin{align*}
\PP\big(\calA_N(M)^c\big)
&=
\PP\big(\calA_N(M)^c\cap \calE_N\big)+\PP(\calE_N^c)
\\
&=
\E\big[\mathbf 1_{\calE_N}\PP(\calA_N(M)^c\mid p_N)\big]+\PP(\calE_N^c)
\\
&\le
m(M)\exp\Big(-\frac{\varepsilon N}{8}\Big)+\PP(\calE_N^c)\ \longrightarrow\ 0.
\end{align*}
Therefore
\begin{equation}\label{eq:EA-prob}
\PP\big(\calE_N\cap \calA_N(M)\big)\to 1.
\end{equation}
On $\calA_N(M)$ we have, for every $(v,u)$, $N_{vu}=\sum_{k=1}^{r_v}N_{vuk}\ge \tfrac{\varepsilon}{2}N$, and for every $(v,u,k)$,
\[
\hat\theta_{vuk}=\frac{N_{vuk}}{N_{vu}}
\ge
\frac{N_{vuk}}{N}
\ge
\tfrac{\varepsilon}{2}.
\]
Hence, on $\calA_N(M)$,
\begin{equation}\label{eq:log-bounds}
\log\Big(\frac{N_{vu}}{N}\Big)\in\big[\log(\varepsilon/2),\,0\big],
\qquad
\log\hat\theta_{vuk}\in\big[\log(\varepsilon/2),\,0\big].
\end{equation}

We now bound each term in $R_{M,N}$ on $\calA_N(M)$.

First, $C_M(\alpha)$ and $\frac12\sum_{v,u}(r_v-1)\log(2\pi)$ are finite and do not depend on $N$.

Second, by \eqref{eq:log-bounds},
\[
\Big|\sum_{v,u,k}\Bigl(\alpha_{vuk}-\tfrac12\Bigr)\log\hat\theta_{vuk}\Big|
\le
\sum_{v,u,k}\Big|\alpha_{vuk}-\tfrac12\Big|\cdot \Big|\log(\varepsilon/2)\Big|,
\]
which is a finite constant because the sum is over finitely many $(v,u,k)$.

Third, we bound $T_{2,vu}$ defined in \eqref{eq:T2vu}.
On $\calA_N(M)$ we have $N_{vuk}\ge (\varepsilon/2)N$ and $N_{vu}\ge (\varepsilon/2)N$. Therefore, 
\begin{align*}
0
\le
\Bigl(\alpha_{vuk}+N_{vuk}-\tfrac12\Bigr)\log\Big(1+\frac{\alpha_{vuk}}{N_{vuk}}\Big)
&\le
\Bigl(\alpha_{vuk}+N_{vuk}\Bigr)\frac{\alpha_{vuk}}{N_{vuk}}\le
\alpha_{vuk}+\frac{2\alpha_{vuk}^2}{\varepsilon N}.
\end{align*}
Similarly,
\[
0
\le
\Bigl(\alpha_{vu}+N_{vu}-\tfrac12\Bigr)\log\Big(1+\frac{\alpha_{vu}}{N_{vu}}\Big)
\le
\alpha_{vu}+\frac{2\alpha_{vu}^2}{\varepsilon N}.
\]
Hence, on $\calA_N(M)$,
\begin{align}
|T_{2,vu}|
&\le
\sum_{k=1}^{r_v}\Bigl(\alpha_{vuk}+\frac{2\alpha_{vuk}^2}{\varepsilon N}\Bigr)
+\Bigl(\alpha_{vu}+\frac{2\alpha_{vu}^2}{\varepsilon N}\Bigr)
\le
2\alpha_{vu}+\frac{2}{\varepsilon N}\Big(\sum_{k=1}^{r_v}\alpha_{vuk}^2+\alpha_{vu}^2\Big).
\label{eq:T2-bound}
\end{align}
Summing \eqref{eq:T2-bound} over finitely many $(v,u)$ shows $\sum_{v,u}T_{2,vu}$ is bounded by a constant plus $O(1/N)$ on $\calA_N(M)$.

Fourth, we bound $\sum_{v,u}E_{vu}$ defined in \eqref{eq:Evu}.
On $\calA_N(M)$ we have $\alpha_{vuk}+N_{vuk}\ge N_{vuk}\ge (\varepsilon/2)N$ and
$\alpha_{vu}+N_{vu}\ge N_{vu}\ge (\varepsilon/2)N$.
Since $|\eta(x)|\le 1/(12x)$ for $x\ge 1$, we obtain, on $\calA_N(M)$,
\[
|\eta(\alpha_{vuk}+N_{vuk})|
\le
\frac{1}{12(\alpha_{vuk}+N_{vuk})}
\le
\frac{1}{12 N_{vuk}}
\le
\frac{1}{12(\varepsilon/2)N}
=
\frac{1}{6\varepsilon N},
\]
and similarly $|\eta(\alpha_{vu}+N_{vu})|\le 1/(6\varepsilon N)$.
Therefore, on $\calA_N(M)$,
\[
|E_{vu}|
\le
\sum_{k=1}^{r_v}\frac{1}{6\varepsilon N}+\frac{1}{6\varepsilon N}
=
\frac{r_v+1}{6\varepsilon N},
\]
and summing over finitely many $(v,u)$ gives $\sum_{v,u}E_{vu}=O(1/N)$ on $\calA_N(M)$.

Finally, the extra term in $R_{M,N}$ is $-\frac12\sum_{v,u}(r_v-1)\log(\frac{N_{vu}}{N})$, which is bounded on $\calA_N(M)$ because of \eqref{eq:log-bounds} and finiteness of $\sum_{v,u}(r_v-1)=d_M$.

Combining the previous bounds with \eqref{eq:EA-prob}, this implies $R_{M,N}=O_p(1)$ under the triangular-array sampling
$\MR_{N,i}\sim p_N$. Since we only compare finitely many models (in particular $\{G,H\}$), the same argument applies to each
of them, and hence all remainders appearing in the score differences can be taken as $O_p(1)$.

\subsubsection*{Case 2: $\Mtheta^\star\in\MXi_H\cap\MXi_G$ and $d_\calG<d_\calH$} 
Recall $V_M(\Mmu)=\sup_{\Mtheta\in\MXi_M}\{\langle\Mmu,\Mtheta\rangle-A(\Mtheta)\}$ and 
$A^\ast(\Mmu)=\sup_{\Mtheta\in\text{int}(\MXi)}\{\langle\Mmu,\Mtheta\rangle-A(\Mtheta)\}$.
Because $\Mtheta^\star\in\MXi_M$ and $A^\ast(\Mmu)$ maximizes over a superset $\text{int}(\MXi)$,
for any $\Mmu$ and any $M\in\{G,H\}$ we have 
\begin{equation}\label{eq:VM-sandwich}
0\ \le\ V_M(\Mmu)-f(\Mtheta^\star;\Mmu)\ \le\ A^\ast(\Mmu)-f(\Mtheta^\star;\Mmu).
\end{equation}

Recall that the gradient map $\nabla A:U_\Mtheta\to U_\Mmu$ is a $C^\infty$ diffeomorphism and $\Mtheta(\cdot)=(\nabla A)^{-1}(\cdot):U_\Mmu\to U_\Mtheta$ is a $C^\infty$ inverse map . By the Fenchel–Young equality,
\[
A^\ast(\Mmu)=\sup_{\Mtheta\in\MXi}\{\langle \Mmu,\Mtheta\rangle-A(\Mtheta)\}
=\langle \Mmu,\Mtheta(\Mmu)\rangle - A\!\big(\Mtheta(\Mmu)\big)
\qquad(\Mmu\in U_\Mmu),
\]
whence, by the chain rule,
$\nabla A^\ast(\Mmu)=\Mtheta(\Mmu) $ and $\nabla^2A^\ast(\Mmu)=\big[\nabla^2A\big(\Mtheta(\Mmu)\big)\big]^{-1}$.
Evaluating at $\Mmu^\star$ gives
$\nabla A^\ast(\Mmu^\star)=\Mtheta^\star$ and 
$\nabla^2 A^\ast(\Mmu^\star)=I(\Mtheta^\star)^{-1}$.

Therefore, by Taylor's theorem on $B(\Mmu^\star,s)$,
\begin{align}
A^\ast(\Mmu)
&=A^\ast(\Mmu^\star)+\langle\nabla A^\ast(\Mmu^\star),\,\Mmu-\Mmu^\star\rangle
+\tfrac12\,(\Mmu-\Mmu^\star)^\top \nabla^2 A^\ast(\Mmu^\star)(\Mmu-\Mmu^\star)
+o(\|\Mmu-\Mmu^\star\|^2)\notag\\
&=\big(\langle\Mmu^\star,\Mtheta^\star\rangle-A(\Mtheta^\star)\big)
+\langle \Mtheta^\star,\Mmu-\Mmu^\star\rangle
+\tfrac12\,(\Mmu-\Mmu^\star)^\top I(\Mtheta^\star)^{-1}(\Mmu-\Mmu^\star)
+o(\|\Mmu-\Mmu^\star\|^2)
\\&=f(\Mtheta^\star;\Mmu^\star)+\langle \Mtheta^\star,\Mmu-\Mmu^\star\rangle
+\tfrac12\,(\Mmu-\Mmu^\star)^\top I(\Mtheta^\star)^{-1}(\Mmu-\Mmu^\star)
+o(\|\Mmu-\Mmu^\star\|^2).\label{eq:Astar-expansion}
\end{align}
On the other hand,
\begin{equation}\label{eq:f-star-line}
f(\Mtheta^\star;\Mmu)
=\langle \Mmu,\Mtheta^\star\rangle-A(\Mtheta^\star)
=f(\Mtheta^\star;\Mmu^\star)+\langle \Mtheta^\star,\Mmu-\Mmu^\star\rangle.
\end{equation}
Subtracting \eqref{eq:f-star-line} from \eqref{eq:Astar-expansion} yields
\begin{equation}\label{eq:Astar-minus-f}
A^\ast(\Mmu)-f(\Mtheta^\star;\Mmu)
=\tfrac12(\Mmu-\Mmu^\star)^\top I(\Mtheta^\star)^{-1}(\Mmu-\Mmu^\star)
+o(\|\Mmu-\Mmu^\star\|^2).
\end{equation}
Combining \eqref{eq:VM-sandwich} and \eqref{eq:Astar-minus-f}, we conclude that there exists $0<\ell<s$ such that for $\Mmu\in B(\Mmu^\star,\ell)$,
\begin{equation}\label{eq:VM-quadratic}
0 \le\ V_M(\Mmu)-f(\Mtheta^\star;\Mmu) \le \tfrac1{\lambda_-}\|\Mmu-\Mmu^\star\|^2.
\end{equation}

Applying \eqref{eq:VM-quadratic} to $M=G,H$ and subtracting, for all $\Mmu\in B(\Mmu^\star,\ell)$ we get
\begin{equation}\label{eq:VH-VG-quad}
\big|V_H(\Mmu)-V_G(\Mmu)\big|
 \le \big|V_H(\Mmu)-f(\Mtheta^\star;\Mmu)\big|+\big|V_G(\Mmu)-f(\Mtheta^\star;\Mmu)\big|
 \le \tfrac2{\lambda_-}\|\Mmu-\Mmu^\star\|^2.
\end{equation}
By Lemma~\ref{lem: convergence-rate}, $\|\hat\Mmu_N-\Mmu^\star\|=O_p(N^{-1/2}+c_N^{-1})$, hence
\begin{equation}\label{eq:value-gap-op}
V_H(\hat\Mmu_N)-V_G(\hat\Mmu_N)
=O_p\Big(\|\hat\Mmu_N-\Mmu^\star\|^2\Big)
=O_p\Big((N^{-1/2}+c_N^{-1})^2\Big).
\end{equation}
Plugging \eqref{eq:value-gap-op} into the score difference, we obtain
\begin{align}
S(H,\mathcal D)-S(G,\mathcal D)
&=\tfrac12(d_\calG-d_\calH)\log N
+N\big[V_H(\hat\Mmu_N)-V_G(\hat\Mmu_N)\big]+O_p(1)\notag\\
&=\tfrac12(d_\calG-d_\calH)\log N
+O_p\Big(1+\sqrt{N}\,c_N^{-1}+N\,c_N^{-2}\Big)+O_p(1).\label{eq:BIC-diff-case2}
\end{align}
In particular, if $c_N=\omega\big(\sqrt{N/\log N}\big)$, then
$\sqrt{N}\,c_N^{-1}=o(\log N)$ and $N\,c_N^{-2}=o(\log N)$, so
\[
S(H,\mathcal D)-S(G,\mathcal D)
=\tfrac12(d_\calG-d_\calH)\log N\ +\ o_p(\log N)+O_p(1).
\]
Consequently, if $d_\calG<d_\calH$, then $S(H,\mathcal D)-S(G,\mathcal D)<0$ with probability $\to1$.

\subsection{Proof of Theorem~\ref{thm:consistency_pipeline}}\label{sec: proof of thm5}
Conclusion (i) follows directly from Theorem~\ref{thm:estimation consistency}. It remains to show Conclusion (ii).

Following the notation in Section~\ref{subsec: Local Consistency}, we conclude that $\widetilde{\Mtheta}_N$ is a $\sqrt{f(N)}$-consistent estimator of $\Mtheta^\star$ and $g(N)=\sqrt{N}$. Combining this with $f(N)= o(N\log N)$, we immediately have \[g(f^{-1}(N))=\omega(\sqrt{\frac{N}{\log N}}).\] 
By Theorem~\ref{cor: local consistency} and the analysis in Section~\ref{subsec: Local Consistency}, the following holds:  
let $G$ be any DAG and $G'$ be a different DAG obtained by adding the edge $i \to j$ to $G$. 
As $N \to \infty$, with probability$\to1$ we have
\begin{enumerate}
    \item[(L1)] If $Z_j \not\!\perp\!\!\!\perp_{\p^\star} Z_i \mid \mathbf{Pa}_j^{G}$, then $S(G',\hat\MZ) > S(G,\hat \MZ) $.
    \item[(L2)] If $Z_j \!\perp\!\!\!\perp_{\p^\star} Z_i \mid \mathbf{Pa}_j^{G}$, then $S(G',\hat\MZ) < S(G,\hat\MZ)$.
\end{enumerate}
Here $S(\cdot)$ is the BDeu score. 
Note that BDeu is score equivalent by Theorem 8 in \cite{chickering1995a}.

Next, we adopt the reformulation in \cite{nazaret2024extremely}. For a MEC $M$, let $\mathcal I(M)$ (insertions) and $\mathcal D(M)$ (deletions) denote, respectively, the sets of MECs reachable from $M$ by adding or deleting a single edge in some DAG representative of $M$. 
Since our score $S$ is score-equivalent, we write $S(M)$ for the common value $S(G)$ over any $G\in M$.
We introduce two propositions:
\[
P_1(M;\p^\star):\ \text{all conditional independencies encoded by $M$ hold in }\p^\star,\]
\[
P_2(M;\p^\star):\ \p^\star \text{ is DAG-perfect for every DAG in }M,
\]
By Theorems~6–8 in \cite{nazaret2024extremely}, the following three statements are enough to guarantee correctness of the two-phase greedy search in MEC space:

\begin{itemize}
\item[(A)] If $\max_{M'\in\mathcal I(M)} S(M')\le S(M)$, then $P_1(M;\p^\star)$ holds.
\item[(B)] If $P_1(M;\p^\star)$ holds and $M'\in\mathcal D(M)$ satisfies $S(M')\ge S(M)$, then $P_1(M';\p^\star)$ holds as well.
\item[(C)] If $P_1(M;\p^\star)$ holds and $\max_{M'\in\mathcal D(M)} S(M')\le S(M)$, then $P_2(M;\p^\star)$ holds.
\end{itemize}

As a result, our last task is to show these three statements hold for score equivalent scores satisfying (L1) and (L2). The following verification essentially follows the ideas of Proposition 8 and Lemmas 9–10 in \cite{chickering2002optimal}, but we reprove it here for completeness.

\noindent\textbf{Verification of (A).}

Assume $\max_{M' \in \mathcal I(M)} S(M') \le S(M)$.
We prove $P_1(M;\p^\star)$ by contraposition. Suppose $P_1(M;\p^\star)$ fails. Then there exist a DAG $G \in M$, a node $Z_j$, and a non-descendant $Z_i$ of $Z_j$ in $G$ such that $Z_j \not\!\perp\!\!\!\perp_{\p^\star} Z_i \mid \MPa_j^G$. Because $Z_i$ is a non-descendant, adding the edge $Z_i \to Z_j$ to $G$ does not create a cycle. Moreover, since the conditional independence with respect to $\MPa_j^G$ is violated, this $Z_i$ cannot belong to $\MPa_j^G$ (if $Z_i \in \MPa_j^G$ the statement $Z_j \perp\!\!\!\perp_{\p^\star} Z_i \mid \MPa_j^G$ is trivially true). Thus we can form an MEC $M' \in \mathcal I(M)$ by inserting $Z_i \to Z_j$ into $G$. By (L1) this insertion strictly increases the score, so
$S(M') > S(M)$, contradicting $\max_{M' \in \mathcal I(M)} S(M') \le S(M)$. Hence $P_1(M;\p^\star)$ must hold.

\noindent\textbf{Verification of (B).}

Let $M$ be an MEC such that $P_1(M;\p^\star)$ holds, and let $M'\in\mathcal D(M)$ satisfy $S(M')\ge S(M)$. Pick $G\in M$ and let $G'\in M'$ be obtained from $G$ by deleting a single edge $Z_i\to Z_j$, so that the parent set of $Z_j$ changes from $\MPa_j^G$ to $\MPa_j^G\setminus\{Z_i\}$.

Since $P_1(M;\p^\star)$ holds, the law $\p^\star$ factorizes according to $G$:
\begin{equation}\label{eq: 1}
\p^\star(\MZ)
= \p^\star(Z_j \mid \MPa_j^G)\;\prod_{k\neq j} \p^\star(Z_k \mid \MPa_k^G).
\end{equation}

Now consider the reverse operation that inserts the edge $Z_i\to Z_j$ into $G'$. This insertion produces $G$.
On the high-probability event where (L1)--(L2) hold for this insertion, the alternative in (L1) would imply
$S(G,\hat{\MZ})>S(G',\hat{\MZ})$, contradicting $S(G',\hat{\MZ})\ge S(G,\hat{\MZ})$.
Hence the (L2) alternative must hold, which yields
\begin{equation}\label{eq: 2}
\p^\star(Z_j \mid Z_i, \MPa_j^{G}\setminus\{Z_i\})
= \p^\star(Z_j \mid \MPa_j^{G}\setminus\{Z_i\}).
\end{equation}

Define a set of local conditionals for $G'$ by keeping all other conditionals the same as in $G$ and by replacing the conditional at $Z_j$ with
\[
q^\star(Z_j \mid \MPa_j^G\setminus\{Z_i\})
:= \p^\star(Z_j \mid \MPa_j^G\setminus\{Z_i\}).
\]
Combining \eqref{eq: 1} and \eqref{eq: 2} we obtain, 
\[
\p^\star(\MZ)
= q^\star(Z_j \mid \MPa_j^G\setminus\{Z_i\})\;\prod_{k\neq j} \p^\star(Z_k \mid \MPa_k^G),
\]
which is exactly the factorization of $\p^\star$ with respect to $G'$.
Hence $\p^\star$ is also represented by $G'$, and therefore $P_1(M';\p^\star)$ holds (Theorem 6.2 in \cite{evans2025graphical}).

\noindent\textbf{Verification of (C).}

Assume $P_1(M;\p^\star)$ holds and $\max_{M'\in\mathcal D(M)} S(M') \le S(M)$. We will show $P_2(M;\p^\star)$.
Suppose, toward a contradiction, that $P_2(M;\p^\star)$ fails.
Let $G^\star$ be a DAG such that $\p^\star$ is DAG-perfect for $G^\star$, and let $M^\star$ be its MEC, so that
$\mathcal I(\p^\star)=\mathcal I(M^\star)$.
Since $P_2(M;\p^\star)$ fails, we have $M^\star\neq M$.
Since $P_1(M;\p^\star)$ holds, we have $\mathcal I(M)\subseteq \mathcal I(\p^\star)=\mathcal I(M^\star)$, and hence
$\mathcal I(M)\subsetneqq \mathcal I(M^\star)$.
Pick representatives $G\in M$ and $G^\star\in M^\star$.
By Theorem 4 in \cite{chickering2002optimal}, there is a finite sequence of single-edge operations (covered edge reversals and edge deletions) that transforms $G$ into $G^\star$, and along the entire sequence $\calI(G^\star)\supseteq\calI(G')$ holds for every intermediate DAG $G'$. In particular, there exists a first deletion in the sequence, say $G\to G_1$ obtained by removing $Z_i\to Z_j$, with $\calI(G^\star) \supseteq\calI( G_1)\supseteq\calI (G)$. Let $M_1$ denote the MEC of $G_1$.

Because $\calI(G^\star) \supseteq\calI( G_1)$ and $\p^\star$ is DAG-perfect for $G^\star$, every conditional independence encoded by $G_1$ holds in $\p^\star$.  Deleting $Z_i\!\to\!Z_j$ from $G$ yields $G_1$ with $\MPa_j^{G_1}=\MPa_j^{G}\setminus\{Z_i\}$ and $Z_i$ a non-descendant of $Z_j$ in $G_1$. Thus, by the local Markov property in $G_1$, we have in particular $Z_j \perp\!\!\!\perp Z_i\big| \MPa_j^{G}\setminus\{Z_i\}$, and so this independence holds in $\p^\star$, namely $Z_j \perp\!\!\!\perp_{\p^\star} Z_i \big| \MPa_j^{G}\setminus\{Z_i\}$.
By (L2) applied to the reverse insertion that adds $Z_i\to Z_j$ to $G_1$ (thereby recovering $G$), we have
$S(G,\hat{\MZ})<S(G_1,\hat{\MZ})$ for large $N$ with probability$\to1$.
Therefore, the deletion $G\to G_1$ strictly increases the score, and by score equivalence $S(M_1)> S(M)$ with probability$\to1$.
But $M_1\in \mathcal D(M)$, contradicting $\max_{M'\in\mathcal D(M)} S(M')\le S(M)$.
\begin{remark}
In \citet[Prop.~8; Lemmas~9–10]{chickering2002optimal}, the proof of GES correctness implicitly mixes local consistency with consistency. In this paper we follow the XGES reformulation \citep{nazaret2024extremely} and provide a new proof using only the two local-consistency conditions (L1)–(L2), thereby avoiding any appeal to consistency.
\end{remark}
\subsection{Implementation Details}\label{sec: Implementation Details}
\subsubsection{Assumptions regrading the penalty function}\label{sec: penalty assumptions}
For completeness, we spell out the shape assumptions on the penalty
$p_{\lambda_N,\tau_N}$ and tuning parameters $(\lambda_N,\tau_N)$ used in
Theorem~\ref{thm:estimation consistency}. For some $\lambda_N,\tau_N>0$,
$p_{\lambda_N,\tau_N}:\mathbb{R}\to[0,\infty)$ is a sparsity–inducing
symmetric penalty that is nondecreasing on $[0,\infty)$, nondifferentiable at
$0$, differentiable on $(0,\tau_N)$, and satisfies
\[
  p_{\lambda_N,\tau_N}(b)=0 \text{ if } b=0,\qquad
  p'_{\lambda_N,\tau_N}(b)\le C\frac{\lambda_N}{\tau_N} \text{ if } |b|\le\tau_N,
  \qquad
  p_{\lambda_N,\tau_N}(b)=\lambda_N \text{ if } |b|\ge\tau_N,
\]
for some constant $C<\infty$ independent of $N$.
\subsubsection{Detailed Initialization Algorithm}\label{sec: Detailed Ini}
Let $\mu: H \to \mathcal{X}_j$ denotes the known mean function of the observed-layer parametric family as described in \eqref{eq:observed exp fam}.
\begin{algorithm}[h!]
\caption{Spectral initialization}
\label{al: Ini}
\SetKwInOut{Input}{Input}
\SetKwInOut{Output}{Output}
\KwData{$\MX, K$, function $\tilde{g}=\mu\circ g$, truncation parameters $\epsilon, \delta$}. Algorithm \ref{al: Ini} summarizes this algorithm.

\begin{enumerate}
    \item Apply SVD to $\MX$ and write $\MX = \MU \Sigma \MV^{\top}$, where $\Sigma = \text{diag}(\sigma_i)$ and $\sigma_1 \ge \ldots \ge \sigma_J$.
    \item Let ${\MX}_{\tilde{K}} = \sum_{k=1}^{\tilde{K}} \sigma_k \bu_k \bv_k^{\top}$, where $\tilde{K} := \max\{K+1, \arg\max_k \{ \sigma_k \ge 1.01 \sqrt{N} \} \}$.
    \item Define $\hat{\MX}_{\tilde{K}}$ by truncating $\MX_{\tilde{K}}$ to the range of responses, at level $\epsilon$.
    \item Define $\hat{\ML}$ by letting $ \hat{l}(i,j) =  \tilde{g}^{-1}(\hat{y}_{\tilde{K}}(i,j))$.
    \item Let $\hat{\ML}_0$ be the centered version of $\hat{\ML}$, that is, $\hat{l}_0(i,j)= \hat{l}(i,j) - \frac{1}{N} \sum_{k=1}^N \hat{l}(k,j)$.
    \item Apply SVD to $\hat{\ML}_0$ and write its rank-$K$ approximation as $\hat{\ML}_0 \approx \hat{\MU} \hat{\Sigma} \hat{\MV}$.
    \item Let $\tilde{\MV}$ be the rotated version of $\hat{\MV}$ according to the Varimax criteria.
    \item Entrywise threshold $\tilde{\MV}$ at $\delta$ to induce sparsity, and flip the sign of each column so that all columns have positive mean.
   Let $\hat{\MQ}$ be the estimated sparsity pattern. 
    \item Estimate the centered $\MZ_0$ by $\hat{\MZ}_0 := \hat{\ML}_0 \tilde{\MV} (\tilde{\MV}^{\top} \tilde{\MV})^{-1}$, and estimate $\MZ$ by reading off the signs: $\hat{Z}(i,k) = \mathbbm{1}(Z_0(i,k) > 0).$ 
    \item Let $\hat{\MZ}_{\text{long}} := [\mathbf{1}, \hat{\MZ}]$.
    Estimate $\MB$ by $\hat{\MB} := C_g ( ( \hat{\MZ}^{\top}_{\text{long}} \hat{\MZ}_{\text{long}})^{-1} \hat{\MZ}_{\text{long}}^{\top} \hat{\ML}_0) \cdot \hat{\MG}$, where $\cdot$ is the element-wise product and $C_g$ is a positive constant.
    \item Define $\hat{\MR}=\MX-\hat{\MZ}_{\text{long}}\hat{\MX}^\top$ and estimate $\gamma_j$ by $\hat{\gamma}_j=\frac{1}{N}\sum_{i=1}^N\hat{\r}(i,j)^2$. 
\end{enumerate}
 \Output {$\hat{\p}, \ \hat{\MB},\ \hat{\bgamma},\ \hat{\MZ}$.}
\end{algorithm}

We explain more on the truncation details in Step 3 by considering specific response types. For Normal responses, the original sample space is $\mathbb{R}$ and the truncation (Steps 1-4 in Algorithm \ref{al: Ini}) may be omitted. For Binary responses, we set 
$$\hat{x}_{K}(i,j) = \begin{cases}\epsilon, &\text{ if } x_{K}(i,j) = 0, \\1-\epsilon, &\text{ if } x_{K}(i,j) = 1.\end{cases}$$
For Poisson responses, we set $$\hat{x}_{K}(i,j) = \begin{cases}\epsilon, &\text{ if } x_{K}(i,j) < \epsilon, \\x_{K}(i,j), &\text{ otherwise.}
\end{cases}$$In terms of implementing the method, we follow the suggestions of \cite{zhang2020note} with $\epsilon = 10^{-4}$.

\subsubsection{Detailed Gibbs--SAEM Algorithm}\label{sec: Detailed SAEM}

This subsection provides the full pseudocode for the penalized Gibbs--SAEM procedure described in Section~\ref{subsec: EM}. The algorithm alternates between a Gibbs-based stochastic E--step for updating the latent variables and a penalized SAEM M--step for updating the model parameters. Algorithm~\ref{al: PSAEM} summarizes the full procedure.
\begin{algorithm}[htbp]
\SetAlgoLined
\caption{Penalized Gibbs--SAEM Algorithm}\label{al: PSAEM}
\KwData{$\MX, K$, tuning parameters $\lambda_{N},\tau_{N}$, stepsizes $\{\theta_t\}_{t\ge1}$, number of Gibbs sweeps $C\ge1$.}

\textbf{Initialize} $\MZ^{[0]}$, $\MTheta^{[0]}=(\mathbf p^{[0]},\bbeta^{[0]},\bgamma^{[0]})$, $F_j^{[0]}\equiv0$ for $j\in[J]$; set $t\gets0$.

\While{$\|{\MTheta}^{[t+1]} - {\MTheta}^{[t]} \|$ is larger than a threshold}{

  $\MZ_{\mathrm{cur}}\gets \MZ^{[t]}$\; 

  \For{$r\in[C]$}{
    \For{$i\in[N]$}{
      \For{$k$ in a random permutation of $\{1,\dots,K\}$}{

        $\mathbf z^{(1)}\gets (1,\MZ_{\mathrm{cur},i,-k}),\quad
         \mathbf z^{(0)}\gets (0,\MZ_{\mathrm{cur},i,-k})$\;

        $\displaystyle \Delta\ell_{ik}\gets
        \log \mathbf p^{[t]}(\mathbf z^{(1)}) - \log \mathbf p^{[t]}(\mathbf z^{(0)})
        + \sum_{j=1}^J \log\frac{\PP(X_{ij}\mid \mathbf z^{(1)};\bbeta_j^{[t]},\gamma_j^{[t]})}
                               {\PP(X_{ij}\mid \mathbf z^{(0)};\bbeta_j^{[t]},\gamma_j^{[t]})}$\; 

        $\MZ_{\mathrm{cur},i,k}\sim\mathrm{Bernoulli}(\mathrm{expit}(\Delta\ell_{ik}))$\;
      }
    }
    $\MZ^{[t+1],r}\gets \MZ_{\mathrm{cur}}$\;
  }

  $\MZ^{[t+1]}\gets \MZ^{[t+1],C}$\;

  $\displaystyle \widehat{\mathbf p}^{[t+1]}(\mathbf z)
  \gets \frac{1}{C N}\sum_{r=1}^{C}\sum_{i=1}^{N}
  \mathbf 1\{\MZ^{[t+1],r}_{i,:}=\mathbf z\}$ for $\mathbf z\in\{0,1\}^K$\;

  $\mathbf p^{[t+1]}(\mathbf z)\gets
  (1-\theta_{t+1})\,\mathbf p^{[t]}(\mathbf z)
  + \theta_{t+1}\,\widehat{\mathbf p}^{[t+1]}(\mathbf z)$\;

  \For{$j\in[J]$}{
    $\displaystyle \widehat{F}_{j}^{[t+1]}(\bbeta_j,\gamma_j)
    \gets \frac{1}{C}\sum_{r=1}^{C}\sum_{i=1}^{N}
    \log \PP\!\big(X_{ij}\mid \MZ_i=\MZ^{[t+1],r}_i;\bbeta_j,\gamma_j\big)$\;

    $\displaystyle F_j^{[t+1]}(\bbeta_j,\gamma_j)
    \gets (1-\theta_{t+1})\,F_j^{[t]}(\bbeta_j,\gamma_j)
       + \theta_{t+1}\,\widehat{F}_{j}^{[t+1]}(\bbeta_j,\gamma_j)$\;

    $(\bbeta_j^{[t+1]},\gamma_j^{[t+1]})
    \gets \arg\max_{\bbeta_j,\gamma_j}
    \big\{F_j^{[t+1]}(\bbeta_j,\gamma_j)-p_{\lambda_N,\tau_N}(\bbeta_j)\big\}$\;
  }

  $t\gets t+1$\;
}

Let $\widehat{\mathbf p}\gets \mathbf p^{[T]}$ and $(\widehat\bbeta,\widehat\bgamma)\gets(\bbeta^{[T]},\bgamma^{[T]})$ at convergence\;
Estimate the measurement graph $\MQ$ from the sparsity pattern of \eqref{eq:graphical matrix estimator}\;

\textbf{Output:} $\hat{\MTheta}=(\widehat{\mathbf p},\widehat\bbeta,\widehat\bgamma)$ and $\MQ$.
\end{algorithm}

\subsubsection{Simulation Setup Details}\label{sec:sim-details}

In all simulations, we set $\MQ$ and $\p$ as follows. The measurement matrix $\MQ$ takes the form
\[
\MQ \;=\; \begin{pmatrix}
\MQ'\\[2pt]
\mathbf{I}_K \\[2pt]
\mathbf{I}_K 
\end{pmatrix},
\qquad 
\MQ_1'=\begin{pmatrix}
 1 & 1 &  & 0 \\
1 & \ddots & \ddots &  \\
 & \ddots & \ddots & 1 \\
0 &  & 1 & 1
\end{pmatrix}_{K \times K},
\qquad
\MQ_2'=\begin{pmatrix}
1 & 1 & 1 &  & 0 \\
1 & \ddots & \ddots & \ddots &  \\
1 & \ddots & \ddots & \ddots & 1 \\
 & \ddots & \ddots & \ddots & 1 \\
0 &  & 1 & 1 & 1
\end{pmatrix}_{K \times K}.
\]
We consider two banded choices for the submatrix $\MQ'$: $\MQ_1'$ and $\MQ_2'$, and denote the corresponding full matrices by $\MQ_1$, $\MQ_2$ respectively. Both choices satisfy the identifiability conditions in Theorem~\ref{thm1}.

Given a DAG on the latent variables, we define the distribution $\p$ so as to yield balanced conditional probabilities and avoid degeneracy. If a child has a single parent, then when the parent equals $1$ the Bernoulli parameter of the child is drawn uniformly from $[0.3,0.35]$ or from $[0.65,0.7]$ with equal probability, and the parameter for parent $0$ is set to be the complement. If a child has two parents, we consider the four parent configurations $(0,0)$, $(0,1)$, $(1,0)$, and $(1,1)$. For configuration $(0,0)$ we draw the Bernoulli parameter independently from $[0.2,0.25]$, and for $(1,1)$ we draw it from $[0.6,0.65]$. For the mixed configurations $(0,1)$ and $(1,0)$, we draw one parameter from $[0.35,0.4]$ and the other from $[0.77,0.82]$.

We consider three parametric families for the observed layer: Gaussian, Poisson, and Bernoulli for continuous, count, and binary data, respectively. This allows us to assess the robustness of our method under both continuous and discrete measurement models. Because these families live on different scales, we specify different values for the regression parameters $\beta$:
\begin{align*}
    \beta_{j,0} &= \begin{cases}
        -1, & \text{Gaussian};\\
        1, & \text{Poisson/Bernoulli};
    \end{cases} \\
    \beta_{j,k} &=\begin{cases}
        \dfrac{3}{\sum_{k'=1}^K q_{j,k'}} \mathbf{1}(q_{j,k} = 1), & \forall j \in [J],~ k \in [K], \quad \text{Gaussian};\\[4pt]
        \dfrac{2}{\sum_{k'=1}^K q_{j,k'}} \mathbf{1}(q_{j,k} = 1), & \forall j \in [J],~ k \in [K], \quad \text{Poisson/Bernoulli}.
    \end{cases} 
\end{align*}
The variance parameter is fixed to be $\gamma_j = \sigma_j^2 = 1$ for all $j$.

\subsubsection{Choice of $f(N)$}\label{sec: choice of fn}
Recalling our algorithmic procedure, we obtain a $f(N)\times K$ matrix by sampling from $\hat{\p}$.
Although we have specified the admissible range for $f(N)$, the exact choice within this range remains to be determined. In our simulations, balancing computational efficiency and performance, we record the recovered DAGs when $f(N) = N, 2N, 3N$. In these three cases, the matrices on which GES is applied are denoted as $\hat{\MZ}_1, \hat{\MZ}_2, \hat{\MZ}_3$. We record the averaged SHD (Structural Hamming Distance) for different scenarios in Table~\ref{tab:SHD-MT9}   

\begin{table}[h!]
\centering
\begin{tabular}{lllccccccccc}
\toprule
& & & \multicolumn{3}{c}{Bernoulli} & \multicolumn{3}{c}{Poisson} & \multicolumn{3}{c}{Lognormal} \\
\cmidrule(lr){4-6}\cmidrule(lr){7-9}\cmidrule(lr){10-12}
Model & $\MQ$ & $\hat\MZ$ & 3000 & 5000 & 7000 & 3000 & 5000 & 7000 & 3000 & 5000 & 7000 \\
\midrule
\multirow{6}{*}{Chain-10}
  & \multirow{3}{*}{$\MQ_1$} & $\hat\MZ_1$ & 5.55  & 3.294 & 2.938 & 2.248 & 1.362 & 0.638 & 0.412 & 0.254 & 0.122 \\
  &                      & $\hat\MZ_2$ & 5.966 & 4.92  & 4.88  & 2.284 & 1.986 & 1.206 & 0.992 & 0.49  & 0.33  \\
  &                      & $\hat\MZ_3$ & 7.852 & 6.91  & 7.344 & 3.518 & 3.178 & 2.542 & 1.82  & 1.246 & 0.95  \\
  & \multirow{3}{*}{$\MQ_2$} & $\hat\MZ_1$ &5.664  &3.258  & 3.042   &2.174&0.704&0.488  & 0.406   &0.208    &0.146  \\
  &                      & $\hat\MZ_2$ & 6.042 &4.784  & 5.314 & 2.24  & 1.306 & 1.048 & 0.894 & 0.562&0.328    \\
  &                      & $\hat\MZ_3$ & 8.058 & 7.016&7.54 & 3.534& 2.53 & 2.272 & 1.75&1.286&1.026   \\
\addlinespace
\multirow{6}{*}{Tree-10}
  & \multirow{3}{*}{$\MQ_1$} & $\hat\MZ_1$ & 4.372 & 2.308 & 1.308 & 1.85  & 1.692 & 1.36  & 0.94  & 0.51  & 0.308 \\
  &                      & $\hat\MZ_2$ & 5.024 & 3.674 & 3.068 & 3.09  & 2.656 & 2.66  & 1.528 & 0.788 & 0.344 \\
  &                      & $\hat\MZ_3$ & 7.07  & 5.566 & 5.106 & 4.622 & 4.518 & 3.982 & 2.388 & 1.32  & 1.032 \\
  & \multirow{3}{*}{$\MQ_2$} & $\hat\MZ_1$ & 4.3    & 1.936    & 1.26  &2.676 &1.556& 1.146 & 1.14 & 0.618    & 0.346  \\
  &                      & $\hat\MZ_2$ & 5.152   &3.346    &2.936 & 3.62  & 2.41 & 2.218 & 1.582  & 0.878 & 0.496\\
  &                      & $\hat\MZ_3$ & 7.26 &5.368  & 4.964  & 5.22& 4.136& 3.706 & 2.66  & 1.578 & 0.882    \\
\addlinespace
\multirow{6}{*}{Model-7}
  & \multirow{3}{*}{$\MQ_1$} & $\hat\MZ_1$ & 7.692 & 6.334 & 5.798 & 5.838 & 5.422 & 4.892 & 0.196 & 0.042 & 0     \\
  &                      & $\hat\MZ_2$ & 7.352 & 6.274 & 6.132 & 5.278 & 4.754 & 4.786 & 0.06  & 0.036 & 0.014 \\
  &                      & $\hat\MZ_3$ & 7.806 & 6.976 & 6.994 & 5.18  & 5.362 & 6.288 & 0.144 & 0.1   & 0.102 \\
  & \multirow{3}{*}{$\MQ_2$} & $\hat\MZ_1$ & 7.554  & 6.304   &5.68  & 5.848   &5.526 & 5.082&0.262  & 0.054  & 0.004 \\
  &                      & $\hat\MZ_2$ & 7.312 & 6.238 &6.022 &5.39  & 5.048& 5.124& 0.1  & 0.036 & 0.036  \\
  &                      & $\hat\MZ_3$ & 7.714 &6.83&6.87    &5.148  & 5.452& 6.552 & 0.188 & 0.098    & 0.084 \\
\addlinespace
\multirow{6}{*}{Model-8}
  & \multirow{3}{*}{$\MQ_1$} & $\hat\MZ_1$ & 4.336 & 2.682 & 2.19  & 2.106 & 1.916 & 1.878 & 0.132 & 0.048 & 0     \\
  &                      & $\hat\MZ_2$ & 5.012 & 3.498 & 3.336 & 2.356 & 2.238 & 2.742 & 0.172 & 0.084 & 0.058 \\
  &                      & $\hat\MZ_3$ & 6.354 & 5.108 & 5.532 & 3.17  & 3.292 & 4.336 & 0.416 & 0.248 & 0.22  \\
  & \multirow{3}{*}{$\MQ_2$} & $\hat\MZ_1$ & 4.342 &2.72   & 2.374 &2.264 & 1.9 &1.782  & 0.202  & 0.052& 0.002  \\
  &                      & $\hat\MZ_2$ & 4.992  &3.416   & 3.29& 2.516 &2.388& 2.588 & 0.15   & 0.07  & 0.036  \\
  &                      & $\hat\MZ_3$ & 6.264& 4.952 & 5.454 & 3.336& 3.808 & 4.17 & 0.498 & 0.302 & 0.25    \\
\addlinespace
\multirow{6}{*}{Model-13}
  & \multirow{3}{*}{$\MQ_1$} & $\hat\MZ_1$ & 22.37  & 16.454 & 14.062 & 24.65 & 16.472 & 14.162 & 3.206 & 1.646 & 1.008 \\
  &                      & $\hat\MZ_2$ & 23.274 & 17.152 & 15.544 & 24 & 16.106 & 15.258& 3.138 & 1.788 & 1.17  \\
  &                      & $\hat\MZ_3$ & 25.134 & 19.594 & 18.24  & 25.16 & 18.706 & 18.732 & 3.89  & 2.554 & 1.728 \\
  & \multirow{3}{*}{$\MQ_2$} & $\hat\MZ_1$ & 22.252  & 16.29  & 14.606    & 25.134  &15.626 & 12.934 &3.032   & 1.872&0.994 \\
  &                      & $\hat\MZ_2$ & 23.348 &16.816  & 15.622  & 24.686  &15.162& 14.258  & 2.852  & 2.082 & 1.158\\
  &                      & $\hat\MZ_3$ & 25.262 &19.334  & 18.68  &25.612& 17.592 & 17.46 &3.624 & 2.696 & 1.68 \\
\bottomrule
\end{tabular}
\caption{SHD.}
\label{tab:SHD-MT9}
\end{table}

One can find that although $\hat\MZ_1$ achieves the best results in most cases, $\hat\MZ_2$ and $\hat\MZ_3$ can occasionally outperform it. Therefore, we regard $\frac{f(N)}{N}$ as a tuning parameter, and the choice of its optimal value remains an open question. At this stage, we recommend using $\hat\MZ_1$.
\subsubsection{Q-matrix for the TIMSS data}\label{subsubsec:TIMSS-Q}

This subsection provides the $\MQ$-matrix used in Section~\ref{sec: TIMSS}; see Table~\ref{tab:TIMSS}.
\begin{table}[htbp]
\centering
\caption{$\MQ$-matrix for the TIMSS 2019 math assessment booklet 14.}
\begin{tabular}{c|ccccccc}
\hline
Item ID & Number & Algebra & Geometry & Data and Prob. & Knowing & Applying & Reasoning \\
\hline
1  & 1 & 0 & 0 & 0 & 1 & 0 & 0 \\
2  & 1 & 0 & 0 & 0 & 0 & 1 & 0 \\
3  & 1 & 0 & 0 & 0 & 0 & 1 & 0 \\
4  & 1 & 0 & 0 & 0 & 1 & 0 & 0 \\
5  & 0 & 1 & 0 & 0 & 0 & 1 & 0 \\
6  & 0 & 1 & 0 & 0 & 1 & 0 & 0 \\
7  & 0 & 1 & 0 & 0 & 0 & 0 & 1 \\
8  & 0 & 1 & 0 & 0 & 0 & 1 & 0 \\
9  & 0 & 0 & 1 & 0 & 0 & 0 & 1 \\
10 & 0 & 0 & 1 & 0 & 0 & 1 & 0 \\
11 & 0 & 0 & 1 & 0 & 0 & 1 & 0 \\
12 & 0 & 0 & 0 & 1 & 0 & 1 & 0 \\
13 & 0 & 0 & 0 & 1 & 0 & 1 & 0 \\
14 & 1 & 0 & 0 & 0 & 1 & 0 & 0 \\
15 & 1 & 0 & 0 & 0 & 1 & 0 & 0 \\
16 & 1 & 0 & 0 & 0 & 0 & 0 & 1 \\
17 & 1 & 0 & 0 & 0 & 0 & 1 & 0 \\
18 & 0 & 1 & 0 & 0 & 1 & 0 & 0 \\
19 & 0 & 1 & 0 & 0 & 1 & 0 & 0 \\
20 & 0 & 1 & 0 & 0 & 0 & 1 & 0 \\
21 & 0 & 1 & 0 & 0 & 0 & 1 & 0 \\
22 & 0 & 1 & 0 & 0 & 0 & 1 & 0 \\
23 & 0 & 0 & 1 & 0 & 1 & 0 & 0 \\
24 & 0 & 0 & 1 & 0 & 0 & 0 & 1 \\
25 & 0 & 0 & 1 & 0 & 0 & 0 & 1 \\
26 & 0 & 0 & 1 & 0 & 0 & 1 & 0 \\
27 & 0 & 0 & 0 & 1 & 0 & 1 & 0 \\
28 & 0 & 0 & 0 & 1 & 0 & 1 & 0 \\
29 & 0 & 0 & 0 & 1 & 0 & 0 & 1 \\
\hline
\end{tabular}\label{tab:TIMSS}
\end{table}
\newpage
\subsubsection{Seesaw Image Generator and Preprocessing Details}\label{app:seesaw_generation}

For each sample $i\in[10000]$, we draw $(Z_{1i},Z_{2i})$ independently as $\mathrm{Bernoulli}(0.5)$.
Then $Z_{3i}$ is generated from a stochastic seesaw-response rule
\[
Z_{3i}\sim \mathrm{Bernoulli}\bigl(p_3(Z_{1i},Z_{2i})\bigr),
\qquad
p_3(1,1)=0.8,
\qquad
p_3(z_1,z_2)=0.2\ \text{for }(z_1,z_2)\neq(1,1).
\]
The occluded-ball visibility $Z_{4i}$ is generated conditionally on $Z_{3i}$ as
\[
Z_{4i}\sim \mathrm{Bernoulli}\bigl(p_4(Z_{3i})\bigr),
\qquad
p_4(1)=0.99,
\qquad
p_4(0)=0.
\]

Given a latent configuration $(z_1,z_2,z_3,z_4)$, the generator renders a $256\times256$ RGB scene and converts it to grayscale.
The scene consists of a fixed pivot and a rigid plank (the seesaw) that rotates by an angle
\[
\alpha(z_3)=
\begin{cases}
-\alpha_0, & z_3=1 \ (\text{left side up}),\\
+\alpha_0, & z_3=0 \ (\text{left side down}),
\end{cases}
\qquad\text{with }\alpha_0=25^\circ.
\]
A left ball is placed on top of the plank near its left endpoint.
The two tray balls ($z_1$ and $z_2$) lie on two well-separated slots of a horizontal forked tray whose position is fixed in the image, so the tray does not rotate with the seesaw.

Across images, we introduce mild nuisance variation through small i.i.d.\ positional jitter in the seesaw subsystem.
In particular, for each rendered scene we perturb the pivot and plank center by independent uniform offsets in $[-\delta,\delta]$ along each axis, with $\delta=0.001$ in normalized coordinates, and we apply an additional small shared offset to the seesaw-side balls. This prevents the images from being identical templates within the same latent state and makes the learning task more challenging, while preserving the intended semantics.

The fourth ball is positioned at the left ball's canonical ``down'' location (the $z_3=0$ geometry) and is drawn before the left ball so that, when $z_3=0$, the left ball occludes it.
When $z_3=1$, the left ball moves upward and the fourth ball may become visible if $z_4=1$.

Let $X^{\mathrm{original}}_i\in\{0,1,\ldots,255\}^{256\times256}$ denote the grayscale image for sample $i$.
We then construct an inverted binary mask at the original resolution by thresholding at $80$:
pixels with grayscale value greater than $80$ are set to $1$, and dark pixels (the balls) are set to $0$.
We resize this binary mask to $96\times96$ using nearest-neighbor interpolation to preserve sharp object boundaries, and denote the result by $X^{\mathrm{mask}}_i\in\{0,1\}^{96\times96}$.
Under this convention, ball pixels are coded as $0$ and non-ball pixels are coded as $1$.

Finally, we produce the $16\times16$ pooled representation $X_i^{\mathrm{pooled}}\in\{0,1\}^{16\times16}$ by applying min-pooling over non-overlapping blocks of the $96\times96$ mask, yielding $256$ binary features per image.
Because the mask is inverted, a pooled entry equals $0$ if at least one ball pixel appears in that spatial cell, and equals $1$ otherwise.

Therefore, we obtain
\[
\MX_{\mathrm{original}}\in\{0,\ldots,255\}^{10000\times 256^2},\quad
\MX_{\mathrm{mask}}\in\{0,1\}^{10000\times 96^2},\]
\[\MX_{\mathrm{pooled}}\in\{0,1\}^{10000\times 256},\quad
\MZ\in\{0,1\}^{10000\times 4},
\]
where each row corresponds to one rendered scene and its associated latent state.
Example images are shown below.
In each figure, the single row displays (from left to right) the original image, the balls-only mask, and the pooled representation.
We fit our model using the pooled data $\MX_{\mathrm{pooled}}$.
\begin{figure}[htbp]
\begin{minipage}{0.99\linewidth}
\vspace{3pt}
\centerline{\includegraphics[height=5cm]{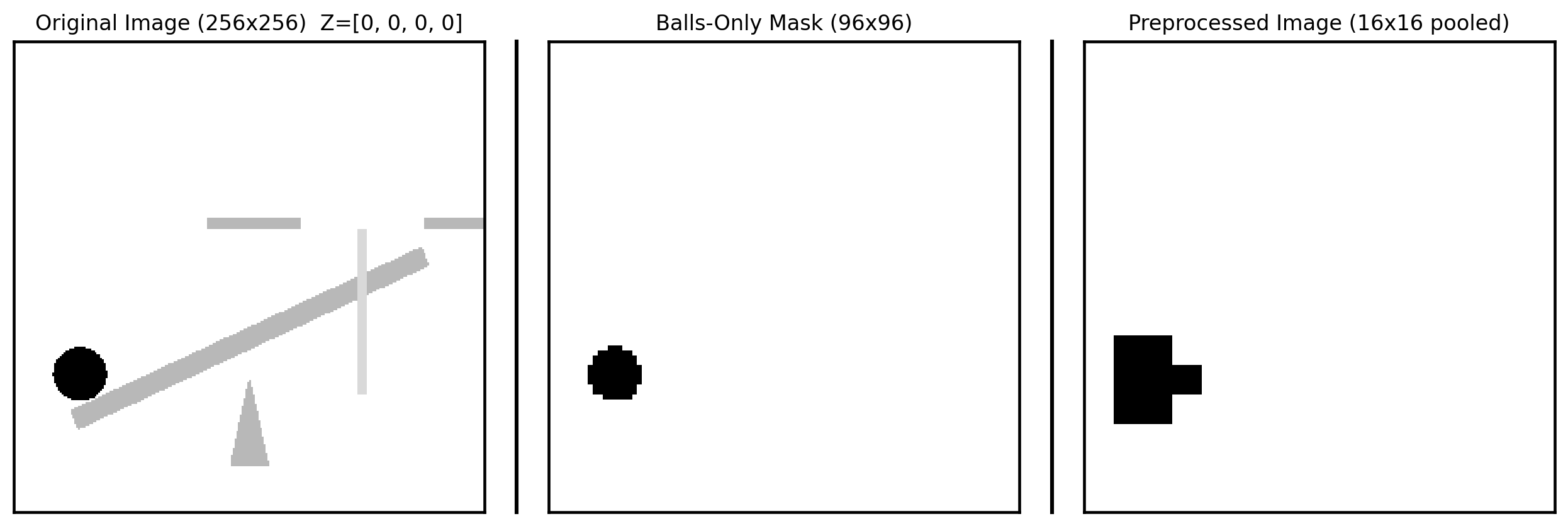}}
\end{minipage}
\end{figure}
\begin{figure}[htbp]
\begin{minipage}{0.99\linewidth}
\vspace{3pt}
\centerline{\includegraphics[height=5cm]{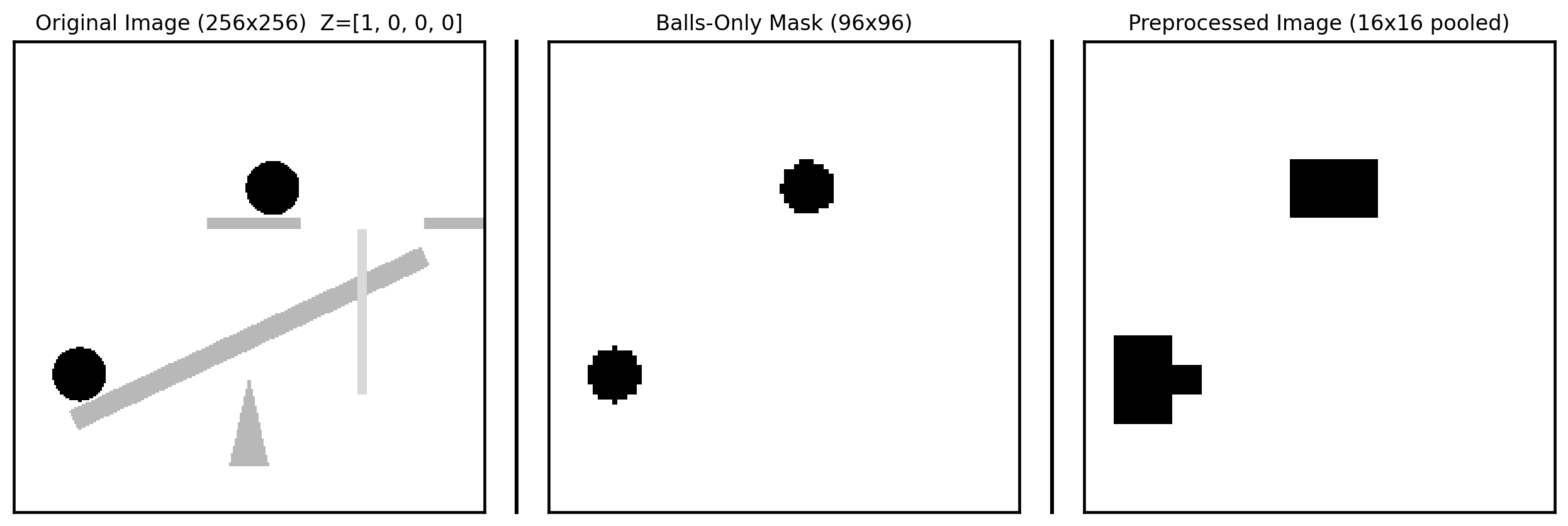}}
\end{minipage}
\end{figure}
\begin{figure}[htbp]
\begin{minipage}{0.99\linewidth}
\vspace{3pt}
\centerline{\includegraphics[height=5cm]{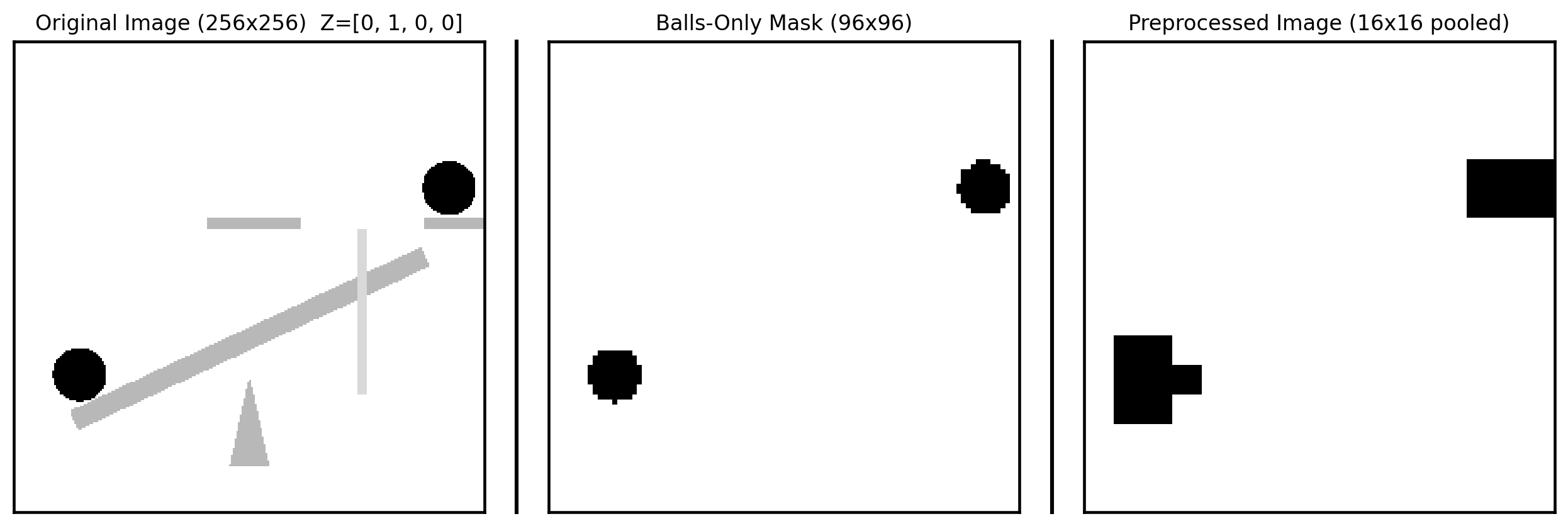}}
\end{minipage}
\end{figure}
\begin{figure}[htbp]
\begin{minipage}{0.99\linewidth}
\vspace{3pt}
\centerline{\includegraphics[height=5cm]{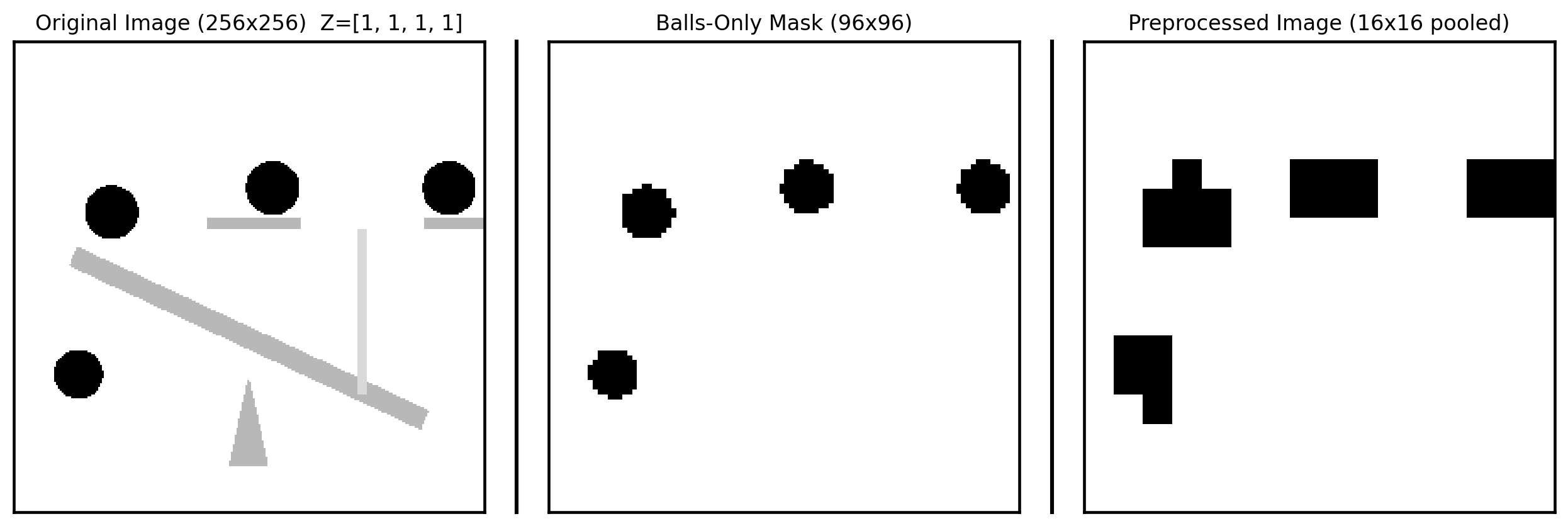}}
\end{minipage}
\end{figure}
\begin{figure}[htbp]
\begin{minipage}{0.99\linewidth}
\vspace{3pt}
\centerline{\includegraphics[height=5cm]{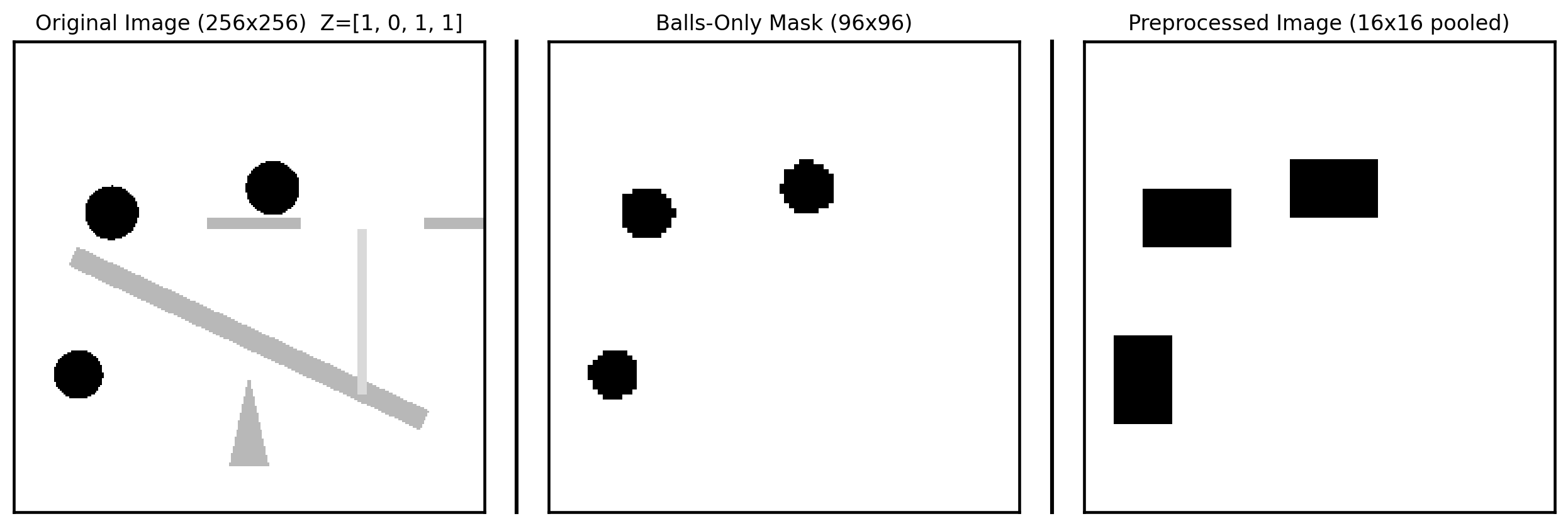}}
\end{minipage}
\end{figure}
\begin{figure}[htbp]
\begin{minipage}{0.99\linewidth}
\vspace{3pt}
\centerline{\includegraphics[height=5cm]{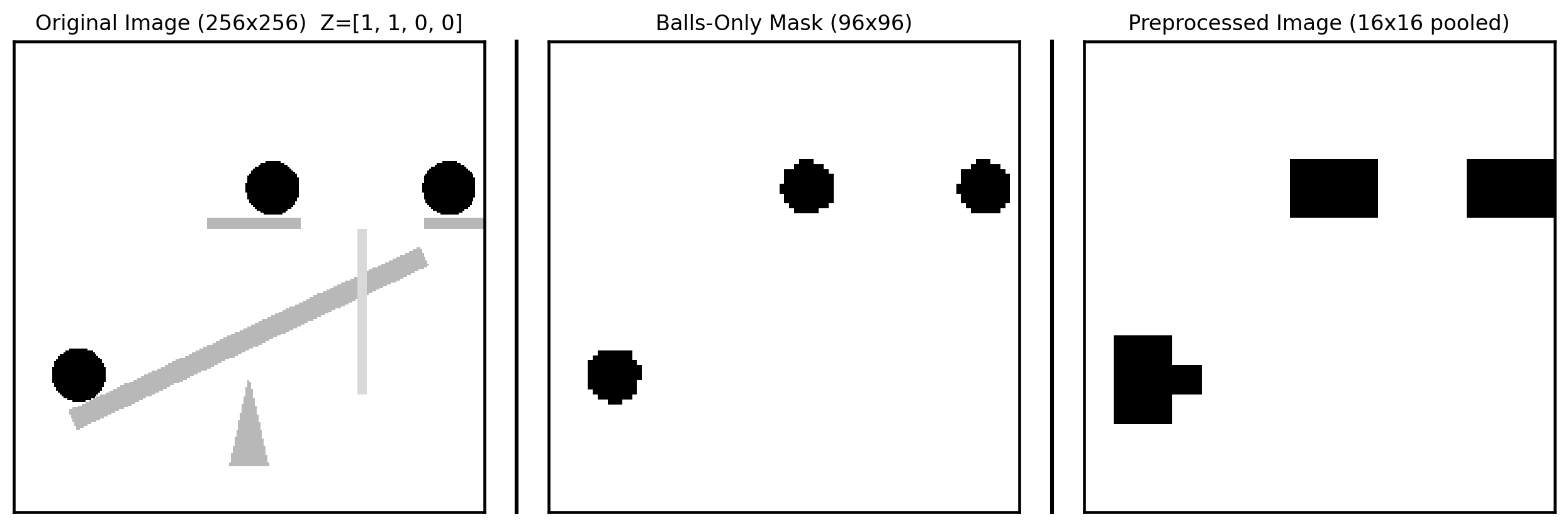}}
\end{minipage}
\end{figure}
\begin{figure}[htbp]
\begin{minipage}{0.99\linewidth}
\vspace{3pt}
\centerline{\includegraphics[height=5cm]{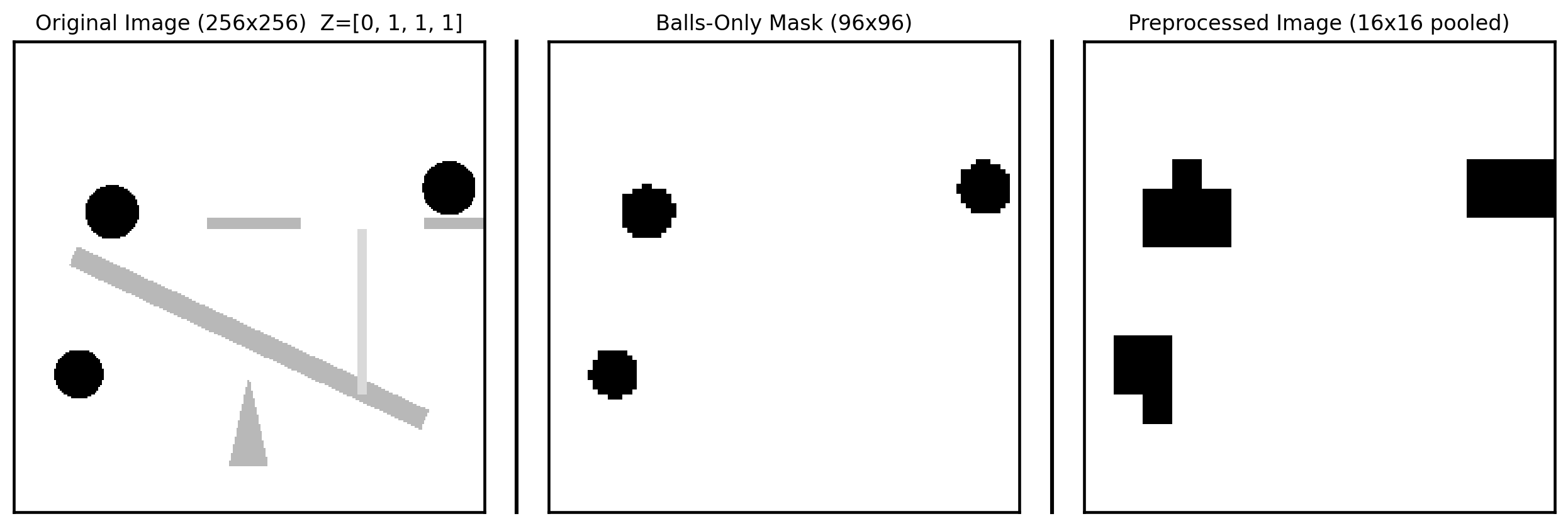}}
\end{minipage}
\end{figure}

\subsection{Connections with Existing Studies}\label{sec: Connections with Existing Studies}
We emphasize that our statistical specification of causal mechanisms is deliberately broad and flexible. For the causal structure among latent variables, we allow for arbitrarily complex dependencies among the latent factors. The only structural limitation is that the latent variables are taken to be discrete. Far from being restrictive, this choice has proven especially fruitful. On the theoretical side, assuming discrete latent variables allows us to employ powerful identifiability results from mixture and latent class models, thereby greatly facilitating more rigorous analysis \citep{teicher1967identifiability,yakowitz1968on,allman2008Identifiability}. On the modeling side, discrete latent hierarchies have formed the basis of influential architectures in machine learning, such as deep Boltzmann machines (DBMs) \citep{salakhutdinov2009deep}, deep belief networks (DBNs) \citep{hinton2006a}. DBMs and DBNs in particular were originally designed with multiple binary latent layers 
\citep{salakhutdinov2015learning}. Moreover, a collection of $K$ binary latent variables yields $2^K$ possible latent configurations, so that a relatively small number of latent factors can encode a combinatorially large family of data-generating regimes. This corresponds to a distributed representation in the usual sense of representation learning: each regime is encoded by a pattern of activations across multiple latent units, rather than a single categorical latent with $2^K$ levels. These examples show that discrete latents are sufficiently rich to capture complex data distributions while yielding parsimonious representations and tractable identifiability analysis. 

For the links between latent and observable variables, we adopt a specification in the spirit of all-effect general-response CDMs (GR-CDMs) \citep{liu2025exploratory}. Broadly speaking, cognitive diagnosis models (CDMs) are latent variable models in which each subject possesses a vector of binary latent attributes indicating mastery versus non-mastery on a collection of skills, and each item is designed to depend only on a specified subset of these attributes, typically encoded by a binary design matrix. Our CDM-style measurement layer is motivated by the remarkable success of CDMs in educational measurement, where they have proven to be powerful tools for modeling multidimensional discrete skills with both mature identifiability theory \citep{lee2024new} and rich representational capacity. In the binary-response case, our parameterization naturally subsumes well-known models such as the additive CDM (ACDM) and the generalized DINA (G-DINA) \citep{delaTorre2011the}, while for continuous responses it also includes extensions such as the continuous DINA (cDINA) model under positive outcomes \citep{minchen2017a}. In this way, the same decomposition provides a unified framework that accommodates polytomous, continuous, and mixed responses, while allowing higher–order interactions when supported by data. 

\citet{prashant2024differentiable} investigate causal discovery in hierarchical latent-variable models whose observed and latent variables are all modeled as continuous.
While the functional relationships among variables are quite flexible, their identifiability theory imposes a strong structural restriction on the latent DAG: latent variables are partitioned into hierarchical layers and edges are permitted only across layers. Consequently, the recoverable graphs are essentially concatenations of bipartite graphs between successive layers. In addition, although their identifiability result is also obtained from observational data alone, it requires a stronger measurement assumption than ours: each latent variable must possess at least two pure children. By contrast, our framework allows arbitrary latent DAGs and provides a unified treatment of both continuous and discrete observed variables, while our subset condition is strictly weaker than such pure-child structure.

Recent work by \citet{dong2026score} also studies a score-based greedy search procedure for partially observed causal models. Their theory is developed for a linear-Gaussian latent-variable SEM, where the greedy search compares maximized scores that depend on the observed covariance matrix \(\Sigma_X\). In our framework, the latent variables are discrete and the likelihood of $X$ is not determined by \(\Sigma_X\) alone. Therefore, the covariance-based Gaussian scoring theory of \citet{dong2026score} is not applicable to our setting. 

\citet{kivva2021learning} consider a general discrete-latent setting and reduce recovery to a mixture-oracle problem. But as already discussed in Section~\ref{sec:sim}, this generality comes at the price of not performing well in all settings. In terms of identifiability, although both our work and \citet{kivva2021learning} impose subset condition, the mechanisms are fundamentally different. \citet{kivva2021learning} obtain identifiability through access to a mixture oracle, hence they have to assume that the mixture model over $X_S$ is identifiable for \emph{every} subset $S\subseteq \mathcal X$. This assumption is typically violated in discrete-response settings, which we also include here. By contrast, we establish identifiability directly from the observed joint distribution in a completely different way, without requiring any oracle knowledge.

We further compare our framework with several works that also establish identifiability from observational data alone. Both \citet{moran2022identifiable} and \citet{kivva2022identifiability} fall into this category, but they differ from ours in several fundamental ways. First, both works focus on continuous latent variables under Gaussian assumptions, whereas our setting targets discrete latent variables. In particular, \citet{moran2022identifiable} also require the observational variables to be Gaussian and impose anchor features, which are analogous to pure children and are strictly stronger than our subset condition. On the other hand, \citet{kivva2022identifiability} assume a well-posed additive-noise observation model together with a piecewise-affine decoder and a Gaussian-mixture latent structure. Their disentanglement claims rely crucially on the Gaussian-mixture covariance structure across mixture components, and identifiability of the decoder further requires an injectivity condition. Consequently, both the assumptions and the proof techniques in these works are fundamentally different from those needed in our discrete-latent framework.

\end{document}